\documentclass{elsarticle}

\usepackage{lineno,hyperref}
\modulolinenumbers[5]

\usepackage{amsmath}
\usepackage{amssymb}
\usepackage{amsthm}

\usepackage{xspace}
\usepackage{misc}
\usepackage{enumerate}

\usepackage{boxedminipage}
\usepackage{tikz}
\usetikzlibrary{arrows}

\usepackage{url}

\theoremstyle{plain}
\newtheorem{theorem}{Theorem}
\newtheorem{lemma}[theorem]{Lemma}

\newtheorem{proposition}[theorem]{Proposition}
\theoremstyle{definition}
\newtheorem{definition}[theorem]{Definition}
\newtheorem{example}[theorem]{Example}
\theoremstyle{remark}
\newtheorem*{remark}{Remark}

\begin{document}

\begin{frontmatter}
 
\title{A Complete Classification of the Complexity and Rewritability of Ontology-Mediated Queries based on the Description Logic \EL}

\author{Carsten Lutz and Leif Sabellek}
\address{Department of Computer Science, University of Bremen, Germany}

\begin{abstract}
  We provide an ultimately fine-grained analysis of the data complexity and
  rewritability of ontology-mediated queries (OMQs) based on an \EL
  ontology and a conjunctive query~(CQ).
  Our main results are that every such OMQ is in~$\AC^0$, $\NL$-complete, or $\PTime$-complete and that containment in $\NL$
  coincides with rewritability into linear Datalog (whereas
  containment in $\AC^0$ coincides with rewritability into first-order
  logic). We establish natural characterizations of the three cases in terms
  of bounded depth and (un)bounded pathwidth, and
  show that every of the associated meta problems such as deciding wether
  a given OMQ is rewritable into linear Datalog is \ExpTime-complete. We
  also give a way to
  construct linear Datalog rewritings when they exist and prove that
  there is no constant bound on the arity of IDB relations in linear
  Datalog rewritings.
\end{abstract}


\end{frontmatter}


\section{Introduction}

An important application of ontologies is to enrich data with a semantics and with domain knowledge while also providing additional vocabulary for query formulation \cite{DBLP:conf/rweb/CalvaneseGLLPRR09,DBLP:conf/rweb/KontchakovRZ13,BienvenuCLW14,DBLP:conf/rweb/BienvenuO15}. In this context, it makes sense to view the combination of a database query and an ontology as a compound query, commonly referred to as an \emph{ontology-mediated query (OMQ)}.  An \emph{OMQ language} $(\Lmc,\Qmc)$ is then constituted by an ontology language and a query language \Qmc.  Prominent choices for \Lmc include many description logics (DLs) such as $\EL$, Horn-$\mathcal{SHIQ}$, and \ALC~\cite{DL-Textbook} while the most common choices for \Qmc are conjunctive queries (CQs), unions thereof (UCQs), and the simple atomic queries (AQ) which are of the form $A(x)$.  Substantial research efforts have been invested into understanding the properties of the resulting OMQ languages, with two important topics being
\begin{enumerate}

\item the \emph{data
    complexity} of OMQ evaluation \cite{hustadt-2005,LutzKri07,Rosati07,DBLP:journals/ai/CalvaneseGLLR13, BienvenuCLW14}, where data complexity means that only the data is
  considered the input while the OMQ is fixed, and 

\item the \emph{rewritability} of OMQs into more standard database query languages such as SQL (which in this context is often equated with first-order logic) and Datalog \cite{DBLP:conf/aaai/EiterOSTX12,BieLuWo-IJCAI13,BienvenuCLW14,kaminski14,DBLP:conf/ijcai/AhmetajOS16,ICDT17-JournalVersion}.  

\end{enumerate}
\PTime data complexity is often considered a necessary condition for efficient query evaluation in practice. Questions about rewritability are also motivated by practical concerns: Since most database systems are unaware of ontologies, rewriting OMQs into standard database query languages provides an important avenue for implementing OMQ execution in practical applications~\cite{DBLP:conf/rweb/CalvaneseGLLPRR09,HLSW-IJCAI15,perezurbina10tractable,DBLP:journals/ws/TrivelaSCS15}.  Both subjects are thoroughly intertwined since rewritability into first-order logic (FO) is closely related to $\AC^0$ data complexity while rewritability into Datalog is closely related to $\PTime$ data complexity. We remark that FO-rewritability of an OMQ implies rewritability into a UCQ and thus into Datalog~\cite{BienvenuCLW14}. From now one, when speaking about complexity we always mean data complexity.

Regarding compexity and rewritability, modern DLs can roughly be divided into two families: `expressive DLs' such as \ALC and $\mathcal{SHIQ}$ that result in OMQ languages with $\coNP$ complexity and where rewritability is guaranteed neither into FO nor into Datalog \cite{BienvenuCLW14,DBLP:journals/ws/TrivelaSCS15,ICDT17-JournalVersion}, and `Horn DLs' such as $\EL$ and Horn-$\mathcal{SHIQ}$ which typically have $\PTime$ complexity and where rewritability into (monadic) Datalog is guaranteed, but FO-rewritability is not \cite{BieLuWo-IJCAI13,HLSW-IJCAI15,ijcai16}. 
In practical applications, however, ontology engineers often need to use language features that are only available in expressive DLs, but they typically do so in a way such that one may hope for hardness to be avoided by the concrete ontologies that are being designed. 

Initiated in \cite{LutzWolter12,BienvenuCLW14}, this has led to studies of data complexity and rewritability that are much more fine-grained than the analysis of entire ontology languages, see also \cite{DBLP:conf/dlog/ZakharyaschevKG18,DBLP:conf/ijcai/LutzSW15}. The ultimate aim is to understand, for relevant OMQ languages $(\Lmc,\Qmc)$, the exact complexity and rewritability status of every OMQ from $(\Lmc,\Qmc)$. For expressive DLs, this turns out to be closely related to the complexity classification of constraint satisfaction problems (CSPs) with a fixed template \cite{FederVardi}. Very important progress has recently been made in this area with the proof that CSPs enjoy a dichotomy between \PTime and \NPclass \cite{DBLP:journals/corr/Bulatov17a,DBLP:journals/corr/Zhuk17}. Via the results in \cite{BienvenuCLW14}, this implies that OMQ evaluation in languages such as $(\ALC,\text{UCQ})$ enjoys a dichotomy between \PTime and \coNP. However, the picture is still far from being fully understood. For example, neither in CSP nor in expressive OMQ languages it is known whether there is a dichotomy between \NL and \PTime, and whether containment in \NL coincides with rewritability into linear Datalog.

The aim of this paper is to carry out an ultimately fine-grained analysis of the data complexity and rewritability of OMQs from the languages $(\EL,\textnormal{CQ})$ and $(\EL,\text{AQ})$ where \EL is a fundamental and widely known Horn DL that is at the core of the OWL EL profile of the OWL 2 ontology language \cite{profiles}. In fact, we completely settle the complexity and rewritability status of each OMQ from $(\EL,\textnormal{CQ})$. Our first main result is a trichotomy: Every OMQ from $(\EL,\textnormal{CQ})$ is in $\AC^0$, $\NL$-complete, or $\PTime$-complete, and all three complexities actually occur already in $(\EL,\textnormal{AQ})$.  We consider this a remarkable sparseness of complexities. Let us illustrate the trichotomy using an example. Formally, an OMQ from $(\EL,\textnormal{CQ})$ is a triple $(\Tmc,\Sigma,q)$ with \Tmc an \EL TBox that represents the ontology, $q$ a CQ, and $\Sigma$ an ABox signature, that is, a set of concept and role names that can occur in the data.
\begin{example}
Consider an ontology that represents knowledge about genetic diseases, where $\mn{Disease1}$ is caused by $\mn{Gene1}$ and $\mn{Disease2}$ by $\mn{Gene2}$. A patient carries $\mn{Gene1}$ if both parents carry $\mn{Gene1}$, and the patient carries $\mn{Gene2}$ if at least one parent carries $\mn{Gene2}$ (dominant and recessive inheritance, respectively). Let
$$
\begin{array}{rcl}
\Tmc = \{ \hspace*{4cm}\mn{Gene1} &\sqsubseteq &\mn{Disease1},\\
\mn{Gene2} &\sqsubseteq &\mn{Disease2} \\
\exists \mn{father}.\mn{Gene1} \sqcap \exists \mn{mother}.\mn{Gene1} &\sqsubseteq &\mn{Gene1} \\
\exists \mn{father}.\mn{Gene2}& \sqsubseteq &\mn{Gene2}\\
\exists \mn{mother}.\mn{Gene2}& \sqsubseteq &\mn{Gene2} \qquad \}.
\end{array}
$$
For $\Sigma = \{\mn{Gene1}, \mn{Gene2}, \mn{mother}, \mn{father}\}$, the OMQ $(\Tmc, \Sigma, \mn{Disease1}(x))$ is $\PTime$-complete and not rewritable into linear Datalog, while $(\Tmc, \Sigma, \mn{Disease2}(x))$ is $\NL$-complete and rewritable into linear Datalog.
\end{example}
Our second main result is that for OMQs from $(\EL,\textnormal{CQ})$, evaluation in \NL coincides with rewritability into linear Datalog. It is known that evaluation in AC$^0$ coincides with FO-rewritability \cite{ijcai16} and thus each of the three occurring complexities coincides with rewritability into a well-known database language: AC$^0$ corresponds to FO, \NL to linear Datalog, and \PTime to monadic Datalog. We also show that there is no constant bound on the arity of IDB relations in linear Datalog rewritings, that is, we find a sequence of OMQs from $(\EL,\textnormal{CQ})$ (and in fact, even from $(\EL, \textnormal{AQ})$) that are all rewritable into linear Datalog, but require higher and higher arities of IDB relations. 

We remark that rewritability into linear Datalog might also be interesting from a practical perspective. In fact, the equation ``SQL = FO'' often adopted in ontology-mediated querying ignores the fact that SQL contains linear recursion from its version 3 published in 1999 on, which exceeds the expressive power of FO. We believe that, in the context of OMQs, linear Datalog might be a natural abstraction of SQL that includes linear recursion, despite the fact that it does not contain full FO.  Indeed, the fact that all OMQs from $(\EL,\textnormal{CQ})$ that are FO-rewritable are also UCQ-rewritable shows that the expressive power of FO that lies outside of linear Datalog is not useful when using SQL as a target language for OMQ rewriting. 
 
The second main result is proved using a characterization of linear Datalog rewritability in terms of bounded pathwidth that may be of independent interest. It is easiest to state for $(\EL,\textnormal{AQ})$: an OMQ $Q$ is rewritable into linear Datalog (equivalently: can be evaluated in $\NL$) if the class $\Mmc_Q$ of the following ABoxes \Amc has bounded pathwidth: \Amc is tree-shaped, delivers the root as an answer to $Q$, and is minimal w.r.t.\ set inclusion regarding the latter property. For $(\EL,\textnormal{CQ})$, we have to replace in $\Mmc_Q$ tree-shaped ABoxes with pseudo tree-shaped ones in which the root is an ABox that can have any relational structure, but whose size is bounded by the size of the actual query in $q$.  These results are closely related to results on bounded pathwidth obstructions of CSPs, see for example \cite{Dalmau05,DalmauK08,CarvalhoDK10}.

Finally, we consider the meta problems associated to the studied properties of OMQs, such as whether a given OMQ is rewritable into linear Datalog, \NL-hard, \PTime-hard, etc. Each of these problems turns out to be \ExpTime-complete, both in $(\EL,\textnormal{CQ})$ and in $(\EL,\textnormal{AQ})$. In the case of linear Datalog rewritability, our results provide a way of constructing a concrete rewriting when it exists. 

This paper is organized as follows. We introduce preliminaries in Section~\ref{sec:prelims} and then start with considering the OMQ language $(\EL,\text{conCQ})$ where conCQ refers to the class of CQs that are connected when viewed as a graph; these CQs might have any arity, including 0. In Section~\ref{sec:AC0NL}, we show that $(\EL,\text{conCQ})$ enjoys a dichotomy between AC$^0$ and \NL, using a notion of bounded depth that was introduced in \cite{ijcai16}. In particular, it was shown in \cite{ijcai16} that when the ABoxes in $\Mmc_Q$ have bounded depth, then $Q$ can be evaluated in AC$^0$. We prove that otherwise, we find certain gadget ABoxes (we say that $Q$ \emph{has the ability to simulate \REACH}) that allow us to reduce the reachability problem in directed graphs, thus showing \NL-hardness.  In Section~\ref{sec:NLPTime}, we prove a dichotomy between \NL and \PTime, still for $(\EL,\text{conCQ})$. We first show that if $\Mmc_Q$ has unbounded pathwidth, then we can find certain gadget ABoxes (we say that $Q$ \emph{has the ability to simulate \PSA}) that allow us to reduce the path accessibility problem, thus showing \PTime-hardness. This result is similar to, but substantially more difficult than the \NL-hardness result in Section~\ref{sec:AC0NL}. We then proceed by showing that if $\Mmc_Q$ has bounded pathwidth, then we can construct a two-way alternating word automaton that accepts suitable representations of pairs $(\Amc,\abf)$ where \Amc is an ABox of low pathwidth and \abf and answer to $Q$ on \Amc. We further show how to convert this automaton into a linear Datalog rewriting, which yields $\NL$ complexity.  Section~\ref{sec:disconnected} is concerned with extending both of our dichotomies to potentially disconnected CQs. In Section~\ref{sec:hierarchy}, we prove that there is a sequence of OMQs that are linear Datalog rewritable but for which the width of IDB relations in linear Datalog rewritings is not bounded by a constant. This strengthens a result by \cite{DalmauK08} who establish an analogous statement for CSPs.  In Section~\ref{sec:decidability} we prove decidability and \ExpTime-completeness of the meta problems. The upper bounds are established using the ability to 
simulate \PSA from Section~\ref{sec:NLPTime} and alternating tree automata.

This paper is an extended version of \cite{LS-IJCAI17}. The main differences are that \cite{LS-IJCAI17} only treats atomic queries but no conjunctive queries, does not provide characterizations in terms of bounded pathwidth, and achieves less optimal bounds on the width of IDB relations in constructed linear Datalog programs.

\section{Preliminaries}
\label{sec:prelims}

We introduce description logics, ontology-mediated queries, central technical notions such as universal models and the pathwidth of ABoxes, as well as linear Datalog and a fundamental glueing construction for ABoxes. We refer to \cite{DL-Textbook} for more extensive information on description logics and to \cite{AbiteboulHV95} for background in database theory.

\smallskip
\noindent
\textbf{TBoxes and Concepts.} In description logic, an ontology is formalized as a TBox. Let $\NC$, $\NR$, and $\NI$ be disjoint countably infinite sets of \emph{concept names}, \emph{role names}, and \emph{individual names}.  An \emph{\EL-concept} is built according to the syntax rule $C,D ::= \top \mid A \mid C \sqcap D \mid \exists r . C$ where $A$ ranges over concept names and $r$ over role names. While this paper focuses on $\EL$, there are some places where we also consider the extension \ELI of \EL with inverse roles. An \emph{\ELI-concept} is built according to the syntax rule $C,D ::= \top \mid A \mid C \sqcap D \mid \exists r . C \mid \exists r^-.C$, the symbol ranges being as in the case of \EL-concepts. An expression of the form $r^-$ is an \emph{inverse role}.  An \emph{\EL-TBox} ($\ELI$-TBox, resp.)  is a finite set of \emph{concept inclusions (CIs)} of the form $C \sqsubseteq D$, $C$ and $D$ \EL-concepts (\ELI-concepts, resp.). 

The \emph{size} of a TBox, a concept, or any other syntactic object $O$, denoted $|O|$, is the number of symbols needed to write $O$, with each concept and role name counting as one symbol.

\smallskip
\noindent
\textbf{ABoxes.} An \emph{ABox} is the DL way to store data. Formally, it is 
defined as a finite set of \emph{concept assertions} $A(a)$ and
\emph{role assertions} $r(a,b)$ where $A$ is a concept name, $r$ a
role name, and $a,b$ individual names. We use $\mn{ind}(\Amc)$ to
denote the set of individuals of the ABox $\Amc$.
A \emph{signature}
is a set of concept and role names. We often assume that the ABox is
formulated in a prescribed signature, which we call the \emph{ABox
  signature}.  An ABox that only uses concept and role names from a signature $\Sigma$ is called a \emph{$\Sigma$-ABox}. We remark that the ABox signature plays the same role as a schema in the database literature \cite{AbiteboulHV95}.
If \Amc is an ABox and $S \subseteq \mn{ind}(\Amc)$, then we use $\Amc|_S$
to denote the restriction of \Amc to the assertions that only use individual
names from $S$.
A \emph{homomorphism} from an ABox $\Amc_1$ to an ABox $\Amc_2$ is a function $h:\mn{ind}(\Amc_1) \rightarrow \mn{ind}(\Amc_2)$ such that $A(a) \in \Amc_1$ implies $A(h(a)) \in \Amc_2$ and $r(a,b) \in \Amc_1$ implies $r(h(a),h(b)) \in \Amc_2$.

\smallskip

Every ABox $\Amc$ is associated with a directed graph $G_{\Amc}$ with nodes $\mn{ind}(\Amc)$ and edges $\{(a,b) \mid r(a,b) \in \Amc\}$. A directed graph $G$ is a \emph{tree} if it is acyclic, connected and has a unique node with indegree $0$, which is then called the \emph{root} of $G$. An ABox $\Amc$ is \emph{tree-shaped} if $G_\Amc$ is a tree and there are no multi-edges, that is, $r(a,b) \in \Amc$ implies $s(a,b) \notin \Amc$ for all $s \neq r$. The \emph{root} of a tree-shaped ABox \Amc is the root of $G_\Amc$ and we call an individual $b$ a \emph{descendant} of an individual $a$ if $a \neq b$ and the unique path from the root to $b$ contains $a$.

\smallskip
\noindent
\textbf{Semantics.} An \emph{interpretation} is a tuple $\Imc = (\Delta^\Imc, \cdot^\Imc)$, where $\Delta^\Imc$ is a non-empty set, called the \emph{domain} of $\Imc$, and $\cdot^\Imc$ is a function that assigns to every concept name $A$ a set $A^\Imc \subseteq \Delta^\Imc$ and to every role name $r$ a binary relation $r^\Imc \subseteq \Delta^\Imc \times \Delta^\Imc$. The function $\cdot^\Imc$ can be inductively extended to assign to every $\ELI$ concept $C$ a set $C^\Imc \subseteq \Delta^\Imc$ in the following way.

\begin{align*}
(C_1 \sqcap C_2)^\Imc &\;=\; C_1^\Imc \cap C_2^\Imc\\
(\exists r.C_1)^\Imc &\;=\; \{d \in \Delta^\Imc \mid \exists \, e \in \Delta^\Imc : r(d,e) \wedge C_1(e)\}\\
(\exists r^-.C_1)^\Imc &\;=\; \{e \in \Delta^\Imc \mid \exists \, d \in \Delta^\Imc : r(d,e) \wedge C_1(d)\}\\
\top^\Imc &\;=\; \Delta^\Imc
\end{align*}

An interpretation $\Imc$ \emph{satisfies} a CI $C \sqsubseteq D$ if $C^\Imc \subseteq D^\Imc$, a concept assertion $A(a)$ if $a \in A^\Imc$, and a role assertion $r(a,b)$ if $(a,b) \in r^\Imc$. It is a \emph{model} of a TBox \Tmc if it satisfies all CIs in it and a \emph{model} of an ABox \Amc if it satisfies all assertions in it. For an interpretation \Imc and $\Delta \subseteq \Delta^\Imc$, we use $\Imc|_\Delta$ to
denote the restriction of \Imc to the elements in $\Delta$. 

\smallskip
\noindent
\textbf{Conjunctive queries.} A \emph{conjunctive query (CQ)} is a first order formula of the form $q=\exists \ybf \phi(\xbf,\ybf)$, $\phi$ a conjunction of relational atoms, that uses only unary and binary relations that must be from $\NC$ and $\NR$, respectively. A \emph{CQ with equality atoms} is a CQ where additionally, atoms of the form $x=y$ are allowed. We also interpret $q$ as the set of its atoms. The variables in $\xbf$ are called \emph{answer variables} whereas the variables in $\ybf$ are called \emph{quantified variables}. We set $\mn{var}(q) = \xbf \cup \ybf$.  Every CQ $q$ can be viewed as an ABox $\Amc_q$ by viewing (answer and quantified) variables as individual names. A CQ is \emph{connected} if $G_{\Amc_q}$ is and \emph{rooted} if every connected component of $G_{\Amc_q}$ contains at least one answer variable. A CQ is \emph{tree-shaped} if $\Amc_q$ is. If $q$ is a CQ and $V \subseteq \mn{var}(q)$, then we use $q|_V$ to denote the restriction of $q$ to the atoms that only use variables from $V$ (this may change the arity of $q$).  An \emph{atomic query (AQ)} is a CQ of the form $A(x)$.

A \emph{union of conjunctive queries (UCQ)} $q$ is a disjunction of CQs that have the same answer variables. We write $q(\xbf)$ to emphasize that \xbf are the answer variables in $q$.  The \emph{arity} of a (U)CQ $q$, denoted $\mn{ar}(q)$, is the number of its answer variables. We say that $q$ is \emph{Boolean} if $\mn{ar}(q)=0$.  Slightly overloading notation, we write CQ to denote the set of all CQs, $\textnormal{CQ}^=$ to denote the set of all CQs where equality atoms are allowed, conCQ for the set of all connected CQs, AQ for the set of all AQs, and UCQ for the set of all UCQs.

Let $q(\xbf)$ be a UCQ and \Imc an interpretation. A tuple $\abf \in (\Delta^\Imc)^{\mn{ar}(q)}$ is an \emph{answer to $q$ on \Imc}, denoted $\Imc \models q(\abf)$, if there is a \emph{homomorphism} $h$ from $q$ to \Imc with $h(\xbf)=\abf$, that is, a function $h:\mn{var}(q) \rightarrow \Delta^\Imc$ such that $A(x) \in q$ implies $h(x) \in A^\Imc$ and $r(x,y) \in q$ implies $(h(x),h(y)) \in r^\Imc$.

\smallskip
\noindent
\textbf{Ontology-mediated queries.}  An \emph{ontology-mediated query (OMQ)} is a triple $Q=(\Tmc, \Sigma, q)$ that consists of a TBox~$\Tmc$, an ABox signature $\Sigma$ and a query $q$ such as a CQ or a UCQ.  Let \Amc be a $\Sigma$-ABox. A tuple $\abf \in \mn{ind}(\Amc)^{\mn{ar}(q)}$ is an \emph{answer to $Q$ on \Amc}, denoted $\Amc \models Q(\abf)$, if for every common model $\Imc$ of $\Amc$ and $\Tmc$, \abf is an answer to $q$ on \Imc. If the TBox should be emphasized, we write $\Tmc, \Amc \models q(\abf)$ instead of $\Amc \models Q(\abf)$. 
For an ontology language \Lmc and query language $\Qmc$, 
we use $(\Lmc, \Qmc)$ to denote the OMQ language in which TBoxes are formulated in $\Lmc$ and the actual queries are from \Qmc; we also identify this language with the set of all OMQs that it admits. In this paper, we mainly concentrate on the OMQ languages $(\EL,\textnormal{CQ})$ and $(\EL,\textnormal{AQ})$.

For an OMQ $Q=(\Tmc,\Sigma,q)$, we use {\sc eval}$(Q)$ to denote the following problem: given a $\Sigma$-ABox \Amc and a tuple $\abf \in \mn{ind}(\Amc)^{\mn{ar}(q)}$, decide whether $\Amc \models Q(\abf)$.

\smallskip
\noindent
\textbf{TBox normal form.} Throughout the paper, we generally and without further notice assume TBoxes to be in {\em normal form}, that is, to contain only concept inclusions of the form
$\exists r.A_1 \sqsubseteq A_2$, $\top \sqsubseteq A_1$, $A_1 \sqcap A_2 \sqsubseteq A_3$, $A_1 \sqsubseteq \exists r.A_2$, where all $A_i$ are concept names and $r$ is a role name or, in the case of \ELI-TBoxes, an inverse role. Every TBox \Tmc can be converted into a TBox $\Tmc'$ in normal form in linear time~\cite{BaaderBL05}, introducing fresh concept names; the resulting TBox $\Tmc'$ is a conservative extension of~\Tmc, that is, every model of $\Tmc'$ is a model of $\Tmc$ and, conversely, every model of $\Tmc$ can be extended to a model of $\Tmc'$ by interpreting the fresh concept names. Consequently, when \Tmc is replaced in an OMQ $Q=(\Tmc,\Sigma,q)$ with $\Tmc'$, resulting in an OMQ $Q'$, then $Q$ and $Q'$ are equivalent in the
sense that they give the same answers on all $\Sigma$-ABoxes. Thus, conversion
of the TBox in an OMQ into normal form does not impact its data complexity nor
rewritability into linear Datalog (or any other language).

\smallskip
\noindent
\textbf{First order Rewritability.} Let $Q = (\Tmc, \Sigma, q) \in (\EL, \textnormal{CQ})$. We call $Q$ \emph{FO-rewritable} if there exists a first-order formula $\varphi(\xbf)$ without function symbols and constants, potentially using equality, and using relational atoms of arity one and two only, drawing unary relation symbols from \NC and binary relation symbols from \NR such that for every ABox $\Amc$ and every tuple $\abf$ of individuals of $\Amc$, we have $\Amc \models Q(\abf)$ if and only if $\Amc \models \varphi(\abf)$, where $\Amc$ is interpreted as a relational structure over $\Sigma$.

\smallskip
\noindent
\textbf{Linear Datalog Rewritability.} A \emph{Datalog rule} $\rho$ has the form $S(\xbf) \leftarrow R_1(\ybf_1)\land \cdots\land R_n(\ybf_n)$, $n> 0$, where $S,R_1,\dots,R_n$ are relations of any arity and $\xbf, \ybf_i$ denote tuples of variables.  We refer to $S(\xbf)$ as the \emph{head} of $\rho$ and to $R_1(\ybf_1) \wedge \cdots \wedge R_n(\ybf_n)$ as the \emph{body}. Every variable that occurs in the head of a rule is required to also occur in its body. 
A \emph{Datalog program} $\Pi$ is a finite set of Datalog rules with a selected \emph{goal relation} \mn{goal} that does not occur in rule bodies.
The \emph{arity of $\Pi$}, denoted $\mn{ar}(\Pi)$, is the arity of the \mn{goal} relation.
Relation symbols that occur in the head of at least one rule of $\Pi$ are \emph{intensional (IDB) relations}, and all remaining relation symbols in $\Pi$ are \emph{extensional (EDB)
  relations}.  In our context, EDB relations must be unary or binary and are identified with concept names and role names.  Note that, by definition, \mn{goal} is an IDB relation. A Datalog program is \emph{linear} if each rule body contains at most one IDB relation. The \emph{width} of a Datalog program is the maximum arity of non-\mn{goal} IDB relations used in it and its \emph{diameter} is the maximum number of variables that occur in a rule in $\Pi$.

For an ABox \Amc that uses no IDB relations from $\Pi$ and a tuple $\abf \in \mn{ind}(\Amc)^{\mn{ar}(\Pi)}$, we write $\Amc \models \Pi(a)$ if $a$ is an answer to $\Pi$ on \Amc, defined in the usual way \cite{AbiteboulHV95}: $\Amc \models \Pi(a)$ if $\mn{goal}(a)$ is a logical consequence of $\Amc \cup \Pi$ viewed as a set of first-order sentences (all variables in rules quantified universally).  We also admit body atoms of the form $\top(x)$ that are vacuously true.  This is just syntactic sugar since any rule with body atom $\top(x)$ can equivalently be replaced by a set of rules obtained by replacing $\top(x)$ in all possible ways with an atom $R(x_1,\dots,x_n)$ where $R$ is an EDB relation and where $x_i=x$ for some $i$ and all other $x_i$ are fresh variables.

A Datalog program $\Pi$ over EDB signature $\Sigma$ is a \emph{rewriting} of an OMQ $Q=(\Tmc,\Sigma,q)$ if $\Amc \models Q(\abf)$ iff $\Amc \models \Pi(\abf)$ for all $\Sigma$-ABoxes \Amc and all $\abf \in \mn{ind}(\Amc)$.  We say that $Q$ is \emph{(linear) Datalog-rewritable} if there is a (linear) Datalog program that is a rewriting of~$Q$. It is well-known that all OMQs from $(\EL,\textnormal{CQ})$ are Datalog-rewritable. It follows from the results in this paper that there are rather simple OMQs $Q=(\Tmc,\Sigma,q)$ that are not linear Datalog-rewritable, choose e.g.\ $\Tmc = \{ \exists r . A \sqcap \exists s . A \sqsubseteq A \}$, $\Sigma = \{r,s,A\}$, and $q=A(x)$.

\smallskip
\noindent
\textbf{Universal models.}
It is well known \cite{DBLP:journals/jsc/LutzW10} that for every $\ELI$-TBox $\Tmc$ and ABox $\Amc$ there is a \emph{universal model} $\Umc_{\Amc,\Tmc}$ with certain nice properties. These are summarized in the following lemma. Homomorphisms
between interpretations are defined in the expected way, ignoring individual names.
\begin{lemma} \label{lem:unimodelproperties}
Let $\Tmc$ be an $\ELI$-TBox in normal form and $\Amc$ an ABox.
Then there is an interpretation $\Umc_{\Amc,\Tmc}$ such that
\begin{enumerate}
\item $\Umc_{\Amc,\Tmc}$ is a model of $\Amc$ and $\Tmc$;
\item for every model $\Imc$ of $\Amc$ and $\Tmc$, there is a homomorphism
from $\Umc_{\Amc,\Tmc}$ to $\Imc$ that is the identity on $\mn{ind}(\Amc)$;
\item for all CQs $q$ and $\abf \in \mn{ind}(\Amc)^{\mn{ar}(q)}$, 
$\Amc,\Tmc \models q(\abf)$ iff $\Umc_{\Amc,\Tmc} \models q(\abf)$.
\end{enumerate}
\end{lemma}
$\Umc_{\Amc,\Tmc}$ can be constructed using a standard chase procedure, as follows.  We define a sequence of ABoxes $\Amc_0,\Amc_1,\dots$ by setting $\Amc_0 = \Amc$ and then letting $\Amc_{i+1}$ be $\Amc_i$ extended as follows:
%
%
\begin{itemize}

\item[(i)] If $\exists r.B \sqsubseteq A \in \Tmc$ and $r(a,b), B(b) \in \Amc_i$, then add $A(a)$ to $\Amc_{i+1}$;

\item[(ii)] If $\exists r^-.A \sqsubseteq B \in \Tmc$ and $r(a,b), A(a) \in \Amc_i$, then add $B(b)$ to $\Amc_{i+1}$;

\item[(iii)] if $\top \sqsubseteq A \in \Tmc$ and $a \in \mn{ind}(\Amc_{i})$,
then add $A(a)$ to $\Amc_{i+1}$;

\item[(iv)] if $B_{1}\sqcap B_{2} \sqsubseteq A\in \Tmc$ and $B_{1}(a),B_{2}(a)\in \Amc_{i}$,
then add $A(a)$ to $\Amc_{i+1}$;

\item[(v)] if $A \sqsubseteq \exists r.B \in \Tmc$, $A(a) \in \Amc_i$ and there is no $b \in \mn{ind}(\Amc_i)$ such that $r(a,b)$ and $B(b)$, then take a fresh individual $b$ and add $r(a,b)$ and $B(b)$ to $\Amc_{i+1}$;

\item[(vi)] if $B \sqsubseteq \exists r^-.A \in \Tmc$, $B(b) \in \Amc_i$ and there is no $a \in \mn{ind}(\Amc_i)$ such that $r(a,b)$ and $A(a)$, then take a fresh individual $a$ and add $r(a,b)$ and $A(a)$ to $\Amc_{i+1}$.
\end{itemize}
Let $\Amc_{\omega}=\bigcup_{i\geq 0}\Amc_{i}$.  We define $\Umc_{\Amc,\Tmc}$ to be the interpretation that corresponds to~$\Amc_{\omega}$. This does actually not define
$\Umc_{\Amc,\Tmc}$ in a unique way since the order or applying the above rules may have an impact on the shape of $\Amc_\omega$. However, all resulting $\Amc_\omega$ are homomorphically equivalent and it does not matter for the constructions in this paper which order we use. Slightly sloppily, we thus live with the fact that $\Umc_{\Amc,\Tmc}$ is not uniquely defined.  Note that $\Umc_{\Amc,\Tmc}$ can be infinite and that its shape is basically the shape of~$\Amc$, but with a (potentially infinite) tree attached to every individual in $\Amc$. The domain elements in these trees are introduced by Rules~(v) and~(vi), and we refer to them as \emph{anonymous elements}. The properties in Lemma \ref{lem:unimodelproperties} are standard to prove, see for example \cite{BO15, DL-Textbook} for similar proofs.

The \emph{degree} of an ABox \Amc is the maximum number of successors of any
individual in \Amc. The following lemma often allows us to concentrate on ABoxes
of small degree. We state it only for $(\mathcal{EL}, \text{AQ})$, since we only use it
for these OMQs.
\begin{lemma}
\label{lem:smalldegree}
  Let $Q=(\Tmc,\Sigma, A(x))\in (\mathcal{EL}, \textnormal{AQ})$ be an OMQ and \Amc
  a $\Sigma$-ABox such that $\Amc \models Q(a)$. 
  Then there exists $\Amc' \subseteq \Amc$ of degree at most $|\Tmc|$ such that
  $\Amc' \models Q(a)$. 
\end{lemma}
\begin{proof}
  (sketch) Assume $\Amc \models Q(a)$ and let $\Amc_\omega$ be the
  ABox produced by the chase procedure described above.
  Since $\Amc \models Q(a)$, by Lemma~\ref{lem:unimodelproperties},
  $A(a) \in \Amc_\omega$. Let $\Amc'$ be obtained from \Amc by removing
  all assertions $r(a,b)$ that did not participate in any application of rule (i), (ii), (v) or (vi)
  and let $\Amc'_c$ be the result of chasing $\Amc'$.  Clearly, we must have
  $A(a) \in \Amc'_c$. Moreover, it is easy to verify that the degree of $\Amc'$
  is at most $|\Tmc|$.
\end{proof}

\smallskip
\noindent
\textbf{Pathwidth.} A \emph{path decomposition} of a (directed or undirected) graph $G=(V,E)$ is a sequence $V_1,\ldots,V_n$ of subsets of $V$, such that
\begin{itemize}
\item $V_i \cap V_k \subseteq V_j$ for $1 \leq i \leq j \leq k \leq n$ and
\item $\bigcup_{i=1}^n V_i = V$.
\end{itemize}
A path decomposition $V_1,\ldots,V_n$ is an \emph{$(\ell,k)$-path decomposition} if $\ell=\max \{|V_i \cap V_{i+1}| \mid 1 \leq i \leq n-1\}$ and $k=\max \{|V_i| \mid 1 \leq i \leq n\}$. The \emph{pathwidth} of $G$, denoted $\mn{pw}(G)$, is the smallest integer $k$, such that $G$ has a $(\ell,k+1)$-path decomposition for some $\ell$. Note that paths have pathwidth~1. For an ABox \Amc, a sequence $V_1,\ldots,V_n$ of subsets of $\mn{ind}(\Amc)$ is a path decomposition of $\Amc$ if $V_1,\ldots,V_n$ is a path decomposition of~$G_\Amc$. We assign a pathwidth to \Amc by setting $\mn{pw}(\Amc) := \mn{pw}(G_\Amc)$.

\smallskip
\noindent
\textbf{Treeifying CQs.}  A Boolean CQ $q$ is \emph{treeifiable} if there exists a homomorphism from $q$ into a tree-shaped interpretation.  With every treeifiable Boolean CQ $q$, we associate a tree-shaped CQ $q^\mn{tree}$ that is obtained by starting with $q$ and then exhaustively \emph{eliminating forks}, that is, identifying $x_1$ and $x_2$ whenever there are atoms $r(x_1,y)$ and $r(x_2,y)$.  Informally, one should think of $q^\mn{tree}$ as the least constrained treeification of $q$.  It is known that a CQ $q$ is treeifiable if and only if the result of exhaustively eliminating forks is tree-shaped \cite{DBLP:conf/cade/Lutz08}. Consequently, it can be decided in polynomial time whether a Boolean CQ is treeifiable.

One reason for why treeification is useful is that every tree-shaped Boolean CQ $q$ can be viewed as an \EL-concept $C_q$ in a straightforward way.  If, for example,
$$
q = \exists w \exists x \exists y \exists z \, r(x,y) \wedge s(y,z) \wedge r(y,w) \wedge A(y) \wedge B(w),
$$
then $C_q= \exists r.(A \;\sqcap\; \exists s.\top \;\sqcap\; \exists r.B)$.

A pair of variables $(x,y)$ from a CQ $q$ is \emph{guarded} if $q$ contains an atom of the form $r(x,y)$.  For every guarded pair $(x,y)$ and every $i \geq 0$, define $\mn{reach}^i(x,y)$ to be the smallest set such that
\begin{enumerate}

\item $x \in \mn{reach}^0(x,y)$ and $y \in \mn{reach}^1(x,y)$;

\item if $z \in \mn{reach}^i(x,y)$, $i > 0$, and $r(z,u) \in q$, then
  $u \in \mn{reach}^{i+1}(x,y)$;

\item if $u \in \mn{reach}^{i+1}(x,y)$, $i > 0$ and $r(z,u) \in q$, then
  $z \in \mn{reach}^{i}(x,y)$.

\end{enumerate}
Moreover, $\mn{reach}(x,y) = \bigcup_i \mn{reach}^i(x,y)$. We use $\mn{trees}(q)$ to denote the set of all (tree-shaped) CQs $p^{\mn{tree}}$ such that $p=q|_{\mn{reach}(x,y)}$ for some guarded pair $(x,y)$ with $p$ treeifiable.

It is easy to verify that the number of CQs in $\mn{trees}(q)$ is linear in $|q|$.  We briefly argue that $\mn{trees}(q)$ can be computed in polynomial time. The number of guarded pairs is linear in $|q|$. For each guarded pair $(x,y)$, $\mn{reach}(x,y)$ can clearly be computed in polynomial time. Moreover, exhaustively eliminating forks on $p=q|_{\mn{reach}(x,y)}$ takes only polynomial time, which tells us whether $p$ is treeifiable and constructs $p^{\mn{tree}}$ if this is the case.

\smallskip
\noindent
\textbf{Pseudo tree-shaped ABoxes.} Throughout the paper, we often concentrate on ABoxes that take a restricted, almost tree-shaped form. These are called pseudo tree-shaped ABoxes, introduced in \cite{ijcai16}. An ABox $\Amc$ is a \emph{pseudo tree-shaped ABox of core size $n$} if there exist ABoxes $\Cmc, \Amc_1, \ldots, \Amc_k$ such that $\Amc = \Cmc \cup \bigcup_{i=1}^k \Amc_i$, \mbox{$|\mn{ind}(\Cmc)|=n$}, and
all $\Amc_i$ are tree-shaped ABoxes with pairwise disjoint individuals and $\mn{ind}(\Cmc) \cap \mn{ind}(\Amc_i)$ consists precisely of the root of~$\Amc_i$. We call \Cmc the \emph{core} of \Amc. The tree-shaped ABoxes $\Amc_1,\dots,\Amc_k$ that are part of a pseudo tree-shaped ABox should not be confused with the anonymous trees that are added when chasing a pseudo tree-shaped ABox to construct a universal model. Note that every tree-shaped ABox is pseudo tree-shaped with core size~$1$.

The following lemma describes the central property of pseudo tree-shaped ABoxes. It essentially says that if \abf is an answer to an OMQ $Q$ based on a connected CQ $q$ on an ABox~$\Amc$, then one can unravel $\Amc$ into a pseudo tree-shaped ABox $\Amc'$ that homomorphically maps to $\Amc$ and such that \abf is an answer to $Q$ on $\Amc'$, witnessed by a homomorphism from $q$ to $\Umc_{\Amc',\Tmc}$ that satisfies the additional property of being within or at least `close to' the core of $\Amc'$. 
\begin{lemma} \label{lem:pseudo} Let $Q = (\Tmc, \Sigma, q) \in (\EL, \textnormal{conCQ})$, $\Amc$ a $\Sigma$-ABox and $\abf \in \mn{ind}(\Amc)^{\mn{ar}(q)}$ such that $\Amc \models Q(\abf)$. Then there is a pseudo tree-shaped $\Sigma$-ABox $\Amc'$ of core size at most $|q|$ and with \abf in its core that satisfies the following conditions:
\begin{enumerate}

\item there is a homomorphism from $\Amc'$ to $\Amc$ that is the identity on $\abf$;

\item $\Amc' \models Q(\abf)$, witnessed by a homomorphism from $q$ to $\Umc_{\Amc', \Tmc}$ whose range consists solely of core individuals and of anonymous elements in a tree rooted in a core individual. 

\end{enumerate}
 \end{lemma}

\begin{proof} (sketch)  Assume that  $\Amc \models Q(\abf)$ and let 
  $h$ be a homomorphism from $q(\xbf)$ to $\Umc_{\Amc, \Tmc}$. Let $I \subseteq \mn{ind}(\Amc)$ be the set of all individuals $b$ that are either in the range of $h$ or such that an anonymous element in the chase-generated tree below $b$ is in the range of $h$. We can unravel $\Umc_{\Amc, \Tmc}$ into a potentially infinite pseudo tree-shaped ABox $\Amc_0$ with core $I$, see \cite{ijcai16} for details. Then $\Amc_0 \models Q(a)$ and this is witnessed by a homomorphism as required by Condition~(2) of Lemma~\ref{lem:pseudo}. However, $\Amc_0$ need not be finite. By the compactenss theorem of first order logic, there exists a finite subset $\Amc_1 \subseteq \Amc_0$ such that $\Amc_1 \models Q(\abf)$. Let $\Amc'$ be the restriction of $\Amc_1$ to those individuals that are reachable in $G_{\Amc'}$ from an individual in $I$. It can be verified that $\Amc'$ is as required. \end{proof}


We shall often be interested in pseudo tree-shaped ABoxes \Amc that give an answer \abf to an OMQ $Q$ and that are minimal with this property regarding set inclusion,
that is, no strict subset of $\Amc$ supports \abf as an answer to $Q$. We introduce
some convenient notation for this. Let $Q = (\Tmc, \Sigma, q) \in (\EL, \textnormal{CQ})$. We use $\Mmc_Q$ to denote the set of all pseudo tree-shaped $\Sigma$-ABoxes $\Amc$ of core size at most $|q|$ such that for some tuple $\abf$ in the core of $\Amc$, $\Amc \models Q(\abf)$ while $\Amc' \not\models Q(\abf)$ for any $\Amc' \subsetneq \Amc$.

\smallskip
\noindent
\textbf{\Tmc-types and Glueing ABoxes.} We introduce a fundamental construction for
merging ABoxes. Let $\Tmc$ be an $\ELI$-TBox. A \emph{$\Tmc$-type} $t$ is a set of concept names from $\Tmc$ that is
closed under $\Tmc$-consequence, that is, if $\Tmc \models \midsqcap t
\sqsubseteq A$, then $A \in t$. For an ABox $\Amc$ and $a \in
\mn{ind}(\Amc)$, we use $\mn{tp}_{\Amc,\Tmc}(a)$ to denote the set of
concept names $A$ from $\Tmc$ such that $\Amc,\Tmc \models A(a)$, which
is a $\Tmc$-type.
The following lemma allows us to glue together ABoxes under certain conditions.





\begin{lemma}
\label{lem:aboxunion}
Let $\Amc_1, \Amc_2$ be $\Sigma$-ABoxes and $\Tmc$ an $\ELI$-TBox such that $\mn{tp}_{\Amc_1, \Tmc}(a) = \mn{tp}_{\Amc_2, \Tmc}(a)$ for all  $a \in \mn{ind}(\Amc_1) \cap \mn{ind}(\Amc_2)$. Then $\mn{tp}_{\Amc_1 \cup \Amc_2, \Tmc} (a) = \mn{tp}_{\Amc_i, \Tmc} (a)$ for all $a \in \mn{ind}(\Amc_i)$, $i \in \{1,2\}$.
\end{lemma}
\begin{proof} 
  Let $\Amc_1$, $\Amc_2$, and \Tmc be as in the lemma. It clearly
  suffices to show that $\mn{tp}_{\Amc_1 \cup \Amc_2, \Tmc}(a)
  \subseteq \mn{tp}_{\Amc_i, \Tmc}(a)$ for all $a \in
  \mn{ind}(\Amc_i)$, $i \in \{1,2\}$. We show the
  contrapositive. Thus, assume that $\Amc_i,\Tmc \not\models A(a)$ for some $i \in \{1, 2\}$. We
  have to show that $\Amc_1 \cup \Amc_2,\Tmc \not\models A(a)$.  Let $\Imc$ be the universal model of $\Tmc$ and $\Amc_1 \cup \Amc_2$ and for each $j \in
  \{1,2\}$, let $\Imc_j$ be the a universal model of \Tmc and
  $\Amc_j$. We can assume w.l.o.g.\ that $\Delta^{\Imc_1} \cap
  \Delta^{\Imc_2} = \mn{ind}(\Amc_1) \cap \mn{ind}(\Amc_2)$. By
  assumption and since $\mn{tp}_{\Amc_1, \Tmc}(a) = \mn{tp}_{\Amc_2,
    \Tmc}(a)$, we must have $a \notin A^{\Imc_1}$ and $a \notin
  A^{\Imc_2}$. Consider the (non-disjoint) union \Imc of $\Imc_1$ and
  $\Imc_2$. Clearly, \Imc is a model of $\Amc_1 \cup \Amc_2$ and $a
  \notin A^\Imc$. To show $\Amc_1 \cup \Amc_2,\Tmc \not\models A(a)$,
  it thus remains to prove that \Imc is a model of \Tmc. To do this,
  we argue that all concept inclusions from $\Tmc$ are satisfied:
  \begin{itemize}
  \item Consider $\exists r. A_1 \sqsubseteq A_2 \in \Tmc$ and $a,b \in \Delta^{\Imc}$ such that $(a,b) \in r^{\Imc}$ and $b \in A_1^{\Imc}$. Then there exist $i,j \in \{1, 2\}$ such that $(a,b) \in r^{\Imc_i}$ and $b \in A_1^{\Imc_j}$. If $i=j$, then $a \in A_2^{\Imc}$, since $\Imc_i$ is a model of $\Tmc$. Otherwise $b \in \Delta^{\Imc_1} \cap \Delta^{\Imc_2} = \mn{ind}(\Amc_1) \cap \mn{ind}(\Amc_2)$, so by assumption, $\mn{tp}_{\Amc_1, \Tmc}(b) = \mn{tp}_{\Amc_2, \Tmc}(b)$. It follows that $A_1 \in \mn{tp}_{\Amc_i, \Tmc}(b)$ and thus, $b \in A_1^{\Imc_i}$. Together with $(a,b) \in r^{\Imc_i}$ and because $\Imc_i$ is a model of $\Tmc$, it follows that $a \in A_2^{\Imc_i} \subseteq A_2^{\Imc}$. Thus, the inclusion $\exists r. A_1 \sqsubseteq A_2$ is satisfied in $\Imc$. 
  \item Consider $\top \sqsubseteq A_1 \in \Tmc$ and $a \in \Delta^{\Imc}$. Then $a \in \Delta^{\Imc_i}$ for some $i \in \{1, 2\}$. Since $\Imc_i$ is a model of $\Tmc$, we have $a \in A_1^{\Imc_i}$, so $a \in A_1^{\Imc}$ and the inclusion $\top \sqsubseteq A_1$ is satisfied in $\Imc$.
  \item Consider $A_1 \sqcap A_2 \sqsubseteq A_3 \in \Tmc$ and $a \in A_1^{\Imc} \cap A_2^{\Imc}$. Then there are $i,j \in \{1, 2\}$ such that $a \in A_1^{\Imc_i}$ and $a \in A_2^{\Imc_j}$. If $i=j$, then $a \in A_3^{\Imc}$ follows, since $\Imc_i$ is a model of $\Tmc$. Otherwise $a \in \Delta^{\Imc_1} \cap \Delta^{\Imc_2} = \mn{ind}(\Amc_1) \cap \mn{ind}(\Amc_2)$, so by assumption, $\mn{tp}_{\Amc_1, \Tmc}(a) = \mn{tp}_{\Amc_2, \Tmc}(a)$. For sure we have $A_1, A_2 \in \mn{tp}_{\Amc_1, \Tmc}(a)$, so we have $a \in A_1^{\Imc_1} \cap A_2^{\Imc_1}$ and since $\Imc_1$ is a model of $\Tmc$, we conclude $a \in A_3^{\Imc_1} \subseteq A_3^{\Imc}$, so the inclusion $A_1 \sqcap A_2 \sqsubseteq A_3$ is satisfied in $\Imc$.
  \item Consider $A_1 \sqsubseteq \exists r. A_2 \in \Tmc$ and $a \in A_1^{\Imc}$. Then $a \in A_1^{\Imc_i}$ for some $i \in \{1, 2\}$. Since $\Imc_i$ is a model of $\Tmc$, we have $b \in \Delta^{\Imc_i}$ and $(a,b) \in r^{\Imc_i}$, hence also $b \in \Delta^{\Imc}$ and $(a,b) \in r^{\Imc}$ and thus, $A_1 \sqsubseteq \exists r. A_2$ is satisfied in $\Imc$.
  \end{itemize} 
\end{proof}

\section{$\AC^0$ versus \NL for Connected CQs}
\label{sec:AC0NL}

We prove a dichotomy between $\AC^0$ and $\NL$ for $(\EL, \textnormal{conCQ})$ and show that for OMQs from this language, evaluation in $\AC^0$ coincides with FO rewritability. The dichotomy does not depend on assumptions from complexity theory since it is known that $\AC^0 \neq \NL$ \cite{FurstSS81}. We generalize the results obtained here to potentially disconnected CQs in Section~\ref{sec:disconnected}.

FO-rewritability of OMQs in $(\EL,\textnormal{CQ})$ has been characterized in \cite{ijcai16} by a property called bounded depth. Informally, an OMQ $Q$ has bounded depth if it looks only boundedly far into the ABox. To obtain our results, we show that unbounded depth implies $\NL$-hardness.  Formally, bounded depth is defined as follows. The \emph{depth} of a tree-shaped ABox $\Amc$ is the largest number $k$ such that there exists a directed path of length $k$ starting from the root in $G_\Amc$. The \emph{depth} of a pseudo tree-shaped ABox is the maximum depth of its trees. 
We say that an OMQ $Q \in (\EL, \textnormal{CQ})$ has \emph{bounded depth} if there is a $k$ such that every $\Amc \in \Mmc_Q$ has depth at most $k$. If there is no such $k$, then $Q$ has \emph{unbounded depth}.

\begin{theorem}
\label{thm:AC0NL}
Let $Q \in (\EL, \textnormal{conCQ})$. The following are equivalent:
\begin{enumerate}[(i)]
\item $Q$ has bounded depth.
\item $Q$ is $FO$-rewritable.
\item {\sc eval}$(Q)$ is in $AC^0$.
\end{enumerate}
If these conditions do not hold, then {\sc eval}$(Q)$ is $\NL$-hard under FO reductions.
\end{theorem}
%
%
The equivalence (ii) $\Leftrightarrow$ (iii) is closely related to a result in CSP.  In fact, every OMQ of the form $(\Tmc, \Sigma, \exists x A(x))$ with $A$ a concept name and $\Tmc$ formulated in $\ELI$ is equivalent to the complement of a CSP \cite{BienvenuCLW14} and it is a known result in CSP that FO-rewritability coincides with $\AC^0$ \cite{BulatovKL08}.  Conjunctive queries, however, go beyond the expressive power of (complements of) CSPs and thus we give a direct proof for (ii) $\Leftrightarrow$ (iii).

The equivalence (i) $\Leftrightarrow$ (ii) follows from Theorem 9 in \cite{ijcai16}. Further, the implication (ii) $\Rightarrow$ (iii) is clear because first order formulas can be evaluated in $\AC^0$. What remains to be shown is thus the implication (iii) $\Rightarrow$ (i) and the last sentence of the theorem. We show that unbounded
depth implies $\NL$-hardness, which establishes both since $\AC^0 \neq \NL$.

We first give a rough sketch of how the reduction works. We reduce from \REACH, the reachability problem in directed graphs, which is $\NL$-complete under FO reductions \cite{immerman}. An input for this problem is a tuple $G=(V,E,s,t)$ where $(V,E)$ is a directed graph, $s \in V$ a \emph{source node} and $t \in V$ a \emph{target node}. Such a tuple is a yes-instance if there exists a path from $s$ to $t$ in the graph $(V,E)$. We further assume w.l.o.g. that $s \neq t$ and that the indegree of $s$ and the outdegree of $t$ are both $0$, which simplifies the reduction.

Let $Q=(\Tmc,\Sigma,q) \in (\EL,\textnormal{conCQ})$ be an OMQ of unbounded depth.  The reduction has to translate a tuple $G=(V,E,s,t)$ into a $\Sigma$-ABox $\Amc_G$ and a tuple $\abf$ such that $\Amc_G \models Q(\abf)$ if and only if there is a path from $s$ to~$t$. We show that any ABox from $\Mmc_Q$ of sufficiently large depth can be used to construct ABoxes $\Amc_\mn{source}$, $\Amc_\mn{edge}$ and $\Amc_\mn{target}$ that can serve as gadgets in the reduction. More precisely, the ABox $\Amc_G$ has (among others) one individual $a_v$ for every node $v \in V$, the edges of $(V,E)$ will be represented using copies of $\Amc_\mn{edge}$, and the source and target nodes will be marked using the ABoxes $\Amc_\mn{source}$ and $\Amc_\mn{target}$, respectively.  We identify two $\Tmc$-types $t_0$ and $t_1$ such that $\mn{tp}_{\Amc_G, \Tmc}(a_v) = t_1$ if $v$ is reachable from $s$ via a path in $G$ and $\mn{tp}_{\Amc_G, \Tmc}(a_v) = t_0$ otherwise. The tuple $\abf$ is then connected to $a_t$ in a way such that $\Amc_G, \Tmc \models q(\abf)$ if and only if $\mn{tp}_{\Amc_G, \Tmc}(a_t) = t_1$.

We next define a property of $Q$, called the \emph{ability to simulate \REACH}, that makes the properties of $\Amc_\mn{source}$, $\Amc_\mn{edge}$, and $\Amc_\mn{target}$ precise, as well as those of the $\Tmc$-types $t_0$ and $t_1$. We then show that $Q$ having unbounded depth implies the ability to simulate \REACH and that this, in turn, implies \NL-hardness via a reduction from \REACH.

If $M$ is a set of concept names, then $M(a)$ denotes the ABox $\{A(a) \mid A \in M\}$. We write $\Amc,\Tmc \models M(a)$ to mean that $\Amc,\Tmc \models A(a)$ for all $A \in M$.  For every pseudo tree-shaped ABox $\Amc$ and a non-core individual $a \in \mn{ind}(\Amc)$, we use $\Amc^a$ to denote the tree-shaped ABox rooted at $a$. Note that every tree-shaped ABox is trivially pseudo tree-shaped with only one tree and where the core consists only of the root individual, so this notation can also be used if $\Amc$ is tree-shaped. Moreover, we use $\Amc_a$ to denote the pseudo tree-shaped ABox $\Amc \setminus \Amc^a$, that is, the ABox obtained from $\Amc$ by removing all assertions that involve descendants of $a$ (making $a$ a leaf) and all assertions of the form $A(a)$.  We also combine these notations, writing for example $\Amc^a_{bc}$ for $((\Amc^a)_b)_c$.

Boolean queries require some special attention in the reduction since they can be made true by homomorphisms to anywhere in the universal model of $\Amc_G$ and~\Tmc, rather than to the neighborhood of the answer tuple~$\abf$ (recall that we work with connected CQs). We thus have to build $\Amc_G$ such that the universal model does not admit unintended homomorphisms. Let 
$\Amc$ be a pseudo tree-shaped $\Sigma$-ABox of core size $|q|$ and $\abf$ a tuple from $\mn{ind}(\Amc)$. We call a homomorphism $h$ from $q$ to $\Umc_{\Amc,\Tmc}$ \emph{core close} if there is some variable $x$ in $q$ such that $h(x) \in \mn{ind}(\Amc)$ is in the core of $\Amc$ or $h(x)$ is an anonymous element in a tree below a core individual. If $\mn{ar}(q)>0$ and $\abf$ is from the core of $\Amc$, then every homomorphism is core close, but this is not true if $q$ is Boolean.
\begin{lemma}
\label{lem:core-close}
Let $Q=(\Tmc,\Sigma,q) \in (\EL,\textnormal{conCQ})$ be Boolean and $\Amc \in \Mmc_Q$. Then every homomorphism from $q$ to $\Umc_{\Amc,\Tmc}$ is core close.
\end{lemma}
\begin{proof}(sketch) Since $\Amc \in \Mmc_Q$, $\Amc$ is minimal with the property that $\Amc \models Q$. Assume that there is a homomorphism $h$ from $q$ to $\Umc_{\Amc,\Tmc}$ that is not core close. Then there is no path in $\Umc_{\Amc, \Tmc}$ from any element in the range of $h$ to any individual in the core of $\Amc$ (though a path in the converse direction might exist). Thus, we can remove all assertions in \Amc that  involve a core individual and the resulting
ABox $\Amc'$ satisfies $\Amc' \models Q$, contradicting the minimality of~$\Amc$. Formally, this can be proved by using Lemma~\ref{lem:pseudo} and showing that
$\Umc_{\Amc,\Tmc}$ and $\Umc_{\Amc',\Tmc}$ are isomorphic when restricted to
non-core individuals and all elements reachable from them on a path.
%
\end{proof}
%
For the rest of this section, we assume w.l.o.g.\ that in any OMQ $Q=(\Tmc,\Sigma,q) \in (\EL,\textnormal{conCQ})$, the TBox \Tmc has been modified as follows: for every $p \in \mn{trees}(q)$, introduce a fresh concept name $A_p$ and add the concept inclusion $C_p \sqsubseteq A_p$ to~$\Tmc$ where $C_p$ is $p$ viewed as an \EL-concept.  Finally, normalize $\Tmc$ again.  It is easy to see that the OMQ resulting from this modification is equivalent to the original OMQ $Q$. The extension is still useful since its types are more informative, now potentially containing also the freshly introduced concept names.
We are now ready to define the ability to simulate \REACH. 
%
\begin{definition}
\label{def:abilitytosimulateREACH}
An OMQ $Q=(\Tmc,\Sigma, q) \in (\mathcal{EL}, \textnormal{conCQ})$ has {\em
  the ability to simulate \REACH} if there exist
\begin{itemize}
\item a pseudo tree-shaped $\Sigma$-ABox $\Amc$ of core size at most $|q|$,
\item a tuple $\abf$ from the core of $\Amc$ of length $\mn{ar}(q)$,
\item a tree $\Amc_i$ of $\Amc$ with two distinguished non-core individuals $b, c$ from $\mn{ind}(\Amc_i)$, where $b$ has distance more than $|q|$ from the core, 
$c$ is a descendant of $b$, and $c$ 
 has distance more than $|q|$ from $b$ and
\item $\Tmc$-types $t_0 \subsetneq t_1$
\end{itemize}
   such that
\begin{enumerate}
\item $\Amc \models Q(\abf)$,
\item $t_1 = \mn{tp}_{\Amc, \Tmc} (b) = \mn{tp}_{\Amc, \Tmc} (c)$,
\item $\mn{tp}_{\Amc_c \cup t_0(c), \Tmc}(b) = t_0$,
\item $\Amc_b \cup t_0(b) \not \models Q(\abf)$ and
\item if $q$ is Boolean, then every homomorphism $h$ from $q$ to $\Umc_{\Amc,\Tmc}$ is core close.
\end{enumerate}
We define $\Amc_{\mn{target}}=\Amc_b$, $\Amc_{\mn{edge}}=\Amc^b_c$,
and $\Amc_{\mn{source}} = \Amc^c$.
\end{definition}
To understand the essence of Definition~\ref{def:abilitytosimulateREACH}, it is worthwhile to consider the special case where $q$ is an AQ~$A(x)$.  In this case, $Q$ has the ability to simulate \REACH if there is a tree-shaped $\Sigma$-ABox \Amc with root $a=\abf$, two distinguished non-root individuals $b,c \in \mn{ind}(\Amc)$, $c$ a descendant of $b$, and \Tmc-types $t_0 \subsetneq t_1$ such that Conditions~(1)-(4) of Definition~\ref{def:abilitytosimulateREACH} are satisfied. All remaining parts of Definition~\ref{def:abilitytosimulateREACH} should be thought of as technical complications induced by replacing AQs with CQs.

\smallskip

In Lemma \ref{lem:simulatereach} we show that unbounded depth implies the ability to simulate \REACH and in Lemma \ref{lem:nlhard} we show that the ability to simulate \REACH enables a reduction from the reachability problem for directed graphs.

\begin{lemma}
\label{lem:simulatereach}
Let $Q \in (\EL, \textnormal{conCQ})$. If $Q$ has unbounded depth, then $Q$ has the ability to simulate \REACH.
\end{lemma}

\begin{proof} We use a pumping argument. Let $Q=(\Tmc,\Sigma,q) \in (\mathcal{EL},\textnormal{conCQ})$ have unbounded depth. There must be a pseudo tree-shaped ABox $\Amc \in \Mmc_Q$ and a tuple $\abf$ from its core such that $\Amc \models Q(\abf)$ and such that one of its trees, say $\Amc_i$, has depth at least $k:= (|q|+2) \cdot 3^{|\Tmc|}+|q|+2$. Consider a path of length at least $k$ from the root of $\Amc_i$ to a leaf. Let $\Amc'$ denote the ABox obtained from $\Amc$ by removing all assertions that involve the leaf in this path. Since $\Amc$ is minimal, $\Amc' \not \models Q(\abf)$. Now, every individual $b$ on the remaining path that has distance more than $|q|$ from the core is colored with the pair $(t'_b,t_b)$ where $t'_b = \mn{tp}_{\Amc', \Tmc}(b)$ and $t_b = \mn{tp}_{\Amc, \Tmc}(b)$. Observing $t'_b \subseteq t_b$, we obtain $3^{|\Tmc|}$ as an upper bound for the number of different colors $(t'_b,t_b)$ that may occur on the path. 
  But the number of individuals on this path with distance more than $|q|$ from the core is $k-|q|-1 = (|q|+2) \cdot 3^{|\Tmc|}+1$, so by the pigeonhole principle there is one color $(t',t)$ that appears $|q|+2$ times on the path. Then there must be distinct individuals $b$ and $c$ that have distance more than $|q|$ from each other and such that $(t'_b,t_b)=(t'_c,t_c)$. W.l.o.g., let $c$ be a descendant of $b$. We set $t_0 = t'_b$ and $t_1 = t_b$.

For this choice of \Amc, $\abf$, $b$, $c$, $t_0$ and $t_1$, Conditions~1 and~2 from Definition~\ref{def:abilitytosimulateREACH} are immediately clear. With Lemma~\ref{lem:aboxunion}, we can replace $\Amc^c$ in $\Amc'$ by $t_0(c)$, so Condition~3 holds. Furthermore, we have $\Amc' \not \models Q(\abf)$ and $\mn{tp}_{\Amc',\Tmc}(b) = t_0$, so again by Lemma~\ref{lem:aboxunion}, if we replace $\Amc^b$ with $t_0(b)$, the types derived in the remaining ABox do not change, thus Condition~4 holds. Condition~5 follows from Lemma~\ref{lem:core-close}.
\end{proof}

Now for the reduction from \REACH to {\sc eval}$(Q)$ when $Q$ has the
ability to simulate \REACH.

\begin{lemma}
\label{lem:nlhard}
Let $Q \in (\EL, \textnormal{conCQ})$. If $Q$ has the ability to simulate \REACH, then $Q$ is $\NL$-hard under FO reductions.
\end{lemma}

\begin{proof}
  Let $Q=(\Tmc,\Sigma,q) \in (\mathcal{EL},\textnormal{conCQ})$ have the ability to simulate \REACH. Then there is a pseudo tree-shaped ABox \Amc, a tuple $\abf$ in its core, distinguished individuals $b$ and $c$, and types $t_0 \subsetneq t_1$ as in Definition~\ref{def:abilitytosimulateREACH}. We reduce \REACH to {\sc eval}$(Q)$. Let $G=(V,E,s,t)$ be an input tuple for \REACH. We construct a $\Sigma$-ABox $\Amc_G$ that represents~$G$. Reserve an individual $a_v$ for every node $v \in V$. For every $(u,v) \in E$, include in $\Amc_G$ a copy $\Amc_{u,v}$ of $\Amc_{\mn{edge}}$ that uses fresh individuals, identifying (the individual that corresponds to) $c$ with $a_u$ and $b$ with~$a_v$.  Further include in $\Amc_G$ one copy of $\Amc_{\mn{target}}$ that uses fresh individuals, identifying $b$ with~$a_t$, and one copy of $\Amc_{\mn{source}}$ that uses fresh individuals, identifying $c$ with~$a_s$. W.l.o.g., we assume that the individuals in $\abf$, added to $\Amc_G$ as part of the copy of $\Amc_{\mn{target}}$, retain their original name. It can be verified that $\Amc_G$ can be constructed from $G$ using an FO-query, see \cite{immerman} for more information on FO reductions. It thus remains to show the following.
\\[2mm]
\textbf{Claim 1.} $t$ is reachable from $s$ in $G$ if and only if $\Amc_G \models Q(\abf)$.
\\[2mm]
For the more straightforward ``$\Rightarrow$'' direction, let $t$ be reachable from $s$. Then there is a path $s=v_0,\ldots,v_n=t$ in $G$.  By definition of $\Amc_G$, there is a copy of $\Amc_{\mn{source}}$ whose root is $a_s$, so Condition~2 from Definition~\ref{def:abilitytosimulateREACH} yields $t_1 \subseteq \mn{tp}_{\Amc_G, \Tmc}(a_s)$. Between any two $a_{v_{i}},a_{v_{i+1}}$ there is a copy of $\Amc_{\mn{edge}}$, so we inductively obtain $t_1 \subseteq \mn{tp}_{\Amc_G,\Tmc}(a_{v_i})$ for all $i$. In particular, $t_1 \subseteq \mn{tp}_{\Amc_G,\Tmc}(a_t)$.  Finally, there is a copy of $\Amc_\mn{target}$ in which $b$ is identified with $a_t$.  By Condition~1, we have $\Amc_G \models Q(\abf)$.

\smallskip

The ``$\Leftarrow$'' direction is more laborious. Assume that $t$ is not reachable from~$s$. Set $$\begin{array}{rcl}\Amc_G' &:=& \Amc_G \cup \{t_0(a_v) \mid v \in V \text{ is not reachable from } s\}\\[1mm] & &\hspace*{6mm} \cup \, \{t_1(a_v) \mid v \in V \text{ is reachable from }s\}.\end{array}$$
We show that $\Amc'_G \not\models Q(\abf)$, which implies $\Amc_G \not\models Q(\abf)$.

\smallskip

We have defined $\Amc'_G$ as an extension of $\Amc_G$. 
%
%
Alternatively and more suitably for what we aim to prove, we can construct $\Amc'_G$ by starting with an ABox $\Amc_0$ that contains only the assertions $t_0(a_v)$ for all unreachable nodes $v \in V$ as well as $t_1(a_v)$ for all reachable nodes $v \in V$ and then exhaustively applying the following rules in an unspecified order, obtaining a sequence of ABoxes $\Amc_0,\Amc_1,\dots,\Amc_m$ with $\Amc_m = \Amc'_G$:
\begin{enumerate} 

\item Choose an edge $(u,v) \in E$ that has not been chosen before, take a copy $\Amc_{\mn{edge}}^{u,v}$ of $\Amc_{\mn{edge}}$ that uses fresh individuals names, with $c$ renamed to $a_u$ and $b$ to $a_v$, and add the assertions $t_{\mn{reach}(x)}(a_x)$ for $x \in \{u,v\}$ where $\mn{reach}(x) = 1$ if $x$ is reachable from $s$ and $\mn{reach}(x) = 0$ otherwise. Set $\Amc_{i+1}=\Amc_i \cup \Amc_{\mn{edge}}^{u,v}$.

\item Introduce a copy $\Amc_{\mn{source}}^s$ of $\Amc_{\mn{source}}$ that uses fresh individual names, with $b$ renamed to $a_s$, and add the assertions $t_1(a_s)$. Set $\Amc_{i+1}=\Amc_i \cup \Amc_{\mn{source}}^s$.

\item Introduce a copy $\Amc_{\mn{target}}^t$ of $\Amc_{\mn{target}}$ that uses fresh individual names with $b$ renamed to $a_t$, and add the assertions $t_0(a_t)$. Set $\Amc_{i+1}= \Amc_i \cup \Amc_{\mn{target}}^t$.

\end{enumerate}
Clearly, rule application terminates after $|E|+2$ steps and results in the
ABox~$\Amc_G'$. Note that we add assertions $t_i(a)$, $i \in \{0,1\}$ to the
ABoxes constructed in the rules to enable application of the ABox glueing
lemma, Lemma~\ref{lem:aboxunion}.
\\[2mm]
\textbf{Claim 2.} $\mn{tp}_{\Amc_i, \Tmc}(a_u) = t_0$ if $u \in V$ is unreachable and $\mn{tp}_{\Amc_i, \Tmc}(a_u) = t_1$ otherwise, for all $i \geq 0$.
\\[2mm]
The proof is by induction on $i$.  For $i=0$, the statement is clear. Now assume that the statement is true for some $i$ and consider $\Amc_{i+1}$. If $\Amc_{i+1}$ was obtained by Rule~1, it follows from Conditions~2 and~3 of Definition~\ref{def:abilitytosimulateREACH} that $\mn{tp}_{\Amc_{\mn{edge}}^{u,v}, \Tmc}(a_x) = t_{\mn{reach}(x)}$ for all $x \in \{u,v\}$. So with Lemma~\ref{lem:aboxunion} and since $\Amc_i$ and $\Amc_{\mn{edge}}^{u,v}$ share only the individuals $a_u,a_v$, the statement follows. If $\Amc_{i+1}$ was obtained by Rule~2, we can use Condition~2 of Definition~\ref{def:abilitytosimulateREACH} and Lemma~\ref{lem:aboxunion}. In the case of Rule~3, it is clear that $\mn{tp}_{\Amc_{\mn{target}}^t, \Tmc}(a_t) = t_0$ and thus it remains to apply Lemma~\ref{lem:aboxunion}. This finishes the proof of Claim~2.

It remains to show that $\Amc'_G,\Tmc \not\models q(\abf)$. Assume to the contrary that $\Amc'_G,\Tmc \not\models q(\abf)$, that is, there is a homomorphism $h$ from $q(\xbf)$ to $\Umc_{\Amc'_G,\Tmc}$ such that \mbox{$h(\xbf)=\abf$}. There can be at most one individual of the form $a_v$ in the range of $h$ by construction of $\Amc'_G$ since $b$ and $c$ have distance exceeding $|q|$.  

If there is no individual $a_v$ in the range of $h$, then $h$ only hits individuals from a single copy of $\Amc_\mn{source}$, $\Amc_\mn{edge}$, or $\Amc_\mn{target}$ as well as anonymous elements in the trees below them (since $q$ is connected).  First assume that this is $\Amc_\mn{target}$. By Claim~2 and since $\mn{reach}(t)=0$, $\mn{tp}_{\Amc'_G, \Tmc}(a_t) = t_0$. It can be shown that the identity function is a homomorphism from $\Umc_{\Amc'_G,\Tmc}|_\Delta$, $\Delta$ the individuals from $\Amc_\mn{target}$ and anonymous elements below them, to $\Umc_{\Amc_\mn{target} \cup t_0(b),\Tmc}$. By composing homomorphisms, it follows that $\Amc_\mn{target} \cup t_0(b),\Tmc \models q(\abf)$, contradicting Condition~4.

Now assume that $h$ only hits individuals from a copy of $\Amc_\mn{source}$ or $\Amc_\mn{edge}$ as well as anonymous elements in the trees below them. Then the restriction \Umc of $\Umc_{\Amc'_G,\Tmc}$ to the range of $h$ is tree-shaped.  Moreover, $q$ must be Boolean since the distance between \abf and the elements of \Umc exceeds $|q|$ and $q$ is treeifiable because $h$ is a homomorphism to a tree-shaped interpretation. We have $d \in C_q^{\Umc_{\Amc'_G,\Tmc}}$ for the root $d$ of \Umc.  From Claim~2 and Lemma~\ref{lem:aboxunion}, it follows that some element of $\Umc_{\Amc_\mn{source} \cup t_1(c),\Tmc}$ or of $\Umc_{\Amc_\mn{edge} \cup t_1(b)\cup t_1(c)\cup t_1(d),\Tmc}$ also satisfies $C_q$. By Condition~2 and Lemma~\ref{lem:aboxunion}, the same is true for an element from $\Umc_{\Amc,\Tmc}$ that is `below' $b$ (reachable from $b$ by a directed path). Since the distance from the core to $b$ in \Amc exceeds~$|q|$, this homomorphism is not core close, contradicting Condition~5.

Now assume that the range of $h$ contains the individual $a_v$. Let $X_0 = \{ x \in \mn{var}(q) \mid h(x)=a_v \}$ and let $X^\downarrow$ (resp.\ $X^\uparrow$) be the set of $x \in \mn{var}(q)$ such that $h(x)$ is some $a \in \mn{ind}(\Amc_G)$ or in an anonymous tree below such an $a$ such that there exists a path of length at least one from $a$ to $a_v$ in $\Amc_G$ (resp.\ from $a_v$ to~$a$). We distinguish three cases.

\emph{Case~1: $v=t$, the target node}. Since we assume that $t$ has outdegree~$0$ in~$G$, $h(x)$ is from the copy of $\Amc_\mn{target}$ or the attached anonymous trees for all $x \in X^\uparrow$ and $h(x)$ is from (potentially multiple) copies of $\Amc_\mn{edge}$ or the attached anonymous trees for all $x \in X^\downarrow$. It is thus possible to construct a homomorphism $g$ from $q$ to $\Umc_{\Amc,\Tmc}$ such that if $h(x) \in \mn{ind}(\Amc_G)$, then $g(x)$ is the individual $h(x)$ in $\Amc$ that $h(x)$ is a copy of. Then $a_v$ being in the range of $h$ implies that $g$ is not core close.

\emph{Case~2: $v=s$, the source node}. Since we assume that $s$ has indegree~$0$ in~$G$, $h(x)$ is from the copy of $\Amc_\mn{source}$ or the attached anonymous trees for all $x \in X^\downarrow$ and $h(x)$ is from (potentially multiple) copies of $\Amc_\mn{edge}$ or the attached anonymous trees for all $x \in X^\uparrow$. We can proceed as in the previous case.

\emph{Case~3: $v \notin \{s,t\}$}.  Then $h$ hits (potentially multiple) copies of $\Amc_\mn{edge}$ and the attached anonymous trees. 
It is possible to construct a homomorphism $g$ from $q|_{X^\uparrow \cup X_0}$ to $\Umc_{\Amc,\Tmc}$ such that if $h(x) \in \mn{ind}(\Amc_G)$ and $x \in X^\uparrow \cup X_0$, then $g(x)$ is the individual in $\Amc$ that $h(x)$ is a copy of and, in particular, $g(x)=c$ for all $x \in X_0$. It remains to extend $g$ to all of $q$.  Let $q'$ be obtained from $q$ by identifying $x_1,x_2 \in \mn{var}(q)$ whenever $r(x_1,y),r(x_2,y) \in q$ and $x_1,x_2 \in X_0 \cup X^{\downarrow}$. It can be verified that the restriction of $\Amc_G$ to all elements between $a_v$ and $\{ h(x) \mid x \in X^\downarrow\}$ is a directed tree. Consequently, $h$ is also a homomorphism from $q'$ to $\Umc_{\Amc_G,\Tmc}$.  Moreover, $q' \setminus q|_{X^\uparrow \cup X_0}$ is the union of tree-shaped CQs $q_1,\dots,q_n$ that all share the same root $x_0$ and are otherwise variable disjoint. Each $q_i$ can be viewed as an \EL-concept $\exists r . C$ such that $\exists r . C\sqsubseteq A_{\exists r . C}$ is in \Tmc and we must thus have $A_{\exists r . C} \in \mn{tp}_{\Amc_G, \Tmc}(a_v) \subseteq \mn{tp}_{\Amc, \Tmc}(b)$. Since $\mn{tp}_{\Amc, \Tmc}(b) = \mn{tp}_{\Amc, \Tmc}(c)$, we find a homomorphism from $q_i$ to $\Umc_{\Amc,\Tmc}$ that maps $x_0$ to $c$, for $1 \leq i \leq n$. Combining all these homomorphisms allows us to extend $g$ to $q'$, thus to $q$.\end{proof}

This finishes the proof of Theorem~\ref{thm:AC0NL}.

\section{$\NL$ versus $\PTime$ for Connected CQs}
\label{sec:NLPTime}

We prove a dichotomy between $\NL$ and $\PTime$ for $(\EL, \textnormal{conCQ})$ and show that for OMQs from this language, evaluation in \NL coincides with rewritability into linear Datalog. We also show that the latter two properties coincide with the OMQ having unbounded pathwidth, as defined below. We generalize our results to potentially disconnected CQs in Section~\ref{sec:disconnected}.

Let $Q=(\Tmc,\Sigma,q) \in (\EL, \textnormal{CQ})$. We say that $Q$ \emph{has pathwidth at most $k$} if for every $\Sigma$-ABox $\Amc$ and tuple $\abf$ with $\Amc \models Q(\abf)$, there is a $\Sigma$-ABox $\Amc'$ of pathwidth at most $k$ such that $\Amc' \models Q(\abf)$ and a homomorphism from $\Amc'$ to $\Amc$ that is the identity on \abf. Now $Q$ has \emph{bounded pathwidth} if it has pathwidth at most $k$ for some $k$. If this is the case, we use $\mn{pw}(Q)$ to denote the smallest $k$ such that $Q$ has pathwidth at most $k$.

\begin{theorem}[$\NL$/$\PTime$ dichotomy]
\label{thm:NLPTime}
Let $Q \in (\EL, \textnormal{conCQ})$. The following are equivalent (assuming $\NL \neq \PTime$):
\begin{enumerate}[(i)]
\item $Q$ has bounded pathwidth.
\item $Q$ is rewritable into linear Datalog.
\item {\sc eval}$(Q)$ is in $\NL$.
\end{enumerate}
If these conditions do not hold, then {\sc eval}$(Q)$ is $\PTime$-hard under FO reductions.
\end{theorem}

\begin{remark}
Without the assumption $\NL \neq \PTime$, Conditions (i) and (ii) are still equivalent to each other and they still imply (iii).
\end{remark}

The equivalence (i) $\Leftrightarrow$ (ii) is closely related to a result in CSP. In fact, it is proved in \cite{Dalmau05} that a CSP has an obstruction set of bounded pathwidth if and only if its complement is expressible in linear Datalog. From the viewpoint
of the connection between OMQs and CSPs \cite{BienvenuCLW14}, obstructions
correspond to homomorphic preimages of ABoxes and thus the result in
\cite{Dalmau05} implies (i) $\Leftrightarrow$ (ii) for OMQs of the form $(\Tmc, \Sigma, \exists x A(x))$, \Tmc formulated in \ELI. We give a direct proof
of (i) $\Leftrightarrow$ (ii) in Section~\ref{sec:boundedpathwidth} to capture
also CQs. 

The implication (ii) $\Rightarrow$ (iii) is clear since every linear Datalog program can be evaluated in $\NL$. It thus remains to prove the converse and the last sentence of the theorem. To achieve both and since we assume $\NL \neq \PTime$, it suffices to show that unbounded pathwidth implies \PTime-hardness.  The structure of the proof is similar to the one for the dichotomy between {\sc AC}$^0$ and \NL in Section~\ref{sec:AC0NL}, but more sophisticated.

\subsection{Unbounded Pathwidth Implies \PTime-hardness}
\label{sec:ptimehardness}

We reduce from the well-known {\sc PTime}-complete problem \emph{path systems accessibility} (\PSA) \cite{immerman}, closely related to alternating reachability on directed graphs and to the evaluation of Boolean circuits.  An instance of \PSA takes the form $G=(V,E,S,t)$ where $V$ is a finite set of nodes, $E$ is a ternary relation on $V$, $S \subseteq V$ is a set of \emph{source nodes}, and $t \in V$ is a \emph{target node}. A node $v \in V$ is \emph{accessible} if $v \in S$ or there are accessible nodes $u,w$ with $(u,w,v) \in E$. $G$ is a yes-instance if the target node $t$ is accessible.  We assume w.l.o.g.\ that $t$ does not appear in the first and second component of a triple in $E$, that no $s \in S$ appears in the third component of a triple in $E$, and that $t \notin S$.

The main difference to the $\NL$-hardness proof in Section~\ref{sec:AC0NL} is that instead of a gadget $\Amc_\mn{edge}$ that transports a selected type $t_1$ from its input individual to its output individual, we now need a gadget $\Amc_\wedge$ with two input individuals and one output individual that behaves like a logical AND-gate.  We formalize this as the ability to simulate \PSA. Instead of proving directly that unbounded pathwidth of an OMQ $Q$ implies that $Q$ has the ability to simulate \PSA, we first prove that unbounded pathwidth of $Q$ implies unbounded branching of $Q$, that is, for any depth bound $n$, there is a pseudo tree-shaped ABox in $\Mmc_Q$ that contains the full binary tree of depth $n$ as a minor. In a second step, we then show that unbounded branching of $Q$ implies that $Q$ has the ability to simulate \PSA. In fact, the ability to simulate \PSA is actually equivalent to unbounded pathwidth and this is useful for the complexity analysis of the meta problems carried out in Section~\ref{sec:decidability}. We thus prove the converse directions as well. Finally, we show that the ability to simulate \PSA implies \PTime-hardness.

We next define unbounded branching more formally. Let $\Amc$ be a tree-shaped ABox. The \emph{full binary tree of depth $k$} is the directed graph $G=(V,E)$ with $V=\{w \in \{1,2\}^* \mid 0 \leq |w| \leq k\}$ and $(v,w) \in E$ if $w=v1$ or $w=v2$. \Amc \emph{has the full binary tree of depth $k$ as a minor} if there is a mapping $f$ from the nodes of the full binary tree of depth $k$ to $\mn{ind}(\Amc)$ such that if $(v,w) \in E$, then $f(w)$ is a descendant of $f(v)$. We do usually not make the mapping $f$ explicit but only say which individuals lie in the range of $f$. We are mostly interested in the largest $k$ such that $\Amc$ has the full binary tree of depth $k$ as a minor. This number, which we call the \emph{branching number} of $\Amc$, denoted by $\mn{br}(\Amc)$, can be easily computed by the following algorithm. Label every leaf of $\Amc$ with $0$ and then inductively label the inner nodes as follows: If $a$ is an inner node whose children have already been labeled and $m$ is the maximum label of its children, label $a$ with $m$ if at most one child of $a$ is labeled with $m$ and label $a$ with $m+1$ if at least two children of $a$ are labeled with $m$. It can be easily proved by induction on the co-depth of an individual that the label of $a$ is equal to $\mn{br}(\Amc^a)$. 
In particular, $\mn{br}(\Amc)$ is the label of the root of $\Amc$. We say that $Q \in (\EL,\textnormal{CQ})$ is \emph{boundedly branching} if there exists a $k$ such that for every pseudo tree-shaped ABox $\Amc \in \Mmc_Q$ and every tree $\Amc_i$ in~\Amc, we have $\mn{br}(\Amc_i) \leq k$. In that case, we define $\mn{br}(Q)$ to be the smallest such~$k$. Otherwise, we call $Q$ \emph{unboundedly branching}.
%
\begin{lemma}
\label{lem:pathwidthbranching}
Let $Q \in (\EL, \textnormal{CQ})$. Then $Q$ has unbounded pathwidth iff $Q$ is unboundedly branching.
\end{lemma}
\begin{proof}
The ``$\Leftarrow$'' direction is clear since the full binary tree of depth $k$ has
pathwith $\lceil \frac{k}{2} \rceil$. For the ``$\Rightarrow$'' direction, we start by showing that for tree-shaped ABoxes, the branching number gives an upper bound on the pathwidth.

\smallskip
\noindent
\textbf{Claim.} Let $\Amc$ be a tree-shaped ABox. Then there exists a $(j,k)$-path decomposition $V_1,\ldots,V_n$ of $\Amc$ with $k \leq \mn{br}(\Amc)+2$ and $j \leq k-1$ such that the root of $\Amc$ is an element of $V_n$.

\smallskip
\noindent
We prove the claim by induction on the depth of $\Amc$. If $\Amc$ has depth $0$, then $\Amc$ has only one individual, $\mn{br}(\Amc) = 0$, and there is a trivial $(0,1)$-path decomposition. If $\Amc$ has depth $1$, then the root $a$ of $\Amc$ has children $a_1,\ldots,a_n$ with $n \geq 1$. We have $\mn{br}(\Amc) \leq 1$ and there is a $(1,2)$-path decomposition $V_1,\ldots,V_n$, where $V_i = \{a,a_i\}$.

If $\Amc$ has depth at least $2$, let the root of $\Amc$ be called $a$ and its children $a_1,\ldots,a_m$. Let $V_1^i,\ldots,V^i_{n_i}$ be the path decomposition of $\Amc^{a_i}$ that exists by induction hypothesis, for $1 \leq i \leq m$. We distinguish two cases:
\begin{itemize}

\item If $\mn{br}(\Amc) = \mn{max}\{\mn{br}(\Amc^{a_i}) \mid 1 \leq i \leq m\}$, then by definition of $\mn{br}$, there is precisely one child $a_i$ of $a$ with $\mn{br}(\Amc^{a_i}) = \mn{br}(\Amc)$. W.l.o.g.\ assume that \mbox{$a_i = a_1$}. Then $V_1^1,\ldots,V^1_{n_1},\{a,a_1\},\{a\} \cup V_1^2,\ldots,\{a\} \cup V_{n_2}^2, \{a\} \cup V_1^3,\ldots,\{a\} \cup V_{n_3}^3,\ldots,\{a\} \cup V_1^m,\ldots,\{a\} \cup V_{n_m}^m$ is a path decomposition of $\Amc$ that fulfils the condition from the claim.  

\item If $\mn{br}(\Amc) = 1+\mn{max}\{\mn{br}(\Amc^{a_i}) \mid 1 \leq i \leq m\}$, then $\{a\} \cup V_1^1,\ldots,\{a\} \cup V^1_{n_1},\{a\} \cup V_1^2,\ldots,\{a\} \cup V_{n_2}^2, \ldots,\{a\} \cup V_1^m,\ldots,\{a\} \cup V_{n_m}^m$ is a path decomposition of $\Amc$ that fulfils the condition from the claim.

 \end{itemize}
This finishes the proof of the claim.

We next show that for every OMQ $Q=(\Tmc,\Sigma,q) \in (\EL, \textnormal{CQ})$, $\mn{br}(Q) =k$ implies $\mn{pw}(Q) \leq k+2+|q|$. Let $Q$ be such an OMQ. Take a $\Sigma$-ABox \Amc and $\abf \in \mn{ind}(\Amc)$ with $\Amc \models Q(\abf)$.  We have to show that there is a $\Sigma$-ABox $\Amc'$ of pathwidth at most $k$ such that $\Amc' \models Q(\abf)$ and there is a homomorphism from $\Amc'$ to \Amc that is the identity on \abf.  By Lemma~\ref{lem:pseudo}, we obtain from \Amc a pseudo tree-shaped $\Sigma$-ABox $\Amc'$ such that there is a homomorphism from $\Amc'$ to \Amc that is the identity on \abf. Clearly, $\Mmc_Q$ contains a subset $\Amc''$ of~$\Amc'$. We show that $\Amc''$ is as required, that is, the pathwidth of $\Amc''$ is at most~$k$. From $\mn{br}(Q) = k$, $\Amc'' \models Q(\abf)$, and $\Amc'' \in \Mmc_Q$, it follows that $\mn{br}(\Amc'') \leq k$. Let $\Amc'$ have core \Cmc and trees $\Amc_1,\ldots,\Amc_m$.
%
By the claim, every $\Amc_i$ has a $(j,k+2)$-path decomposition $V_1^i,\ldots,V_{n_i}^i$. Then we find a $(j+|q|,k+2+|q|)$-path decomposition of~$\Amc'$:
$
\mn{ind}(\Cmc) \cup V_1^1, \ldots, \mn{ind}(\Cmc) \cup V_{n_1}^1, \ldots,\mn{ind}(\Cmc) \cup V_1^m, \ldots, \mn{ind}(\Cmc) \cup V_{n_m}^m.
$
\end{proof}

%
%

Our next goal is to identify suitable gadgets for the reduction from \PSA. To achieve this, it is convenient to extend the TBox of the OMQ $Q=(\Tmc,\Sigma,q)$ involved in the reduction. Recall that we have also used such an extension in the $\NL$-hardness proof in Section~\ref{sec:AC0NL} and that it has helped us to avoid unintended homomorphisms from the CQ to the (universal model of the) reduction ABox~$\Amc_G$, in case the CQ is Boolean. Avoiding such homomorphisms is more complicated in the reduction of \PSA which leads us to a different TBox extension that introduces \ELI-concepts. This is unproblematic since, as in Section~\ref{sec:AC0NL}, the OMQ based on the extended TBox is equivalent to the original one. 

First assume that \Tmc is Boolean and treeifiable. A \emph{role path between variables $x$ and $y$ in $q$} is a sequence of role names $r_1 \cdots r_n$ such that for distinct variables $x=x_0,\dots,x_n=y$, $q$ contains the atoms $r_1(x_1,x_2),\dots,r_n(x_{n-1},x_n)$. If $q$ is treeifiable, then there are only polynomially many role paths in $q^\mn{tree}$: the paths that occur in $q^\mn{tree}$, the least constrained treeification of $q$ defined in Section~\ref{sec:prelims}.  Let $\Cmc_q$ denote the set of \ELI-concepts of the form $\exists r_n^- . \cdots . \exists r_1^- . C$ where $r_1 \cdots r_n$ is a (potentially empty) role path in $q^\mn{tree}$ and $C$ is $\top$ or a concept name from $q$ or a CQ from $\mn{trees}(q)$ viewed as an \EL-concept. Extend \Tmc with $C \sqsubseteq A_C$, $A_C$ a fresh concept name, for all $C \in \Cmc_q$. Finally, normalize again. Clearly, the number of concept inclusions added to $\Tmc$ is polynomial in $|q|$ and the resulting OMQ is equivalent to the original one.

Now assume that \Tmc is not Boolean and treeifiable. Then unintended homomorphisms are ruled out automatically. To prepare for the complexity analysis of the meta problems carried out in Section~\ref{sec:decidability}, however, we still carry out the same modification that we have also used in Section~\ref{sec:AC0NL}: For every $p \in \mn{trees}(q)$, view $p$ as an $\EL$-concept $C$ and extend $\Tmc$ with $C \sqsubseteq A_C$, $A_C$ a fresh concept name.

The following lemma captures the use of the concepts $\Cmc_q$ later on and
gives an idea of why we use this particular set of concepts.
%
\begin{lemma} \label{lem:homtransfer} Let $q \in \text{conCQ}$ be Boolean and treeifiable, $\Imc_1,\Imc_2$ tree-shaped interpretations, and $d_i \in \Delta^{\Imc_i}$ for $i \in \{1,2\}$ such that $d_1 \in C^{\Imc_1}$ implies $d_2 \in C^{\Imc_2}$ for all $C \in \Cmc_q$. If there is a homomorphism from $q$ to 
$\Imc_1$ with $d_1$ in its range, then there is a homomorphism from $q$ to
$\Imc_2$ with $d_2$ in its range.
\end{lemma}
\begin{proof}
Assume that there is a homomorphism $h_1$ from $q$ to $\Imc_1$ with $d_1$ in its range. Since $\Imc_1$ is tree-shaped, this homomorphism factors into $h_1 = g_1 \circ h_q$, where $h_q$ is the obvious homomorphism from $q$ to $q^\mn{tree}$ and $g_1$ is a homomorphism from $q^\mn{tree}$ to $\Imc_1$. It clearly suffices to
show that there is a homomorphism $g_2$ from $q^\mn{tree}$ to $\Imc_2$ with $d_2$ in its range. 

Let $X_0 = g_1^{-1}(d_1)$. In a first step, we set $g_2(x) = d_2$ for all $x \in X_0$ and extend $g_2$ upwards as follows. Whenever $g_2(y)$ is already defined and there is an atom $r(x,y)$ in $q^\mn{tree}$, define $g_2(x)$ to be the (unique) predecessor of $g_2(y)$ in~$\Imc_2$.  We show that $g_2$ is a homomorphism from $q^{\mn{tree}}|_{\mn{dom}(g_2)}$ to $\Imc_2$, $\mn{dom}(g_2)$ the domain of $g_2$. If $r(x,y) \in q^\mn{tree}$ with $g_2(x),g_2(y)$ defined, then $q^{\mn{tree}}$ contains a role path $r_1 \cdots r_n$ from $x_1$ to $x_n \in X_0$ with $r_1=r$. Thus, there is a concept $C=\exists r_n^- \ldots \exists r_1^-. \top \in \Cmc_q$ such that $d_1 \in C^{\Imc_1}$. It follows that $d_2 \in C^{\Imc_2}$ and since $\Imc_2$ is tree-shaped and by construction of $g_2$, this yields $(g_2(x),g_2(y)) \in r^{\Imc_2}$. The argument for atoms $A(x) \in q^\mn{tree}$ where $g_2$ has been defined on $x$ is similar, using concepts of the form $C=\exists r_n^- \ldots \exists r_1^-. A \in \Cmc_q$.

In a second step, we define $g_2$ on all the remaining variables. Whenever $g_2(z)$ is still undefined for some $z \in \mn{var}(q^\mn{tree})$, there must be some $r(x,y) \in q^\mn{tree}$ such that $g_2(x)$ is already defined, $g_2(y)$ is not yet defined, and $z$ is in $q^{\mn{tree}}|_{\mn{reach}(x,y)}$. Since $q$ is connected, there must be a (potentially empty) role path $r_1 \cdots r_n$ in $q^\mn{tree}$ from $x$ to a variable $x_0 \in X_0$. Thus, $\Cmc_q$ contains $C = \exists r_n^-. \cdots \exists r_1^-. D \in \Cmc_q$ where $D$ is the \EL concept that corresponds to $q^\mn{tree}|_{\mn{reach}(x,y)}$ and since $g_1(x_0) = d_1$, we have $d_1 \in C^{\Imc_1}$. Consequently, $d_2 \in C^{\Imc_2}$. Since $\Imc_2$ is tree-shaped, this implies $g_2(x) \in D^{\Imc_2}$ and thus there is a homomorphism from $q|_{\mn{reach}(x,y)}$ to $\Imc_2$ that maps $x$ to $g_2(x)$. We use this homomorphism to extend $g_2$ to all variables in $\mn{reach}(x,y)$. 
\end{proof}

If $\Amc$ is a pseudo tree-shaped ABox and $b \in \mn{ind}(\Amc)$ has distance at least $n$ from the core, we define the \emph{ancestor path of $b$ up to length $n$} to be the unique sequence $r_1 r_2 \ldots r_n$ of role names such that $r_1(b_1,b_2),r_2(b_2,b_3),\ldots r_n(b_n,b) \in \Amc$.

\begin{definition} \label{def:psa}
Let $Q = (\Tmc,\Sigma, q) \in (\mathcal{EL}, \textnormal{conCQ})$. We say that $Q$ has {\em
  the ability to simulate \PSA} if there exist
\begin{itemize}
\item $\Tmc$-types $t_0 \subsetneq t_1$,
\item a pseudo tree-shaped $\Sigma$-ABox $\Amc$ of core size $|q|$,
\item a tuple $\abf$ from the core of $\Amc$ of length $\mn{ar}(q)$,
\item a tree $\Amc_i$ in $\Amc$ with three distinguished non-core individuals $b, c$ and $d$ from $\mn{ind}(\Amc_i)$ where $c$ and $d$ are incomparable descendants of $b$ and such that $b$ has distance more than $|q|$ from the core and the individuals $b$, $c$ and $d$ have pairwise distance more than $|q|$ from each other
\end{itemize} such that
\begin{enumerate} 
\item $\Amc, \Tmc \models q(\abf)$;
\item $t_1 = \mn{tp}_{\Amc, \Tmc} (b) = \mn{tp}_{\Amc, \Tmc} (c) = \mn{tp}_{\Amc, \Tmc} (d)$;
\item $\Amc_b \cup t_0(b), \Tmc \not\models q(\abf)$;
\item $\mn{tp}_{\Amc_c \cup t_0(c), \Tmc}(b) = \mn{tp}_{\Amc_d \cup t_0(d), \Tmc}(b) = t_0$,
\item if $q$ is Boolean, then every homomorphism from $q$ to $\Umc_{\Amc, \Tmc}$ is core close and
\item if $q$ is Boolean, then $b$, $c$ and $d$ have the same ancestor path up to length~$|q|$.
\end{enumerate}
We define $\Amc_\mn{target} := \Amc_b$, $\Amc_\wedge := \Amc^b_{cd}$ and $\Amc_\mn{source} := \Amc^c$.
\end{definition}
With $c$ and $d$ being `incomparable' descendants of $b$, we mean that
neither $d$ is a descendant of $c$ nor vice versa.

To understand the essence of Definition~\ref{def:psa}, it is worthwhile to consider the special case where $q$ is an AQ~$A(x)$.  In this case, $Q$ has the ability to simulate \PSA if there is a tree-shaped $\Sigma$-ABox \Amc with root $a=\abf$, three distinguished non-root individuals $b,c,d \in \mn{ind}(\Amc)$, $c$ and $d$ incomparable descendants of $b$, and \Tmc-types $t_0 \subsetneq t_1$ such that Conditions~(1)-(4) of Definition~\ref{def:psa} are satisfied. Figure~\ref{fig:psa-example} shows an ABox witnessing the ability to simulate PSA for such an OMQ. All remaining parts of Definition~\ref{def:psa} should be thought of as technical complications induced by replacing AQs with CQs.
\begin{figure}[t]
\begin{boxedminipage}{\columnwidth}
\centering
\begin{tikzpicture}[->,>=stealth',level/.style={sibling distance = 12cm/#1,
  level distance = 1.3cm}, scale=0.5]
\tikzstyle{node}=[shape=circle, draw,inner sep=2.0pt, fill=black]
\node [node] [label=above:$a$] {}
    child{ node [node] [label=above:$b$] {}
        child{ node [node] {}
            child{ node [node] [label=above left:$c$] {} 
                child{ node [node] [label=left:$A$] {} edge from parent node[above left] {$r$}}
                child{ node [node] [label=left:$A$] {} edge from parent node[above right] {$s$}}
                edge from parent node[above left] {$r$}}
            child{ node [node] [label=above right:$d$] {}
                child{ node [node] [label=right:$A$] {} edge from parent node[above left] {$r$}}
                child{ node [node] [label=right:$A$] {} edge from parent node[above right] {$s$}}
                edge from parent node[above right] {$s$}}
            edge from parent node[above left] {$r$}}
        child{ node [node] [label=right:$A$] {} 
            edge from parent node[above right] {$s$}}          
        edge from parent node[above left] {$r$}}
    child{ node [node] [label=right:$A$] {} 
        edge from parent node[above right] {$s$}
    }
; 
\end{tikzpicture}
\end{boxedminipage}
\caption{A witness ABox for the abilty to simulate PSA for the OMQ $Q = (\Tmc, \Sigma, A(x))$, where $\Tmc = \{\exists r.A \sqsubseteq B, \exists s.A \sqsubseteq C, B \sqcap C \sqsubseteq A\}$ and $\Sigma = \{r, s, A\}$.}
\label{fig:psa-example}
\vspace*{-4mm}
\end{figure}
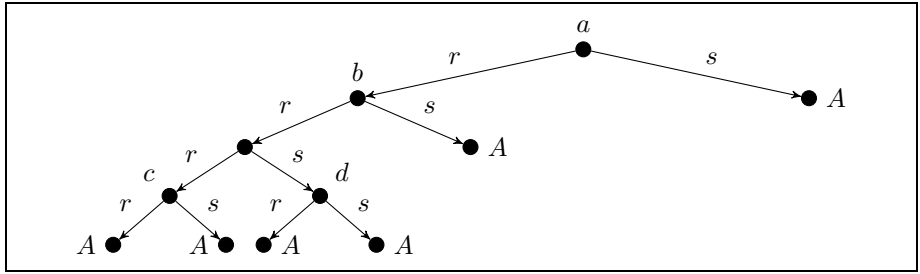
As a preliminary for showing that unbounded branching implies the ability to simulate \PSA, we give the following combinatorial lemma.
\begin{lemma} \label{lem:treecol} Let $T$ be a full binary tree of depth $n \cdot k \cdot d$ whose nodes are colored with $n$ colors, $k \geq 0$ and $n,d \geq 1$. Then $T$ has as a minor a monochromatic full binary tree of depth $k$ such that any two distinct nodes of the minor have distance at least $d$ from each other in $T$.\end{lemma}
\begin{proof}
  Let $T$ be a full binary tree of depth $k$ whose nodes are colored with $n$ colors.
  We assoicate $T$ with a tuple $(m_1,\dots,m_n)$ by letting, for $1 \leq i \leq m$, $m_i$ be the minimum integer such that $T$ does not have the color $i$ monochromatic full binary tree of depth $m_i$ as a minor.
 We prove the following.
\\[2mm]
\textbf{Claim.} 
$\sum_{i=1}^n m_i \geq k + 1$.
\\[2mm]
We proof the claim by induction on $k$. For $k=0$, there is only one node, say of color $i$. Then clearly $\sum_{i=1}^n m_i = 1 \geq 1 = k + 1$.

Now assume that the claim holds for $k$ and consider a tree $T$ of depth $k+1$, with associated tuple $(m_1, \ldots, m_n)$.  Let $a$ be the root of $T$ and let the children of $a$ root the subtrees $T_1$ and $T_2$, $(m_1^j, \ldots, m_n^j)$ the tuple associated with $T_j$ for $j \in \{1, 2\}$. We distinguish two cases.

First assume that there exists a color $j$ such that $m_j^1 \neq m_j^2$. W.l.o.g.\ let $m_j^1 < m_j^2$. Then $m_j = \max\{ m_j^1,m_j^2 \} > m_j^1$ and $m_i \geq m_i^1$ for all $i \neq j$. By the claim, $\sum_{i=1}^n m^1_i \geq k +1$. It
follows that $\sum_{i=1}^n m_i \geq k +2$, as required.

Now assume that there is no such color $j$. Let $i_0$ be the color of $a$. From
$m_{i_0}^1 = m_{i_0}^2$, it follows that $m_{i_0} > m^1_{i_0}$ and thus we can 
proceed as before with $i_0$ in place of $j$. This finishes the proof of the claim.

\smallskip
The statement of the lemma now follows easily for $d=1$: Let $T$ be a full binary tree of depth $n \cdot k$ whose nodes are colored with $n$ different colors. 
If there is no full monochromatic binary tree of depth $k$ as a minor in~$T$, then $m_i \leq k$ for all colors~$i$, in contradiction to $\sum_{i=1}^n m_i \geq n \cdot k + 1$. 

Now consider the case where $d>1$. From the case $d=1$, $T$ contains as a minor a full monochromatic binary tree $T'$ of depth $d \cdot k$. To obtain the desired full monochromatic binary tree $T''$ of depth $k$ whose nodes have distance at least $d$ from each other, we choose appropriate nodes from $T'$. Recall that the nodes of $T'$ are $V= \{1,2\}^k$. Then $T''$ can be constructed by choosing the nodes $V \cap \{1^d, 2^d\}^*$. Clearly, $T''$ is as required.
\end{proof}

\begin{lemma}
\label{lem:bintree-psa}
Let $Q \in (\mathcal{EL}, \textnormal{conCQ})$. Then $Q$ has the ability to simulate
\PSA
iff $Q$ is unboundedly branching.
\end{lemma}

\begin{proof}
  ``$\Rightarrow$''. Assume that $Q=(\Tmc,\Sigma,q) \in (\mathcal{EL}, \textnormal{conCQ})$ has the ability to simulate \PSA. Then there are
  $\Amc, \abf, b, c, d, t_0$, and $t_1$ as in Definition \ref{def:psa}.
  Let $k \geq 1$.
  We have to show that there is a pseudo tree-shaped $\Sigma$-ABox $\Amc \in \Mmc_Q$ that has a tree $\Amc_i$ that has the full binary tree of depth $k$ as
  a minor. We start with constructing an ABox $\Amc_0$ built up from the following set of ABoxes:
\begin{itemize}
\item one copy of $\Amc_{\mn{target}}$;
\item for every $w \in S$, one copy $\Amc_{\wedge, w}$ of $\Amc_{\wedge}$;
\item for every $w \in \{0,1\}^k$, one copy $\Amc_{\mn{source},w}$ of $\Amc_{\mn{source}}$.
\end{itemize}
We identify the individual $b$ of $\Amc_{\mn{target}}$ with the individual $b$ of $\Amc_{\wedge, \varepsilon}$. For every $w0 \in \{0,1\}^{k-1}$, we identify the individual $b$ of $\Amc_{\wedge,w0}$ with the individual $c$ of $\Amc_{\wedge,w}$ and for every $w1 \in \{0,1\}^{k-1}$, we identify the individual $b$ of $\Amc_{\wedge,w1}$ with the individual $d$ of $\Amc_{\wedge,w}$. Finally, for every $w0 \in \{0,1\}^{k}$, we identify the individual $b$ of $\Amc_{\mn{source},w0}$ with the individual $c$ of $\Amc_{\wedge,w}$ and for every $w1 \in \{0,1\}^{k}$, we identify the individual $b$ of $\Amc_{\mn{source},w1}$ with the individual $d$ of $\Amc_{\wedge,w}$. Since all $\Amc_\mn{\wedge}$ and $\Amc_\mn{source}$ are tree-shaped, the resulting ABox is $\Amc_0$ pseudo tree-shaped with the same core as $\Amc_\mn{target}$.

It is clear that $\Amc_0$ has the full binary tree of depth $k$ as a minor, formed by the set of roots of all $\Amc_{\wedge,w}$ and $\Amc_{\mn{source},w}$. From Conditions~1 and~2 from Definition~\ref{def:psa}, it follows that $\Amc_0 \models Q(\abf)$.  But there is no guarantee that $\Amc_0$ is minimal with this property, thus $\Amc_0$ need not be from $\Mmc_Q$. Let $\Amc_0,\dots\Amc_\ell$ be the sequence of ABoxes obtained by starting with $\Amc_0$ and exhaustively removing assertions such that $\Amc_i \models Q(\abf)$ still holds. We argue that the resulting ABox still has the full binary tree of depth $k$ as a minor.

It suffices to show that role assertions connecting two individuals that lie on the same path from the core to a root of a $\Amc_{\mn{source},w}$ are never removed. Assume
to the contrary that such a role assertion is removed when transitioning from $\Amc_i$ to $\Amc_{i+1}$. We distinguish three cases:
\begin{itemize}

\item
  The removed role assertion lies in $\Amc_{\wedge,w}$ on the path from $b$ to $c$. Then $\mn{tp}_{\Amc_{i+1},\Tmc}(b) \subseteq \mn{tp}_{\Amc^b_c \cup t_0(c),\Tmc}(b)$.
By Condition~4 from Definition~\ref{def:psa}, the latter type is $t_0$. By iteratively using Conditions~3 and~4, it follows that $\mn{tp}_{\Amc_{i+1},\Tmc}(b) = t_0$, with $b$ the individual from the copy of $\Amc_\mn{target}$. With Conditions~2 and~5, it follows that $\Amc_k \not \models Q(\abf)$. Contradiction.

\item The removed role assertion lies in $\Amc_{\wedge,w}$ on the path from its $b$ to its $d$. The proof is analogous.

\item The removed role assertion lies in $\Amc_{\mn{target}}$ on the path from the core to $b$. It follows that $\mn{tp}_{\Amc_{i+1}, \Tmc}(a) \subseteq \mn{tp}_{\Amc_b \cup t_0(b), \Tmc}(a)$ for every individual $a$ in the copy of $\Amc_{\mn{target}}$. With
Condition~3, it again follows that $\Amc_k \not \models Q(\abf)$.

\end{itemize}
%

\smallskip

``$\Leftarrow$''. Assume that $Q=(\Tmc,\Sigma,q) \in (\mathcal{EL}, \textnormal{conCQ})$ is not boundedly branching. Let $\mn{TP}$ denote the set of all $\Tmc$-types and set $m = 2^{|\Tmc|}$. Clearly, $|\mn{TP}| \leq m$. Set $k = m \cdot 2^m \cdot |\Tmc|^{|q|} \cdot (2m^m+1) \cdot |q|$. Since $Q$ is not boundedly branching, we find a $\Sigma$-ABox $\Amc \in \Mmc_Q$ and a tuple $\abf$ from its core such that $\Amc, \Tmc \models q(\abf)$ and one the trees $\Amc_i$ of \Amc has the full binary tree of depth $k$ as a minor. We show that \Amc and \abf can serve as the ABox and tuple in Definition~\ref{def:psa}, that is, as a witness for $Q$ having the ability to simulate \PSA.

To identify the distinguished individuals $b,c,d$, we use a suitable coloring of the individuals of $\Amc_i$ and Lemma~\ref{lem:treecol}. In fact, we color every $b \in \mn{ind}(\Amc_i)$ with the color $(\mn{tp}_{\Amc, \Tmc}(b),S_b,r_1^b r_2^b \ldots r_{|q|}^b)$ where $\mn{TP} \supseteq S_b = \{ t \in S_b \mid \Amc_b \cup t(b),\Tmc \models q(\abf)\}$ and where $r_1^b r_2^b \ldots r_{|q|}^b$ is the ancestor path of $b$ up to length $|q|$. There are no more than $m \cdot 2^m \cdot |\Tmc|^{|q|}$ colors, so from Lemma~\ref{lem:treecol} we know that \Amc has as a minor a monochromatic full binary tree $T$ of depth $2m^m+1$ whose nodes have distance at least $|q|$ from each other. Let $b$ be a child of the root of $T$ (to make sure that $b$ has depth at least $|q|$ from the core) and $T' \subseteq T$ the subtree of $T$ rooted at $b$, so $T'$ is a full binary tree of depth $2m^m$.  We color every $c \in T'$ with the function $f_c : \mn{TP} \rightarrow \mn{TP}$ that is defined by $f_c(t) = \mn{tp}_{\Amc_c \cup t(c), \Tmc}(b)$. There are at most $m^m$ such functions, so again by Lemma~\ref{lem:treecol}, there will be the monochromatic binary tree of depth $2$ as a minor. In particular, we find two incomparable individuals $c$ and $d$ in $T'$ that are colored with the same function. We show that $\Amc$ we can find types $t_1$ and $t_0$ such that with the distinguished nodes $b, c, d$, \Amc and \abf satisfy Conditions~1-6 from Definition~\ref{def:psa}.

Condition~1 is true by choice of \Amc.  Set $t_1 := \mn{tp}_{\Amc, \Tmc}(b)$.  Then Condition~2 is satisfied because $b$, $c$ and $d$ were colored with the same color by the first coloring. For the same reason, Conditon~6 is fulfilled. Condition~5 follows from Lemma~\ref{lem:core-close}.

To find $t_0$, we define a sequence $t'_0,t'_1,\dots$ of $\Tmc$-types where $t'_0 = \emptyset$ and $t'_{i+1} = \mn{tp}_{\Amc_c \cup t'_i(c), \Tmc}(b)$. It is clear that $t'_i \subseteq t'_{i+1}$ for all $i$. Let $t_0$ be the limit of the sequence.  Since $c$ and $d$ were colored with the same function $f_c=f_d$, Condition~4 holds. It thus remains to argue that Condition~3 holds. 

We show by induction on $i$ that $\Amc_c \cup t'_i(c), \Tmc \not\models q(\abf)$ for all $i \geq 0$. It is clear that $\Amc_c \cup t'_0(c), \Tmc \not\models q(\abf)$ since \Amc is minimal with $\Amc,\Tmc \models q(\abf)$. Now assume that $\Amc_c \cup t'_i(c), \Tmc \not\models q(\abf)$ for some $i$. Then $\Amc_c \cup t'_{i+1}(b), \Tmc \not\models q(\abf)$. Since $S_b = S_c$ has been assured by the
first coloring, we obtain $\Amc_c \cup t'_{i+1}(c), \Tmc \not\models q(\abf)$ which completes the induction. 

Thus $\Amc_c \cup t_0(c), \Tmc \not\models q(\abf)$ and using again that $S_b = S_c$, we obtain Condition~3.
\end{proof}
It remains to show that the ability to simulate \PSA implies \PTime-hardness.

\begin{lemma}
\label{lem:psa-ptimehard}
If $Q \in (\EL,\textnormal{conCQ})$ has the ability to simulate \PSA, then {\sc eval}$(Q)$ is $\PTime$-hard under FO reductions.
\end{lemma}

\begin{proof}
Let $Q=(\Tmc,\Sigma,q) \in (\mathcal{EL},\textnormal{conCQ})$ have the ability to simulate \PSA. Then there is a pseudo tree-shaped $\Sigma$-ABox \Amc, a tuple $\abf$ in its core, distinguished individuals $b$, $c$ and $d$, and types $t_0 \subsetneq t_1$ as in Definition~\ref{def:psa}. We reduce \PSA to {\sc eval}$(Q)$. 

Let $G=(V,E,S,t)$ be an input for \PSA. We construct a $\Sigma$-ABox $\Amc_G$ that represents~$G$. Reserve an individual $a_v$ for every node $v \in V$. For every $(u,v,w) \in E$, include in $\Amc_G$ a copy $\Amc_{u,v,w}$ of $\Amc_\wedge$ that uses fresh individuals, identifying (the individual that corresponds to) $c$ with $a_u$, $d$ with $a_v$ and $b$ with $a_w$. For every $s \in S$ include in $\Amc_G$ one copy of $\Amc_{\mn{source}}$ that uses fresh individuals, identifying $c$ with $a_s$.
Finally, include in $\Amc_G$ the ABox $\Amc_{\mn{target}}$ identify $b$ with $a_t$. (Note that we do not use a `copy' of $\Amc_\mn{target}$, so individuals from $\Amc_\mn{target}$ except for $b$ retain their name.) It can be verified that $\Amc_G$ can be constructed from $G$ using an FO-query. 
It thus remains to show the following.
\\[2mm]
\textbf{Claim 1.} $t$ is accessible in $G$ iff $\Amc_G \models Q(\abf)$.
\\[2mm]
For the ``$\Rightarrow$'' direction, assume that $t$ is accessible in $G$. Define a sequence $S=S_0 \subseteq S_1 \subseteq \cdots
\subseteq V$ by setting 
\[
S_{i+1}= S_i \cup \{ w \in V \mid \text{there is a $(u,v,w) \in E$
such that } u,v \in S_i \}
\]
and let the sequence stabilize at $S_n$. Clearly, the elements of $S_n$ are exactly the accessible nodes. It can be shown by induction on $i$ that whenever $v \in S_i$, then $t_1 \subseteq \mn{tp}_{\Amc_G,\Tmc}(a_v)$. In fact, the induction start follows from $t_1 = \mn{tp}_{\Amc, \Tmc} (b)$ and the induction step from Condition~2 of Definition~\ref{def:psa}. It follows from Conditions~1 and~2 that $\Amc_\mn{target} \cup t_1(b) \models Q(\abf)$, thus $\Amc_G \models Q(\abf)$ as required.

\smallskip

The ``$\Leftarrow$'' direction is more laborious. Assume that $t$ is not accessible in $G$. Set
\[
\begin{array}{rcl}\Amc_G' &:=& \Amc_G \cup \{t_0(a_v) \mid v \in V \text{ is not accessible}\} \\ & & \cup \, \{t_1(a_v) \mid v \in V \text{ is accessible}\}.\end{array}
\]
We show that $\Amc'_G \not\models Q(\abf)$, which implies $\Amc_G \not\models Q(\abf)$.

\smallskip

We have defined $\Amc'_G$ as an extension of $\Amc_G$. Alternatively and more suitable for what we want to prove, $\Amc'_G$ can be obtained by starting with
an ABox $\Amc_0$ that contains only the assertions $t_0(a_v)$ for all inaccessible nodes $v \in V$ as well as $t_1(a_v)$ for all accessible nodes $v \in V$, and then
exhaustively applying the following rules in an unspecified order, obtaining a sequence of ABoxes $\Amc_0,\Amc_1,\dots,\Amc_m$ with $\Amc_m = \Amc'_G$:
\begin{enumerate} 

\item Choose a triple $(u,v,w) \in E$ that has not been chosen before, take a copy $\Amc_\wedge^{u,v,w}$ of $\Amc_\wedge$ using fresh individual names, with $c$ renamed to $a_u$, $d$ to $a_v$, and $b$ to $a_w$, and add the assertions $t_{\mn{acc}(x)}(a_x)$ for $x \in \{u,v,w\}$ where $\mn{acc}(x) = 1$ if $x$ is accessible and $\mn{acc}(x) = 0$ otherwise. Set $\Amc{i+1}=\Amc_i \cup \Amc_\wedge^{u,v,w}$.

\item Choose a node $s \in S$ that has not been chosen before, introduce a copy $\Amc_{\mn{source}}^s$ of $\Amc_{\mn{source}}$ that uses fresh individuals, with $c$ renamed to $a_s$, and add the assertions $t_1(a_s)$. Let the resulting ABox be called $\Amc_{\mn{source}}^s$. Set $\Amc_{i+1}=\Amc_i \cup \Amc_{\mn{source}}^s$.  

\item Set $\Amc_{i+1} = \Amc_i \cup \Amc'_{\mn{target}}$, where $\Amc'_{\mn{target}}$ is obtained from $\Amc_\mn{target}$ by renaming $b$ to $a_t$ and adding the assertions $t_0(a_t)$.

\end{enumerate}
Clearly, rule application terminates after finitely many steps and results in the
ABox $\Amc_G'$. Note that we add assertions $t_i(a)$, $i \in \{0,1\}$ to the ABoxes constructed in the rules to enable application of Lemma~\ref{lem:aboxunion}.
\\[2mm]
{\bf Claim 2.} $\mn{tp}_{\Amc_i, \Tmc}(a_u) = t_0$ if $u \in V$ is inaccessible and $\mn{tp}_{\Amc_i, \Tmc}(a_u) = t_1$ otherwise, for all $i \geq 0$.
\\[2mm]
The proof is by induction on $i$.  For $i=0$, the statement is clear since $t_0$ and $t_1$ are $\Tmc$-types. Now assume the statement is true for some $i$ and consider $\Amc_{i+1}$. If $\Amc_{i+1}$ was obtained by Rule~1, it can be verified using Conditions~2 and~4 from Definition~\ref{def:psa} that $\mn{tp}_{\Amc_\wedge^{u,v,w}, \Tmc}(a_x) = t_{\mn{acc}(x)}$ for all $x \in \{u,v,w\}$. So with Lemma~\ref{lem:aboxunion} and since $\Amc_i$ and $\Amc_\wedge^{u,v,w}$ share only the individuals $u,v,w$, the statement follows. If $\Amc_{i+1}$ was obtained by Rule~2, we can use Condition~2 and Lemma~\ref{lem:aboxunion}. If $\Amc_{i+1}$ was obtained by Rule~3, using $\mn{acc}(t)=0$ it can be verified that $\mn{tp}_{\Amc'_{\mn{target}}, \Tmc}(a_t) = t_0$ and with Lemma~\ref{lem:aboxunion}, the statement follows. This finishes the proof of the Claim~2.

\smallskip

It remains to show that $\Amc'_G,\Tmc \not\models q(\abf)$. Assume to the contrary that $\Amc'_G,\Tmc \models q(\abf)$, that is, there is a homomorphism $h$ from $q(\xbf)$ to $\Umc_{\Amc'_G,\Tmc}$ such that \mbox{$h(\xbf)=\abf$}. There can be at most one individual of the form $a_v$ in the range of $h$ by construction of $\Amc'_G$ and since $b, c, d$ have distance more than $|q|$ from each other in \Amc.  

If there is no individual $a_v$ in the range of $h$, then $h$ only hits individuals from a single copy of $\Amc_\mn{source}$, $\Amc_\wedge$, or $\Amc'_\mn{target}$ as well as anonymous elements in the trees below them (since $q$ is connected).  First assume that this is $\Amc'_\mn{target}$. By Claim~2 and since $\mn{acc}(t)=0$, $\mn{tp}_{\Amc'_G, \Tmc}(a_t) = t_0$. It can be shown that the identity function is a homomorphism from $\Umc_{\Amc'_G,\Tmc}|_\Delta$, $\Delta$ the individuals from $\Amc'_\mn{target}$ and anonymous elements below them, to $\Umc_{\Amc_\mn{target} \cup t_0(b),\Tmc}$. By composing homomorphisms, it follows that $\Amc_\mn{target} \cup t_0(b),\Tmc \models q(\abf)$, contradicting Condition~3.

Now, assume that $h$ only hits individuals from a copy of $\Amc_\mn{source}$ or $\Amc_\wedge$ as well as anonymous elements in the trees below them. Then the restriction \Umc of $\Umc_{\Amc'_G,\Tmc}$ to the range of $h$ is tree-shaped.  Moreover, $q$ must be Boolean since the distance between the core and the elements of \Umc exceeds $|q|$ and $q$ is treeifiable because $h$ is a homomorphism to a tree-shaped interpretation. Since $\Umc_{\Amc'_G,\Tmc}$ is a model of \Tmc, we have $d \in C_q^{\Umc_{\Amc'_G,\Tmc}}$ for the root $d$ of \Umc, where $C_q$ is $q^\mn{tree}$ seen as an $\EL$-concept.  From Claim~2 and Lemma~\ref{lem:aboxunion}, it follows that some element of $\Umc_{\Amc_\mn{source} \cup t_1(c),\Tmc}$ or of $\Umc_{\Amc_\wedge \cup t_1(b)\cup t_1(c)\cup t_1(d),\Tmc}$ also satisfies $C_q$. By Condition~2 and Lemma~\ref{lem:aboxunion}, the same is true for an element from $\Umc_{\Amc,\Tmc}$ that is `below' $b$ (reachable from $b$ by a directed path). Since the distance from the core to $b$ in \Amc exceeds $|q|$, this homomorphism is not core close, contradicting Condition~5.

Next, assume that the range of $h$ contains $a_v$. Then $q$ is Boolean since it is connected and the distance between the core in $\Amc'_\mn{target}$ and $a_v$ exceeds $|q|$. Let \Umc be the restriction of $\Umc_{\Amc'_G,\Tmc}$ to all elements within distance at most $|q|$ from~$a_v$. Then \Umc is almost tree-shaped except that $a_v$ can have multiple predecessors.  By Condition~6, however, there is a unique sequence of roles $r_n \cdots r_1$ such that in each path $d_m s_m \cdots d_2 s_2 d_1 s_1 a_v$ in \Umc, $s_m \cdots s_1$ is a postfix of $r_n \cdots r_1$. We can thus obtain a tree-shaped interpretation $\Umc'$ from \Umc by exhaustively identifying elements $d_1,d_2$ whenever $(d_1,e),(d_2,e) \in r^\Imc$ for some $e$ and $r$. 
Clearly, we can find a homomorphism $h'$ from $q$ to $\Umc'$. Consequently, $q$ is treeifiable and the TBox has been extended with $C \sqsubseteq A_C$ for all $C \in \Cmc_q$.
%
\\[1mm]
{\bf Claim 3}. 
  $a_v \in C^{\Umc'}$ iff $a_v \in C^\Umc$ for all $C \in \Cmc_q$.
\\[1mm]
The ``$\Leftarrow$'' direction is immediate. For ``$\Rightarrow$'', assume to the contrary of what we aim to show that that $a_v \in C^{\Umc'}$ but $a_v \notin C^\Umc$ for some $C \in \Cmc_q$. Then $C$ has the form $\exists r^-_n . \cdots \exists r^-_1 . C$ where $r_1 \cdots r_n$ is a (potentially empty) role path in $q^{\mn{tree}}$ and $C$ is $\top$ or a concept name from $q$ or a CQ from $\mn{trees}(q)$ viewed as an \EL-concept. In the former two cases, we clearly have $a_v \in C^\Umc$ by construction of $\Umc'$.  In the latter case, $C$ is of the form $\exists r . D$. Since $a_v \in C^{\Umc'}$, there is a path $d_1 r_1 d_2 \cdots r_{n-1} d_{n-1} r_n a_v$ in $\Umc'$ and $d_1 \in (\exists r .D)^{\Umc'}$. If there is an $e \neq d_2$ with $(d_1,e) \in r^{\Umc'}$ and $e \in D^{\Umc'}$, then again $a_v \in C^\Umc$ by construction of~$\Umc'$. Assume that this is not the case, that is, $d_1 \in (\exists r .D)^{\Umc'}$ is true only because $r_1=r$ and $d_2 \in D^{\Umc'}$. We may assume w.l.o.g.\ that $C$ was chosen so that $n$ is minimal, that is, there is no concept $C' \in \Cmc_q$ with a shorter existential prefix than $C$ such that $a_v \in {C'}^{\Umc'}$ but $a_v \notin {C'}^\Umc$. Let $D=A_1 \sqcap \cdots \sqcap A_{n_1} \sqcap \exists s_1 . E_1 \cdots \sqcap \exists s_{n_2} . E_{n_2}$.  Then $\exists r^-_n . \cdots \exists r^-_2 . A_i$ and $\exists r^-_n . \cdots \exists r^-_2 . \exists s_j .E_j$ are also in $\Cmc_q$ for all relevant $i$ and $j$. Let $\Gamma$ be the set of all these concepts. We have $a_v \in G^{\Umc'}$ for all $G \in \Gamma$ and, since $n$ is minimal, $a_v \in G^\Umc$ for all $G \in \Gamma$. By Claim~2, we have $\mn{tp}_{\Amc'_G, \Tmc}(a_v) \in \{t_0, t_1\}$. Let $\Bmc = \Amc$ if $\mn{tp}_{\Amc'_G, \Tmc}(a_v) =t_1$ and $\Bmc = \Amc_c$ if $\mn{tp}_{\Amc'_G, \Tmc}(a_v) = t_0$. Then it follows from $\mn{tp}_{\Bmc, \Tmc}(b) =\mn{tp}_{\Amc'_G, \Tmc}(a_v)$ that $a_v \in A_C^{\Umc_{\Amc'_G,\Tmc}}$ iff $b \in A_C^{\Umc_{\Bmc,\Tmc}}$ for all $C \in \Cmc_q$. Since both $\Umc_{\Amc'_G,\Tmc}$ and $\Umc_{\Bmc,\Tmc}$ are universal models and by construction of \Tmc, $a_v \in C^{\Umc_{\Amc'_G,\Tmc}}$ iff $b \in C^{\Umc_{\Bmc,\Tmc}}$ for all $C \in \Cmc_q$. By choice of \Umc, the same is true when $\Umc_{\Amc'_G,\Tmc}$ is replaced with \Umc.
 Thus, $b$ satisfies all concepts from $G$ as well as $\exists r^-_n . \cdots \exists r^-_1 . \top$ in $\Umc_{\Bmc,\Tmc}$. Since $\Umc_{\Bmc,\Tmc}$ is tree-shaped, $b \in C^{\Umc_{\Bmc,\Tmc}}$ and, consequently, $a_v \in C^\Umc$ as desired. This finishes the proof of Claim~3.

\smallskip

By Claims~2 and~3 and since both $\Umc_{\Amc'_G,\Tmc}$ and $\Umc_{\Amc,\Tmc}$ are universal models and by construction of \Tmc, $a_v \in C^{\Umc'}$ implies $b \in C^{\Umc_{\Amc,\Tmc}}$ for all $C \in \Cmc_q$. The same is true if we replace $\Umc_{\Amc,\Tmc}$ with its restriction $\Umc''$ to all elements that have distance at most $|q|$ from $b$. Note that $\Umc''$ is tree-shaped. We can apply Lemma~\ref{lem:homtransfer} with $\Umc',a_v$ in place of $\Imc_1,d_1$ and $\Umc'',b$ in place of~$\Imc_2,d_2$, obtaining a homomorphism from $q$ to $\Umc_{\Amc,\Tmc}$ with $b$ in its range. Since the distance from the core to $b$ in \Amc exceeds $|q|$, this homomorphism is not core close, contradicting Condition~5.
\end{proof}

\subsection{Bounded Pathwidth Implies Linear Datalog Rewritability}
\label{sec:boundedpathwidth}

We prove the equivalence (i) $\Leftrightarrow$ (ii) from Theorem \ref{thm:NLPTime}.
Our proof works even for OMQs from $(\ELI, \textnormal{CQ})$, that is, when inverse role are admitted in the TBox and when the conjunctive queries are not necessarily connected.  We thus establish our result for this more general class of OMQs right away.

\begin{lemma}
\label{lem:LDLog}
Let $Q=(\Tmc, \Sigma, q) \in (\ELI,\textnormal{CQ})$. Then $Q$ has bounded pathwidth if and only if $Q$ is rewritable into linear Datalog. In the positive case, there exists a linear Datalog program of width $\mn{pw}(Q)+\mn{ar}(q)$.
\end{lemma}
Before proving Lemma~\ref{lem:LDLog}, we give some additional preliminaries about Datalog. Let $\Pi$ be a Datalog program, $\Amc$ an ABox and $\abf$ a tuple from $\mn{ind}(\Amc)$. A \emph{derivation of $\Pi(\abf)$ in \Amc} is a labelled directed tree $(V,E,\ell)$ where
\begin{enumerate}

\item $\ell(x_0)=\mn{goal}(\abf)$ for $x_0$ the root node;

\item for each $x \in V$ with children $y_1,\dots,y_k$, $k >0$, there
  is a rule $S(\ybf) \leftarrow p(\xbf)$ in $\Pi$ and a substitution $\sigma$
  of variables by individuals from \Amc such that $\ell(x)=S(\sigma\ybf)$ and
  $\ell(y_1),\dots,\ell(y_k)$ are exactly the facts in $p(\sigma\xbf)$;

\item if $x$ is a leaf, then $\ell(x) \in \Amc$.

\end{enumerate}
%
It is well known that $\Amc \models
\Pi(\abf)$ iff there is a derivation of $\Pi(\abf)$ in $\Amc$.

We associate with each derivation $D=(V,E,\ell)$ of $\Pi(\abf)$ in \Amc an ABox~$\Amc_D$.
In fact, we first associate an instance $\Amc_x$ with every $x \in V$
and then set $\Amc_D := \Amc_{x_0}$ for $x_0$ the root of~$D$.  If $x
\in V$ is a leaf, then $\ell(x) \in \Amc$ and we set $\Amc_x=\{
\ell(x) \}$. If $x \in V$ has children $y_1,\dots,y_k$, $k>0$, such
that $y_1,\dots,y_\ell$ are non-leafs and
$y_{\ell+1},\dots,y_k$ are leafs, then $\Amc_x$ is
obtained by starting with the assertions from
$\ell(y_{\ell+1}),\dots,\ell(y_k)$ and then adding a copy of
$\Amc_{y_i}$, for $1 \leq i \leq \ell$, in which all individuals
except those in $\ell(x)$ are substituted with fresh individuals.

The following lemma is well known \cite{AbiteboulHV95} and easy to verify. 

\begin{lemma}
\label{lem:nicestructure}
 Let $\Pi$ be a linear Datalog program and let $D$ be a derivation of $\Pi(\abf)$ in \Amc and $\Pi$ of diameter $d$. Then
  \begin{enumerate}

  \item $\Amc_D \models \Pi(\abf)$;

  \item there is a homomorphism $h$ from $\Amc_D$ to $\Amc$ with $h(\abf)=\abf$;

  \item $\Amc_D$ has pathwidth at most $d$.

  \end{enumerate}
\end{lemma}

With this lemma, the ``$\Leftarrow$'' direction of Lemma \ref{lem:LDLog} is easy to prove.
%
%
Assume that $Q \in (\ELI,\textnormal{CQ})$ is rewritable into a linear Datalog program $\Pi$. We show that $\mn{pw}(Q)$ is at most the diameter $d$ of $\Pi$. Take any pair $(\Amc, \abf)$ such that $\Amc \models Q(\abf)$. Since $\Pi$ is a linear Datalog rewriting of $Q$, there exists a derivation of $\Pi(\abf)$ in $\Amc$. By Lemma~\ref{lem:nicestructure}, there exists an ABox $\Amc_D$ of pathwidth at most $d$ such that $\Amc_D \models Q(\abf)$ and a homomorphism from $\Amc_D$ to \Amc that is the identity on \abf. Hence, $Q$ has pathwidth at most $d$. 

\smallskip

The rest of this section takes care of the ``$\Rightarrow$'' direction of Lemma \ref{lem:LDLog}. Assume that $Q$ has bounded pathwidth, say $\mn{pw}(Q)=k$. We obtain a linear Datalog program in the following way: We encode pairs $(\Amc,\abf)$ of an ABox $\Amc$ of pathwidth at most $k$ and a tuple $\abf$ from $\Amc$ as words over a finite alphabet, where one symbol of the word encodes one bag of the path decomposition of $\Amc$. We then construct an alternating two-way automaton on finite words that accepts precisely those words that encode a pair such that $\Amc \models Q(\abf)$. Such an automaton can always be transformed into a deterministic one-way automaton that accepts the same language \cite{GeffertO14}. From the latter automaton, we then construct the linear Datalog program that is equivalent to $Q$.

\smallskip
\noindent
\textbf{Derivation trees for AQs.} It is well known that in $\ELI$, entailment of AQs can be characterized in terms of derivation trees. \cite{BieLuWo-IJCAI13} Let \Tmc be an \ELI-TBox in normal form, \Amc an ABox, $a_0 \in \mn{ind}(\Amc)$ and $A_0 \in \NC$. A \emph{derivation tree} for $A_0(a_0)$ is a finite $\mn{ind}(\Amc) \times \NC$-labeled tree $(T,\ell)$ such that
\begin{itemize}
\item $\ell(\varepsilon) = (a_0,A_0)$;
\item if $\ell(x) = (a,A)$ and either $A(a) \in \Amc$ or $\top \sqsubseteq A \in \Tmc$, then $x$ is a leaf;
\item if $\ell(x) = (a,A)$ and neither $A(a) \in \Amc$ or $\top \sqsubseteq A \in \Tmc$, then one of the following holds:
\begin{itemize}
\item $x$ has successors $y_1,\ldots,y_n$ with $n \geq 1$ and $\ell(y_i) = (a,A_i)$ such that $\Tmc \models A_1 \sqcap \ldots \sqcap A_n \sqsubseteq A$;
\item $x$ has a single successor $y$ with $\ell(y) = (b,B)$ and there is $\exists r.B \sqsubseteq A \in \Tmc$ and $r(a,b) \in \Amc$, where $r$ is a (possibly inverse) role.
\end{itemize}
\end{itemize}

For proving the correctness of the constructed automaton later on, we need the following lemma, which is a special case of Lemma~29 in \cite{BieLuWo-IJCAI13}.

\begin{lemma}
\label{lem:AQderivation}
Let \Tmc be an \ELI-TBox in normal form, \Amc an ABox, $a \in \mn{ind}(\Amc)$ and $A \in \NC$. Then $\Tmc, \Amc \models A(a)$ if and only if there exists a derivation tree for $A(a)$.
\end{lemma}

\smallskip
\noindent
\textbf{Two way alternating finite state automata.} We introduce \emph{two way alternating finite state automata (2AFAs)}. For any set $X$, let $\Bmc^+(X)$ denote the set of all positive Boolean formulas over $X$, i.e., formulas built using conjunction and disjunction over the elements of $X$ used as propositional variables, and where the special formulas $\mn{true}$ and $\mn{false}$ are admitted
as well. A 2AFA is a tuple $\Amf = (S, \Gamma, \delta, s_0)$, where $S$ is a finite set of \emph{states}, $\Gamma$ a finite alphabet, $\delta : S \times (\Gamma \cup \{\vdash, \dashv\} \rightarrow \Bmc^+(\{\mn{left},\mn{right},\mn{stay}\} \times S)$ the \emph{transition function} and $s_0 \in S$ the \emph{initial state}. The two symbols $\vdash$ and $\dashv$ are used as the left end marker and right end marker, respectively, and it is required that $\delta(s,\vdash) \in \Bmc^+(\{\mn{right}\}\times S)$ and $\delta(s,\dashv) \in \Bmc^+(\{\mn{left}\}\times S)$ for all $s \in S$ so that the 2AFA can never leave the space of the input word.

For an input word $w=w_1\ldots w_n \in \Gamma^n$, define $w_0 = \; \vdash$ and $w_{n+1} = \; \dashv$. A \emph{configuration} is a pair $(i,s) \in \{0,\ldots,n+1\} \times S$. An \emph{accepting run} of a 2AFA $\Amf = (S, \Gamma, \delta, s_0)$ on $w$ is a pair $(T,r)$ that consists of a finite tree $T$ and a labeling $r$ that assigns a configuration to every node in $T$ such that
\begin{enumerate}
\item $r(\varepsilon) = (1,s_0)$, where $\varepsilon$ is the root of $T$ and
\item if $m \in T$, $r(m)=(i,s)$, and $\delta(s,w_i)=\vp$, then
    there is a (possibly empty) set $V \subseteq \{\mn{left},\mn{right},\mn{stay}\} \times S$ such
    that $V$ (viewed as a propositional valuation) satisfies $\vp$ and for
    every $(\mn{left},s') \in V$ there is a successor of $m$ in $T$ labeled with $(i-1,s')$, for every $(\mn{right},s') \in V$ there is a successor of $m$ in $T$ labeled with $(i+1,s')$ and for every $(\mn{stay},s') \in V$ there is a successor of $m$ in $T$ labeled with $(i,s')$.
\end{enumerate}
The \emph{language accepted by a 2AFA} $\Amf$, denoted by $L(\Amf)$, is the set of all words $w \in \Gamma^*$ such that there is an accepting run of $\Amf$ on $w$. Note that there is no set of final states, acceptance is implicit via the transition function $\delta$ by using the formulas $\mn{true}$ and $\mn{false}$. In particular, if there is a leaf labeled $(i,s)$ in an accepting run, then $\delta(s,w_i) = \mn{true}$.

\smallskip
\noindent
\textbf{Construction of the 2AFA.} Let $Q = (\Tmc, \Sigma, q) \in (\ELI,\textnormal{CQ})$ with \mbox{$\mn{pw}(Q)=k$}, and let $\xbf=x_1 \cdots x_{\mn{ar}(q)}$ be the answer variables in $q$. We encode pairs $(\Amc,\abf)$ with \Amc a $\Sigma$-ABox of pathwidth at most $k$ and $\abf \in \mn{ind}(\Amc)^{\mn{ar}(q)}$ as words over a suitable finite alphabet $\Gamma$. Reserve a set $\mn{N} \subseteq \NI$ of $2k+2$ individual names. Then $\Gamma$ consists of all tuples $(\bbf,\Bmc,\cbf, f)$, where
\begin{itemize}

\item $\Bmc$ is a $\Sigma$-ABox with $|\mn{ind}(\Bmc)| \leq k$ that uses only individual names from $\mn{N}$,

\item $\bbf$ and $\cbf$ are tuples over $\mn{ind}(\Bmc)$ of arity at most $k$, and

\item  $f$ is a partial function from $\xbf$ to $\mn{ind}(\Bmc)$. 

\end{itemize}
Let $(\Amc,\abf)$ be a pair as described above with $\abf=a_1 \cdots a_{\mn{ar}(q)}$ and let $V_1,\ldots,V_n$ be a $(j,k+1)$ path decomposition of $\Amc$. We encode $(\Amc, \abf)$ by a word $(\bbf_1,\Bmc_1,\cbf_1, f_1)$ $\cdots(\bbf_n,\Bmc_n,\cbf_n,f_n)$ from $\Gamma^n$, as follows:
\begin{itemize}

\item As $\Bmc_1$, we use a copy of $\Amc |_{V_1}$ that uses only individual names from $\mn{N}$.

\item For $1 < i \leq n$, $\Bmc_i$ is a copy of $\Amc |_{V_i}$ that uses the same individual names as $\Bmc_{i-1}$ on $\mn{ind}(\Amc |_{V_{i-1}}) \cap \mn{ind}(\Amc |_{V_i})$ and otherwise only individual names from $\mn{N} \setminus \mn{ind}(\Bmc_{i-1})$.  Since bags have size at most $k+1$, $|\mn{N}|=2k+2$ individual names suffice.
%

\item  $\bbf_1 = \cbf_n$ is the empty tuple.
  
\item For $1 < i \leq n$, $\bbf_{i-1} = \cbf_i$ is the tuple that contains every individual from $\mn{ind}(\Bmc_{i-1}) \cap \mn{ind}(\Bmc_i)$ exactly once, ascending in some fixed order on $\mn{N}$. 

\item For $1 \leq i \leq n$, $f_i$ is defined as follows.  If $V_i$ contains a copy $a'_i$ of $a_i$, then $f_i(x_i)=a'_i$; otherwise $f_i(x_i)$ is undefined.
  
\end{itemize}
It is easy to see that $(\Abf,\abf)$ can be recovered from $w$, and in particular \abf from $f$. Note that different words over $\Gamma$ might encode the same pair $(\Amc,\abf)$, for example because we can choose different path decompositions, and there are words over $\Gamma$ that do not properly encode a pair $(\Amc,\abf)$.
Neither of this is problematic for the remaining proof.

We now construct a 2AFA $\Amf$ that accepts a word that encode a pair $(\Amc, \abf)$ if and only if $\Amc \models Q(\abf)$. The idea is that an accepting run of the automaton has one main path on which it traverses the word from left to right, while guessing a homomorphism $h$ from $q$ to $\Umc_{\Amc, \Tmc}$ with $h(\xbf)=\abf$ in a stepwise fashion. The truth of all concept memberships in $\Umc_{\Amc, \Tmc}$ that are necessary to realize this homomorphism is then checked by partial runs that branch off from the main path.

We now describe the set $S$ of states of $\Amf$. For the main path, we use states $s_{V,W}^g$ where $V \subseteq \mn{var}(q)$, $g:V \rightarrow \mn{N}$ is a partial function, and $W$ is a subset of the binary atoms in $q$. Informally, the meaning of the state $s_{V,W}^g$ is that the variables from $V$ have already been mapped to individuals in bags seen before, the binary atoms from $W$ are already satisfied via this mapping, and $g$ describes how variables are mapped to individuals that are in the intersection of the previous and the current bag. The initial state of $\Amf$ is $s_{\emptyset,\emptyset}^g$ with $g$ the empty map. We also use states $s_A^a$ that make sure that the concept name $A$ can be derived at $a \in \mn{N}$.

We have to take care of the fact that a homomorphism from $q$ to $\Umc_{\Amc, \Tmc}$ can map existentially quantified variables to anonymous individuals, which are not explicitly represented in the input word. Let \Bmc be a $\Sigma$-ABox. A \emph{partial $q$-match} in $\Bmc$ is a partial function $h : \mn{var}(q) \rightarrow \mn{ind}(\Bmc) \times \{\mn{named},\mn{anon}\}$ such that if $x,y \in \mn{dom}(h)$, $r(x,y) \in q$ and $h_2(x) = h_2(y) = \mn{named}$, then $r(h_1(x),h_1(y)) \in \Bmc$, where $h_1$ and $h_2$ are the projections of $h$ to the first and second component, respectively.  Informally, a partial $q$-match $h$ partially describes a homomorphism $g$ from $q$ to $\Umc_{\Bmc,\Tmc}$ where $h(x)=(a,\mn{named})$ means that $g(x)=a$ and $h(x)=(a,\mn{anon})$ means that $g(x)$ is some element in the subtree below $a$ generated by the chase.  Whether a part of the query can map into the anonymous part below some individual $a$ only depends on the type realized at~$a$.  Define a relation $\Rmc \subseteq \mn{TP} \times 2^{\mn{var}(q)} \times 2^{\mn{var}(q)}$ by putting $(t,V_1,V_2) \in \Rmc$ if and only if $V_1 \subseteq V_2$, $V_2 \cap \xbf \subseteq V_1$ and there is a homomorphism from $q |_{V_2}$ to the universal model of the ABox $\{A(a) \mid A \in t\}$ and \Tmc that maps precisely the variables from $V_1$ to the root of the (tree-shaped) universal model.

An \emph{explanation set} for a partial $q$-match $h : \mn{var}(q) \rightarrow \mn{ind}(\Bmc) \times \{\mn{named},\mn{anon}\}$ is a set $Z$ of concept assertions that uses only individuals from $\mn{ind}(\Bmc)$ and satisfies the following conditions:
\begin{enumerate}

\item if $h(x) = (a, \mn{named})$, then $\Bmc \cup Z, \Tmc \models A(h(x))$ for all $A(x) \in q$ and

\item if $h(x) = (a,\mn{anon})$, then 
  $(\{A \mid A(a) \in Z\}, h^{-1}(a,\mn{named}), h_1^{-1}(a)) \in \Rmc.$
  
\end{enumerate}
Next, we describe the transition function $\delta$. The following transitions are used for the main branch of automata runs:
\[\delta(s_{V,W}^g,(\bbf,\Bmc,\cbf,f)) = \bigvee_{h \in H} \left( (\mn{right},s_{V_h,W_h}^{g_h}) \wedge \bigvee_{Z \in \Zmc_h} \left( \bigwedge_{A(a) \in Z} (\mn{stay},s_A^a)  \right) \right) \]
where $H$ is the set of all partial $q$-matches $h$ for $\Bmc$ such that $\mn{dom}(g) \subseteq \mn{dom}(h)$, $h_1$ and $g$ agree on the intersection of their domains,
and so do $h_1$ and $f$, and where
\begin{itemize}
\item $g_h$ is $h_1$ restricted to answer variables $x_i$ with $h_1(x_i)$ in~$\cbf$,
\item $V_h = V \cup \mn{dom}(h)$,
\item $V \cap \mn{dom}(h) = \mn{dom}(g)$,
\item $W_h$ is the union of $W$ and all binary atoms from $q$ that only use variables from $h$, and
\item $\Zmc_h$ is the set of all explanation sets for $h$.
\end{itemize}
When the automaton reads the right end marker $\dashv$ and is in a state signifying that a complete homomorphism from $q$ to $\Umc_{\Amc,\Tmc}$ has been found, then we accept the input using the transition $\delta(s_{V,W}^g,\dashv) = \mn{true}$
where $V = \mn{var}(q)$, $W$ is the set of all binary atoms of $q$, and $g$ is
the empty map.
 
The following transitions are used to verify the required concept memberships by checking for the existence of a suitable derivation tree (Lemma~\ref{lem:AQderivation}). Consider a state $s_A^a$ and a symbol $(\bbf, \Bmc, \cbf, f)$ such that $a \in \mn{ind}(\Bmc)$. If $A(a) \in \Bmc$, we set $\delta(s_A^a,(\bbf,\Bmc,\cbf,f)) = \mn{true}$. If $a$ appears neither in $\bbf$ nor in $\cbf$:
\[\delta(s_A^a,(\bbf,\Bmc,\cbf,f)) = \bigvee_{\substack{Z \\ \Bmc \cup Z, \Tmc \models A(a)}} \bigwedge_{B(b) \in Z}(\mn{stay},s_B^b) \]
If $a$ appears in $\bbf$ but not in $\cbf$:
\[\delta(s_A^a,(\bbf,\Bmc,\cbf,f)) = (\mn{left},s_A^a) \vee \bigvee_{\substack{Z \\ \Bmc \cup Z, \Tmc \models A(a)}} \bigwedge_{B(b) \in Z}(\mn{stay},s_B^b) \]
If $a$ appears in $\cbf$ but not in $\bbf$:
\[\delta(s_A^a,(\bbf,\Bmc,\cbf,f)) = (\mn{right},s_A^a) \vee \bigvee_{\substack{Z \\ \Bmc \cup Z, \Tmc \models A(a)}} \bigwedge_{B(b) \in Z}(\mn{stay},s_B^b) \]
If $i$ appears in both $\bbf$ and $\cbf$:
\[\delta(s_A^a,(\bbf,\Bmc,\cbf,f)) = (\mn{left},s_A^a) \vee (\mn{right},s_A^a) \vee \bigvee_{\substack{Z \\ \Bmc \cup Z, \Tmc \models A(a)}} \bigwedge_{B(b) \in Z}(\mn{stay},s_B^b) \]
Set $\delta(\cdot,\cdot) = \mn{false}$ for all pairs from $S \times \Gamma$ that were not mentioned.

\smallskip
\noindent
The automaton is now defined as $\Amf = (S, \Gamma, \delta, s_0)$.

\begin{lemma} \label{lem:2afacorrect} Let 
  \Amc be an ABox of pathwidth at most $k$, $\abf \in \mn{ind}(\Amc)^{\mn{ar}(q)}$, and $w \in \Gamma^*$ a word that encodes $(\Amc, \abf)$.
Then $\Amc \models Q(\abf)$ if and only if $w \in L(\Amf)$.  \end{lemma}

\begin{proof}
We start by proving that the states of the form $s_A^a$, used for checking the existence of a derivation for $A(a)$, work as in intended.
\\[2mm]
\textbf{Claim.} Let $a \in V_i$ and $a'$ the name of its copy in the $\Bmc_i$. Then there exists a successful run starting from the configuration $(s_A^{a'},i)$ if and only if $\Amc, \Tmc \models A(a)$.
\\[2mm]
First, assume that $\Amc, \Tmc \models A(a)$. We have to construct a successful run starting from the configuration $(i, s_A^{a'})$. By Lemma~\ref{lem:AQderivation}, there exists a derivation tree for $A(a)$. The statement can be proved by induction on the minimal number $k$ such that $A(a)$ has a derivation tree of depth $k$. If $k=0$, then $A(a) \in \Amc$, so $A(a') \in \Bmc_i$, and in this case we have $\delta(s_A^a, (\bbf_i, \Bmc_i, \cbf_i, f_i)) = \mn{true}$, which means the run is successful. Now let $k>0$ and consider a derivation tree for $A(a)$ of depth $k$.
\begin{itemize}
\item If the children of the root are of the form $B_1(a),\ldots,B_n(a)$ such that $\Tmc \models B_1 \sqcap \ldots \sqcap B_n \sqsubseteq A$, then choose the set $Z = \{B_1(a'),\ldots,B_n(a')\}$ in the transition, so in the run, we add the children labeled with $(i,s_{B_j}^{a'})$ for all $1 \leq j \leq n$. By induction hypothesis, from all these configurations there exists a successful run, so these runs can be combined to obtain a successful run for $(i,s_A^{a'})$.
\item If the root has one child labeled $B(b)$ and we have $\Tmc \models \exists r.B \sqsubseteq A$, then there exist $b \in \mn{ind}(\Amc)$ and $r(a,b) \in \Amc$. This individual $b$ does not necessarily lie in $\Bmc_i$, but by the properties of a path decomposition, there exists a bag $V_j$ such that $a,b \in V_j$ and, since $a \in V_i$, we also have $a \in V_k$ for all $k$ between $i$ and $j$. We extend the run as follows: If $j<i$, then use the transition $(\mn{left},s_A^{a'})$ for $i-j$ times. If $j>i$, then use the transition $(\mn{right},s_A^{a'})$ for $j-i$ times. Then we are in the configuration $(j,s_A^{a'})$ and if we choose $Z = \{B(b')\}$, we can extend the run successfully by the induction hypothesis.
\end{itemize}

For the other direction, assume that there is a successful run starting from the configuration $(i,s_A^{a'})$. We have to argue that $\Amc, \Tmc \models A(a)$. The proof is by induction on the depth of the run. If the run has depth $0$, i.e.\ the configuration $(i,s_A^{a'})$ does not have any successors, then we must have $\delta(s_A^{a'}, (\bbf_i, \Bmc_i, \cbf_i, f_i)) = \mn{true}$. This is only the case if $A(a') \in \Bmc_i$, so $A(a) \in \Amc$ and clearly, $\Tmc, \Amc \models A(a)$. Now assume the run has depth $k>0$. If the root node has a successor labeled $(i-1,s_A^{a'})$ or $(i+1,s_A^{a'})$, by induction hypothesis we have $\Amc, \Tmc \models A(a)$. If the root node does not have a successor of this kind, then there exists a set $Z$ and successors $(i,s_B^{b'})$ for all $B(b') \in Z$ such that $\Bmc \cup Z \models A(a')$. By induction hypothesis, we have $\Amc, \Tmc \models B(b)$ for all $B(b') \in Z$. Together, this gives $\Amc, \Tmc \models A(a)$. This finishes the proof of the claim.

Now we are ready to prove the lemma.

``$\Rightarrow$''. Let $\Amc \models Q(\abf)$ and let $(\bbf_1,\Bmc_1,\cbf_1,f_1) \ldots (\bbf_n,\Bmc_n,\cbf_n,f_n)$ be an encoding of $(\Amc, \abf)$ based on some $(j,k+1)$ path decomposition $V_1,\ldots,V_n$ of $\Amc$. There exists a homomorphism $h_0$ from $q$ to $\Umc_{\Amc,\Tmc}$ that maps the answer variables to $\abf$. We use $h_0$ to guide the accepting run of $\Amf$ on the word encoding $(\Amc, \abf)$. In the $i$-th step of the main branch of the run, always choose the partial $q$-match $h \in H$ according to $h_0$, i.e.\ if $h_0(x) = a \in \mn{ind}(\Amc) \cap V_i$ then $h(x) = (a',\mn{named})$, and if $h_0(x) = b$ for some anonymous individual $b$ that lies in the subtree below some $c \in \mn{ind}(\Amc)$ then $h(x) = (c',\mn{anon})$. As the explanation set for $h$ we can just choose $Z_h = \{A(a') \mid a \in V_i \text{ and } \Tmc, \Amc \models A(a)\}$.

We argue that following these choices, the main path will be successful, i.e.\ the leaf of the main branch is labeled with $(s_{V,W}^g, \dashv)$ such that $V = \mn{var}(q)$, $W$ is the set of all binary atoms of $q$ and $g$ the empty map. Let $x \in \mn{var}(q)$. Then either $h_0(x) \in \mn{ind}(\Amc)$ or $h_0(x)$ is an anonymous individual below some $b \in \mn{ind}(\Amc)$. If $h_0(x) \in \mn{ind}(\Amc)$, then let $V_i$ be the first bag such that $h_0(x) \in V_i$ and thus there is a copy of $h_0(x)$ in $\mn{ind}(\Bmc_i)$. Thus, in the $i$-th step of the main branch, $x$ is added to $V$. Similarly, if $h_0(x)$ is an anonymous individual below some $b \in \mn{ind}(\Amc)$, then let $V_i$ be the first bag such that $b \in V_i$. Again, one can conclude that $x$ is added to $V$ in the $i$-th step of the main branch. Overall, it follows that $V = \mn{var}(q)$. Now, let $r(x,y)$ be a binary atom from $q$. If both $h_0(x)$ and $h_0(y)$ are in $\mn{ind}(\Amc)$, then, since $V_1,\ldots,V_n$ is a path decomposition of $\Amc$, there exists a bag $V_i$ such that $h_0(x),h_0(y) \in V_i$, so there exists a copy of $r(h_0(x),h_0(y))$ in $\Bmc_i$ and in the $i$-th step of the main branch, $r(x,y)$ is added to $W$. If at least one of $h_0(x)$ and $h_0(y)$ is not in $\mn{ind}(\Amc)$, but is an anonymous individual below some $b \in \mn{ind}(\Amc)$, then either both $h_0(x)$ and $h_0(y)$ are mapped to anonymous individuals below $b$ or one of them is mapped to $b$. In any case, $r(x,y)$ is added to $W$ in the $i$-th step of the main branch. Overall, it follows that $W$ is the set of all binary atoms of $q$. Finally, $g$ must be the empty map, since $\cbf_n = \emptyset$.

It follows immediately from the claim above that the other paths will be successful as well, i.e.\ whenever $\Amc, \Tmc \models A(a)$ for some $a \in V_i$, then there is a successful run that starts at $(s_A^{a'}, i)$. This concludes the proof of the first direction.

``$\Leftarrow$''. Assume that there is a successful run of $\Amf$ on the input word $w=(\bbf_1,\Bmc_1,\cbf_1,f_1) \ldots (\bbf_n,\Bmc_n,\cbf_n,f_n)$. The run must have one main path with states of the form $s_{V,W}^g$. In every step of the main path, one partial $q$-match $h$ together with an explanation set $Z_h$ is chosen. From these partial $q$-matches we can construct a map $h_0$ from $\mn{var}(q)$ to the universal model of $\Amc$ and $\Tmc$ in the following way: Whenever a partial $q$-match $h$ maps a variable $x$ to $(a,\mn{named})$, we set $h_0(x) = a$. Whenever a partial $q$-match $h$ maps a variable $x$ to $(a,\mn{anon})$, then consider the explanation set $Z_h$. By the definition of the explanation set, we have $(\{B \mid B(a) \in Z_h\}, h^{-1}(a,\mn{named}), h_1^{-1}(a)) \in \Rmc$, so there exists a homomorphism from $h_1^{-1}(a)$ to the canonical model of $\{B \mid B(a) \in Z_h\}$ that maps precisely the variables from $h^{-1}(a,\mn{named})$ to the root, which is the homomorphism we use to build $h_0$. From the condition $V \cap \mn{dom}(h) = \mn{dom}(g)$ it follows that a partial $q$-match chosen later in the run will not assign a different image to a variable that has appeared earlier in the domain of a partial $q$-match, so $h_0$ is well defined.

We show that $h_0$ is indeed a homomorphism from $q$ to $\Umc_{\Amc, \Tmc}$ with $q(\xbf) = \abf$. Since the main branch ends in a configuration $(s_{V,W}^g, \dashv)$, where $V = \mn{var}(q)$, we know that $\mn{dom}(h_0) = \mn{var}(q)$. We argue that every atom of $q$ is satisfied by $h_0$.
\begin{itemize}
\item Let $A(x)$ be a unary atom from $q$ such that $h_0(x) \in \mn{ind}(\Amc)$. Let $h$ be the partial $q$-match that determined $h_0(x)$, so $h(x) = (a', \mn{named})$ for some $a' \in \mn{ind}(\Bmc)$, and let $Z_h$ be the explanation set chosen in the run, so the configuration has children labeled $(i,s_B^b)$ for every $B(b) \in Z_h$. Since the run is successful and we know from the claim above that a partial run starting from the configuration $(i,s_B^{b'})$ is successful if and only if $\Tmc, \Amc \models B(b)$, we have $\Tmc, \Amc \models B(b)$ for all $B(b') \in Z_h$ and thus, $\Tmc, \Amc \models A(a)$.
\item Let $r(x,y)$ be a binary atom from $q$ such that both $h_0(x)$ and $h_0(y)$ are in $\mn{ind}(\Amc)$. Since the main branch ends in the state $s_{V,W}^g$, where $W$ is the set of all binary atoms from $q$, there must be one step in the main branch, where $r(x,y)$ has been added to $W$, say the $i$-th step. This means that both $x$ and $y$ lie in $\mn{dom}(h)$, where $h$ is the partial $q$-match chosen in the $i$-th step. Since $h$ respects role atoms, we have $r(h_0(x),h_0(y)) \in \Amc$.
\item Let $A(x)$ be a unary atom from $q$ such that $h_0(x)$ is an anonymous individual below some $a \in \mn{ind}(\Amc)$. Let $h$ be the partial $q$-match that determined $h_0(x)$, so $h(x) = (a', \mn{anon})$, and let $Z_h$ be the explanation set chosen in the run. By definition of an explanation set, there is a partial homomorphism from $q$ with $x$ in its domain to the universal model of $\{B(a) \mid B(a) \in Z_h\}$, which was used to define $h_0$ on $x$, so we have $\Tmc, \Amc \models A(h_0(x))$.
\item Let $r(x,y)$ be a binary atom from $q$ such that at least one of $h_0(x)$ and $h_0(y)$ is an anonymous individual. Then the argument is similar to the previous case.
\end{itemize}
\end{proof}

\smallskip
\noindent
\textbf{Construction of the linear Datalog program.} Since every 2AFA can be transformed into an equivalent deterministic finite automaton (DFA) \cite{GeffertO14}, Lemma~\ref{lem:2afacorrect} also ensures the existence of a DFA $\Amf = (Q,\Sigma,\delta,s_0,F)$ that a word that encodes a pair $(\Amc, \abf)$
if and only if $\Amc \models Q(\abf)$. We use $\Amf$ to construct the desired linear Datalog rewriting of $Q$.

The idea for the program is to guess a tuple $\abf \in \mn{ind}(\Amc)^{\mn{ar}(q)}$ up front and then verify that $\Amc \models Q(\abf)$ by simulating $\Amf$. The program uses the states of \Amf as IDBs. Each of these IDBs can appear in any arity between $\mn{ar}(q)$ and $\mn{ar}(q) + k$ with $k$ the pathwidth of $Q$---technically, this means that we have $k+1$ different IDBs for every state, but we use the same symbol for all of them since the arity will always be clear from the context. The first $\mn{ar}(q)$ components of each IDB are used to store the tuple $\abf$ while the other components are used to store the individuals that occur in both of two consecutive bags of a path decomposition of $\Amc$.

\smallskip
\noindent
Start rules: Given the ABox $\Amc$, the program starts by guessing a tuple $\abf\in \mn{ind}(\Amc)^{\mn{ar}(q)}$ using the following rule:
\[s_0(x_1,\ldots,x_n) \leftarrow \top(x_1) \wedge \top(x_2) \wedge \cdots \wedge \top(x_n).\]

\smallskip
\noindent
Transition rules: Consider any transition $\delta(s_1,(\bbf,\Bmc,\cbf,f)) = s_2$. Let $\varphi_\Bmc$ be \Bmc viewed as a conjunction of atoms with individual names viewed as variables and let $\xbf'$ be obtained from $\xbf$ by replacing every variable $x_i \in \mn{dom}(f)$ by the individual name $f(x_i) \in \mn{ind}(\Bmc)$, also here viewed as a variable. We then include the following rule:
\[s_2(\xbf',\cbf) \leftarrow s_1(\xbf',\bbf) \wedge \varphi_\Bmc.\]
This rule says that if the DFA is in state $s_1$, the intersection between the last bag and the current bag is $\bbf$, and we see a homomorphic image of $\Bmc$, then the DFA can transition into state $s_2$ and remember the tuple $\cbf$. Applying such a rule leaves the tuple \abf stored in the first $\mn{ar}(q)$ components unchanged, but some of the variables in $\xbf'$ can appear in $\varphi_\Bmc$ to enforce that \abf is
compatible with the mapping of the answer variables that is prescribed by $f$ and
used in the simulated run of \Amf.

\smallskip
\noindent
Goal rules: If $s \in F$, then include the following rule:
\[\mn{goal}(\xbf) \leftarrow s(\xbf). \]

\begin{lemma}
  \label{lem:ldlogprogram}
  $\Pi$ is a rewriting of $Q$.
\end{lemma}
\begin{proof}
  Let \Amc be a $\Sigma$-ABox and let $\abf \in \mn{ind}(\Amc)^{\mn{ar}(q)}$. First
  assume that $\Amc \models Q(\abf)$. Since $Q$ is of pathwidth $k$, there must be a $\Sigma$-ABox $\Amc'$ of pathwidth at most~$k$ such that $\Amc' \models Q(\abf)$ and there is a homomorphism from $\Amc'$ to $\Amc$ that is the identity on \abf.
  Let $w=w_1\ldots w_n \in \Gamma^*$ encode the pair $(\Amc', \abf)$
  where $w_i = (\bbf_i,\Bmc_i,\cbf_i,f_i)$, and assume that $w \in L(\Amf)$.
  There is an accepting run of \Amf on $w$ and thus we
  find states $s_0,\ldots,s_n$ of $\Amf$ such that $\delta(s_i,w_i) = s_{i+1}$ for $0 \leq i < n$ and $s_n$ is an accepting state. This yields a derivation of $\Pi(\abf)$ in $\Amc'$, as follows. First, use the start rule to derive $s_0(\abf)$. For the next $n$ steps, use the rule introduced for the transitions $\delta(s_i,w_i) = s_{i+1}$. In this way, we derive $s_n(\abf)$ since the individuals in the first $\mn{ar}(q)$ components do not change when using the transition rules. Because $s_n \in F$, a goal rule can be applied to derive $\mn{goal}(\abf)$ and thus $\Amc' \models \Pi(\abf)$. It is well-known that answers to Datalog programs are preserved under ABox homomorphisms \cite{AbiteboulHV95} and there is a homomorphism from $\Amc'$ to \Amc that is the identity on \abf, we obtain $\Amc \models \Pi(\abf)$ as desired.

  For the converse direction, assume that $\Amc \models \Pi(\abf)$. Then there is a derivation $D$ of $\Pi(\abf)$ in \Amc. Since $\Pi$ is of diameter at most $k$, $\Amc_D$ has pathwidth at most~$k$. Consider the encoding of $(\Amc_D,\abf)$ as a word $w \in \Gamma^*$, based on the path decomposition induced by $D$ in the obvious way. By construction of $\Pi$, $D$ must use a start rule for $\abf$, then a number of transition rules, and then a goal rule. Using the way in which these rules
  are constructed, it can be verified that this yields an accepting run of \Amf on
  $w$. Thus $\Amc_D \models Q(\abf)$. It remains to recall that there is a homomorphism from $\Amc_D$ to \Amc that is the identity on \abf, and that answers to OMQs from $(\ELI,\textnormal{CQ})$ are preserved under ABox homomorphisms~\cite{BienvenuCLW14}.  
\end{proof}

\section{The Trichotomy for Disconnected CQs}
\label{sec:disconnected}

We now lift the trichotomy result that is provided by Theorems~\ref{thm:AC0NL} and~\ref{thm:NLPTime} from connected CQs to unrestricted CQs. To achieve this, we show that the complexity of an OMQ $Q=(\Tmc,\Sigma,q)$ with $q$ a disconnected CQ is precisely the complexity of the hardest OMQ $(\Tmc,\Sigma,q')$ with $q'$ a maximal connected component (MCC) of $q$, provided that we first have removed redundant MCCs from $q$.

Let $Q=(\Tmc, \Sigma, q) \in (\EL, \textnormal{CQ})$. We say that $Q$ is \emph{empty} if $\Amc \not\models Q(\abf)$ for all $\Sigma$-ABoxes \Amc and tuples \abf. Every
empty OMQ is trivially FO-rewritable. An MCC of $q$ is \emph{Boolean} if it contains no answer variables. We call $Q$ \emph{redundant} if there is a Boolean MCC of $q$ such that the OMQ obtained from $Q$ by dropping that MCC from $q$ is equivalent to $Q$. For proving the intended trichotomy result, it is clearly sufficient to consider OMQs that are non-empty and non-redundant.
\begin{theorem}
\label{thm:AC0NLdisc}
Let $Q \in (\EL, \textnormal{CQ})$. Then either 
\begin{enumerate}
\item $Q$ is FO-rewritable and thus {\sc eval}$(Q)$ is in $\AC^0$ or
\item $Q$ is not FO-rewritable and {\sc eval}$(Q)$ is $\NL$-hard under FO reductions.
\end{enumerate}
\end{theorem}
\begin{proof}
  Let $Q=(\Tmc, \Sigma, q) \in (\EL, \textnormal{CQ})$ be non-redundant and non-empty and let $q_1(\xbf_1),\dots, q_n(\xbf_n)$ be the MCCs of $q(\xbf)$. 

  If every OMQ $Q_i=(\Tmc, \Sigma, q_i)$ is FO-rewritable, then the conjunction of all these FO-rewritings is an FO-rewriting of $Q$. It thus suffices to show that otherwise, $Q$ is $\NL$-hard. Thus assume that some $Q_i$ is not FO-rewritable. Since $q_i$ is connected, \textsc{eval}($Q_i$) is $\NL$-hard under FO reductions by Theorem~\ref{thm:AC0NL}. We prove that \textsc{eval}($Q$) is $\NL$-hard under FO reductions by giving an FO reduction from \textsc{eval}$(Q_i)$. Let $\Amc_i$ be a $\Sigma$-ABox and $\abf_i$ a tuple from $\mn{ind}(\Amc_i)^{\mn{ar}(q_i)}$. Since $Q$ is non-empty and non-redundant, for every $j \neq i$ we find a $\Sigma$-ABox $\Amc_j$ and a tuple $\abf_j$ such that
\begin{enumerate}
\item $\Tmc, \Amc_j \models q_j(\abf_j)$ and
 \item if $q_i$ is Boolean, then $\Tmc, \Amc_j \not \models q_i$.
\end{enumerate}
Define $\Amc$ to be the disjoint union of $\Amc_1,\dots,\Amc_n$ and $\abf=\abf_1\cdots\abf_n$. Clearly, $\Amc$ and $\abf$ can be defined by an FO-query, so this is an FO reduction.

We have to show that $\Amc_i \models Q_i(\abf_i)$ iff $\Amc \models Q(\abf)$.  The ``$\Rightarrow$'' direction is clear by construction of \Amc and \abf. For ``$\Leftarrow$'', assume that $\Amc \models Q(\abf)$. This implies $\Amc \models Q_i(\abf_i)$, so there is a homomorphism $h$ from $q_i$ to $\Umc_{\Amc, \Tmc}$ with $h(\xbf_i) = \abf_i$. The universal model $\Umc_{\Amc, \Tmc}$ is the disjoint union of the universal models $\Umc_{\Amc_j, \Tmc}$, $1 \leq j \leq n$. Since $q_i$ is connected, the range of $h$ lies completely inside one of the $\Umc_{\Amc_j, \Tmc}$. In fact, it must lie in $\Umc_{\Amc_i, \Tmc}$. If $q_i$ is not Boolean, this is the case because $h(\xbf_i)=\abf_i$ is a tuple from $\Amc_i$. If $q_i$ is Boolean, then this follows from $\Tmc, \Amc_j \not \models q_i$ which implies $\Umc_{\Amc_j, \Tmc} \not\models q_i$. We have thus shown that $\Umc_{\Amc_i, \Tmc} \models q_i(\abf_i)$, implying $\Amc_i \models Q_i(\abf_i)$ by Lemma~\ref{lem:unimodelproperties}, as desired. We have shown that $(\Tmc, \Sigma, q)$ is $\NL$-hard. It follows that $(\Tmc, \Sigma, q)$ is not FO-rewritable \cite{FurstSS81}.  \end{proof}

To lift the dichotomy between $\NL$ and $\PTime$ dichotomy including the equivalence of \NL and linear Datalog rewritability, we first give a helpful
lemma aboutlinear Datalog programs.

\begin{lemma}
\label{lem:LDlogproduct}
Let $\Pi_1,\ldots,\Pi_n$ be linear Datalog programs. Then there exists a linear Datalog program $\Pi$ of arity $\Sigma_{i=1}^n \mn{ar}(\Pi_i)$ such that for all ABoxes $\Amc$ and tuples $\abf_1,\ldots,\abf_n$,
\[ \Amc \models \Pi_i(\abf_i) \; \text{ for $1 \leq i \leq n$ iff } \; \Amc \models \Pi(\abf_1,\ldots,\abf_n) \;.\]
\end{lemma}
\begin{proof}
  It suffices to give a proof for the case $n=2$, the general case follows by repeatedly applying the lemma for $n=2$. So let $\Pi_1$, $\Pi_2$ be linear Datalog programs. We assume w.l.o.g. that $\Pi_1$ and $\Pi_2$ use disjoint sets of variables. Define a program $\Pi$ that contains
  the following rules:
  \begin{itemize}

  \item for all rule $S_i(\xbf_i) \leftarrow \vp_i(\xbf_i,\ybf_i) \in \Pi_i$, $i \in \{1,2\}$,
    such that neither $\vp_1$ not $\vp_2$ contains an EDB relation, the rule
    $$
    (S_1,S_2)(\xbf_1,\xbf_2) \leftarrow \vp_1(\xbf_1,\ybf_1) \wedge \vp_2(\xbf_2,\ybf_2);
    $$

  \item for each rule $S_i(\xbf_i) \leftarrow \vp_i(\xbf_i,\ybf_i) \in \Pi_i$ and
    each IDB relation $S_{3-i}$ from $\Pi_{3-i}$, $i \in \{1,2\}$, the rule
    $$
    (S_1,S_2)(\xbf_1,\xbf_2) \leftarrow \vp_i(\xbf_i,\ybf_i) \wedge S_{3-i}(\xbf_{3-i})
    $$
    where $\xbf_{3-i}$ is a tuple of fresh variables;

  \item the rule
    $$
       \mn{goal}(\xbf_1,\xbf_2) \leftarrow (\mn{goal},\mn{goal})(\xbf_1,\xbf_2).
    $$
    
  \end{itemize}
  It can be verified that $\Amc \models \Pi(\abf_1, \abf_2)$ iff both $\Amc \models \Pi_1(\abf_1)$ and $\Amc \models \Pi_2(\abf_2)$, for all $\Sigma$-ABoxes \Amc
  and tuples $\abf_1$, $\abf_2$.
\end{proof}
\begin{theorem}
\label{thm:NLPTimedisc}
Let $Q \in (\EL, \textnormal{CQ})$. The following are equivalent (assuming $\NL \neq \PTime$):
\begin{enumerate}
\item $Q$ has bounded pathwidth;
\item $Q$ is linear Datalog rewritable;
\item {\sc eval}$(Q)$ is in $\NL$.
\end{enumerate}
If these conditions do not hold, then {\sc eval}$(Q)$ is $\PTime$-hard under FO reductions.
\end{theorem}
\begin{proof}
  The equivalence of (1) and (2) has been shown in Lemma~\ref{lem:LDLog} and the implication (2) $\Rightarrow$ (3) is clear. To finish the proof, we show that if (2) does not hold, then {\sc eval}$(Q)$ is $\PTime$-hard, proving the implication (3) $\Rightarrow$ (2) as well as the last sentence of the theorem.

  Let $Q=(\Tmc, \Sigma, q) \in (\EL, \textnormal{CQ})$ and assume that $Q$
  is not rewritable into linear Datalog. As before, we can assume
  $Q$ to be non-empty non-redundant. Let $q_1(\xbf_1), \dots, q_n(\xbf_n)$ be the connected components of $q(\xbf)$. By Theorem~\ref{thm:NLPTime}, every OMQ
  $Q_i = (\Tmc, \Sigma, q_i)$ is either rewritable into linear Datalog or $\PTime$-hard.
  By Lemma~\ref{lem:LDlogproduct}, every $(\Tmc, \Sigma, q_i)$ being rewritable into linear Datalog implies that also $(\Tmc, \Sigma, q)$ is linear Datalog rewritable. Since this is not the case, there must be some
  $Q_i$ that is not rewritable into linear Datalog and thus $\PTime$-hard.

  It is now possible to show \PTime-hardness of {\sc eval}$(Q)$ by an FO reduction from
   {\sc eval}$(Q_i)$, exactly as in the proof of Theorem~\ref{thm:AC0NLdisc}.
%
\end{proof}

\begin{remark}
  Recall that we are working with CQs that do not admit equality throughout this
  paper. However, the trichotomy result established in this section can easily be generalized to $(\EL, \textnormal{CQ}^=)$ where CQ$^=$ denotes the class of
  CQs with equality. This is due to the observation that for every OMQ $Q \in (\EL, \textnormal{CQ}^=)$,
  there is an OMQ $Q' \in (\EL, \textnormal{CQ})$ such that there is an FO reduction
  from {\sc eval}$(Q)$ to {\sc eval}$(Q')$ and vice versa. Let $Q=(\Tmc, \Sigma, q) \in (\EL, \textnormal{CQ}^=)$. To construct $Q'$, we simply eliminate the equality atoms in $q$ by identifying variables, ending up with a CQ  $q'$ without equality atoms that might have lower arity than $q$. It is easy to see that the required FO reductions exist,
  which essentially consist of dropping resp.\ duplicating components from answer tuples.
\end{remark}

\section{Width Hierarchy for Linear Datalog Rewritability}
\label{sec:hierarchy}

The width of the linear Datalog rewritings constructed in Section \ref{sec:NLPTime} depends on $\mn{pw}(Q)$, so if $Q$ has high pathwidth, then we end up with a linear Datalog rewriting of high width. We aim to show that this is unavoidable, that is, there is no constant bound on the width of linear Datalog rewritings of OMQs from $(\EL,\text{CQ})$ and in fact not even from $(\EL,\text{AQ})$. It is interesting to contrast this with the fact that every OMQ from $(\mathcal{EL},\textnormal{CQ})$ can be rewritten into a \emph{monadic} (non-linear) Datalog program \cite{DL-Textbook}.  Our result strengthens a result by \cite{DalmauK08} who establish an analogous statement for CSPs. This does not imply our result: While every OMQ from $(\mathcal{EL},\textnormal{AQ})$ is equivalent to a CSP up to complementation~\cite{BienvenuCLW14}, the converse is false and indeed the CSPs used by Dalmau and Krokhin are not equivalent to an OMQ from $(\mathcal{EL},\textnormal{AQ})$.
Our main aim is to prove the following.
\begin{theorem}
\label{thm:widthhier}
  For every $k >0$, there is an OMQ $Q_k \in (\EL,\textnormal{AQ})$ that is
  rewritable into linear Datalog, but not into a linear Datalog program  
  of width $k$.  
\end{theorem}
When constructing the OMQs $Q_1,Q_2,\dots$ for Theorem~\ref{thm:widthhier}, we would like the ABoxes in~$\Mmc_{Q_k}$ to contain \emph{more and more branching}. Note that since we only work with $(\EL, \textnormal{AQ})$, $\Mmc_{Q_k}$ consists of tree-shaped ABoxes rather than of pseudo tree-shaped ones. Intuitively, if the ABoxes from $\Mmc_{Q_k}$ have a lot of branching, then a linear Datalog rewriting of $Q_k$ needs large width to simultaneously collect information about many different branches. However, we want $Q_k$ to be linear Datalog rewritable so by Lemma~\ref{lem:pathwidthbranching} the ABoxes from $\Mmc_{Q_k}$ must still branch only boundedly.  We thus construct $Q_k$ such that $\mn{br}(\Amc) \leq k$ for all $\Amc \in \Mmc_{Q_k}$ while for every $n \geq 1$, $\Mmc_{Q_k}$ contains an ABox $\Amc^n_k$ that takes the form of a tree of outdegree~2 and of depth $n$ such that $\mn{br}(\Amc^n_k) = k$ and $\Amc^n_k$ has the maximum number of leaves that any such ABox can have. To make the latter more precise, let $\ell_d^k(n)$ denote the maximum number of leaves in any tree that has degree $d$, depth $n$, and does not have the full binary tree of depth $k+1$ as a minor, $d,k,n \geq 0$. We then want $\Amc^n_k$ to have exactly $\ell_2^k(n)$ leaves, which ensures that it is maximally branching.

We now construct $Q_k$. For every $k \geq 1$, let $Q_k = (\Tmc_k, \Sigma, A_k(x))$ where $\Sigma = \{r, s, t, u\}$ and
$$
\begin{array}{rcl}
\Tmc_k &=& \{\top \sqsubseteq A_0\} \, \cup \\[1mm]
         &&\{\exists x.A_i \sqsubseteq B_{x,i} \mid x \in \{r,s,t,u\}, 0 \leq i \leq k-1 \} \, \cup \\ [1mm]
         &&\{\exists x.B_{x,i} \sqsubseteq B_{x,i} \mid x \in \{r,s,t,u\}, 0 \leq i \leq k-1 \} \, \cup \\ [1mm]
         &&\{B_{r,i} \sqcap B_{s,i} \sqsubseteq A_{i+1} \mid 0 \leq i \leq k-1\} \, \cup \\[1mm]
         &&\{B_{t,i} \sqcap B_{u,i+1} \sqsubseteq A_{i+1} \mid 0 \leq i \leq k-1\}. 
\end{array}
$$
Each concept name $A_i$ represents the existence of a full binary tree of depth $i$, that is, if $A_i$ is derived at the root of a tree-shaped $\Sigma$-ABox \Amc, then \Amc contains the full binary tree of depth $i$ as a minor. 
%
The concept inclusions $\exists x.B_{x,i} \sqsubseteq B_{x,i}$ in $\Tmc_k$ ensure that $Q_k$ is closed under subdivisions of ABoxes, that is, if $\Amc \in \Mmc_{Q_k}$ and $\Amc'$ is obtained from \Amc by subdividing an edge into a path (using the same role name as the original edge), then $\Amc \models Q_k(a)$ if and only if $\Amc' \models Q_k(a)$ for all $a \in \mn{ind}(\Amc)$. Informally spoken, subdivision will allow us to assume that every (connected) rule body in a linear Datalog rewriting can only `see' a single branching.

To provide a better understanding of the four role names used, we now explicitly define the ABoxes $\Amc^n_k$ mentioned above. We refer to non-leaf individuals by the combination of role names of their outgoing edges, e.g.\ an $rs$-individual is an individual that has one outgoing $r$-edge, one outgoing $s$-edge and no other outgoing edges.  Let $n,k \geq 1$.
If $n \leq k$, then $\Amc_k^n$ is the full binary tree of depth~$n$, where every non-leaf is an $rs$-individual. If $n > k = 1$, then $\Amc_k^n$ consists of a root that is a $tu$-individual where the $t$-successor is a leaf and the $u$-successor is the root of a copy of $\Amc_k^{n-1}$. Finally, for $n > k > 1$, take the disjoint union of $\Amc_{k-1}^{n-1}$ and $\Amc_k^{n-1}$ and introduce a new $tu$-individual as the root, the $t$-successor being the root of $\Amc_{k-1}^{n-1}$ and the $u$-successor the root of~$\Amc_k^{n-1}$. As an example, Figure~\ref{fig:tree-with-many-leaves} shows $\Amc^4_2$.

\begin{figure}
\begin{boxedminipage}[t]{\columnwidth}
\centering
\begin{tikzpicture}[->,>=stealth',
level 1/.style={sibling distance=18em},
level 2/.style={sibling distance=9em},
level 3/.style={sibling distance=5em},
level 4/.style={sibling distance=2.5em},
level 5/.style={sibling distance=0.7em},
level distance = 0.8cm,
font=\sffamily\small]
\tikzstyle{node}=[shape=circle, draw,inner sep=2.0pt, fill=black]
    \node [node] {}
        child {node [node] {}
            child {node [node] {}
                child {node [node] {}
                    child {node [node] {} edge from parent node[left] {$r$}}
                    child {node [node] {} edge from parent node[right] {$s$}}
                edge from parent node[left] {$r$}}
                child {node [node] {}
                    child {node [node] {} edge from parent node[left] {$r$}}
                    child {node [node] {} edge from parent node[right] {$s$}}
                edge from parent node[right] {$s$}}
            edge from parent node[above left] {$u$}}
            child {node [node] {}
                child {node [node] {}
                    child {node [node] {} edge from parent node[left] {$r$}}
                    child {node [node] {} edge from parent node[right] {$s$}}
                edge from parent node[left] {$u$}}
                child {node [node] {} edge from parent node[right] {$t$}}
            edge from parent node[above right] {$t$}}
        edge from parent node[above left] {$u$}}
        child {node [node] {}
            child {node [node] {}
                child {node [node] {}
                    child {node [node] {} edge from parent node[left] {$r$}}
                    child {node [node] {} edge from parent node[right] {$s$}}
                edge from parent node[left] {$u$}}
                child {node [node] {} edge from parent node[right] {$t$}}
            edge from parent node[above left] {$u$}}
            child {node [node] {}
            edge from parent node[above right] {$t$}}
        edge from parent node[above right] {$t$}};
\end{tikzpicture}
\end{boxedminipage}
\caption{The ABox $\Amc^4_2$, which has depth $4$, branching number $2$ and lies in $\Mmc_{Q_2}$. Since $4>2$, $\Amc^4_2$ is composed of one copy of $\Amc^3_2$ and one copy of $\Amc^3_1$ and a new $tu$-individual as root. This ABox has $11$ leaves, which is the largest number of leaves that a binary tree of depth $4$ can have, unless it contains the full binary tree of depth $3$ as a minor.}
\label{fig:tree-with-many-leaves}
\vspace*{-3mm}
\end{figure}

\begin{lemma}
\label{lem:Ank}
For all $n,k \geq 1$, 
\begin{enumerate}
\item $\Amc^n_k \in \Mmc_{Q_k}$;
\item $\mn{br}(\Amc^n_k) = k$;
\item $\Amc^n_k$ has exactly $\ell_2^k(n)$ leaves.
\end{enumerate}
\end{lemma}
All three points can be proved by induction on $n$.
The following lemma establishes the first part of Theorem~\ref{thm:widthhier}.
\begin{lemma}
\label{lem:hierarchyone}
For every $k \geq 1$, $Q_k$ is rewritable into linear Datalog.
\end{lemma}
\begin{proof}
We show that $\mn{br}(Q_k) = k$, which implies rewritability into linear Datalog by Lemma~\ref{lem:pathwidthbranching} and Theorem~\ref{thm:NLPTime}.

Let $\Amc \in \Mmc_{Q_k}$. We show that $\mn{br}(\Amc) = k$. First, let us analyse the types $\mn{tp}_{\Amc,\Tmc_k}(a)$, $a \in \mn{ind}(\Amc)$, and the structure of~\Amc. Since $\top \sqsubseteq A_0 \in \Tmc_k$, none of the types $\mn{tp}_{\Amc,\Tmc_k}(a)$ is empty. It is easy to verify that $\Tmc_k \models A_i \sqsubseteq A_{i-1}$ and $\Tmc_k \models B_{x,i} \sqsubseteq B_{x,i-1}$ for $1 \leq i \leq k$ and $x \in \{r,s,t,u\}$. 
We say that $a$ {\em is of type $i$} if $i$ is the largest integer such that $A_i \in \mn{tp}_{\Amc, \Tmc_k}(a)$ and that $a$ {\em is of $x$-type $j_x$} if $j_x$ is the largest integer such that $B_{x,j_x} \in \mn{tp}_{\Amc, \Tmc_k}(a)$.
\\[2mm]
\textbf{Claim 1.} Every individual in $\Amc$ has degree at most two and every individual of degree two is an $rs$-individual or a $tu$-individual.
\\[2mm]
We first argue that every individual has at most one $x$-successor for every $x \in \{r,s,t,u\}$. Assume to the contrary that there are distinct individuals $a,b,c$ and assertions $x(a,b), x(a,c) \in \Amc$ for some $x \in \{r,s,t,u\}$. Let $b$ have type $j$ and $x$-type $\ell$ and $c$ have type $m$ and $x$-type $n$. Then $B_{x,j}$, $B_{x,\ell}$, $B_{x,m}$ and $B_{x,n}$ are derived at $a$, but since $\Tmc_k \models B_{x,i} \sqsubseteq B_{x,i-1}$ for $1 \leq i \leq k$, these four concept names are already implied by $B_{x,\mn{max}\{j,\ell,m,n\}}$. Thus one of the individuals $b,c$ can be removed without altering the result of the query. 

Now we argue that every individual with degree greater than one is either an $rs$-individual or a $tu$-individual. All other combinations do not appear due to the minimality of \Amc. For example, assume that there is an $rst$-node $a$. Then some $B_{r,j}, B_{s,\ell}, B_{t,m}$ are derived at~$a$, assume that $j, \ell, m$ are maximal with this property. If $a$ is the root of \Amc, then the $t$-edge can be removed. If $a$ is not the root, it must be connected to its parent by a $t$-edge, since otherwise, the $t$-edge below $a$ can be removed. So assume, $a$ is a $t$-successor of its parent. If now $m \leq \mn{min}(j,\ell)$, the $t$-edge below $a$ can be removed. If $m > \mn{min}(j,\ell)$, but then both the $r$-edge and the $s$-edge can  be removed. In either case, \Amc is not minimal. This finishes the proof of Claim~1.

\smallskip

Using minimality of \Amc, it can be argued that for every $x$-individual $a$ (an indiviual with only one outgoing edge) with $x \in \{r,s,t,u\}$, there is some $b \in \mn{ind}(\Amc)$ with $x(b,a) \in \Amc$ and it follows that a path from one branching point to the next is always a chain of the same role. 
\\[2mm]
\textbf{Claim 2.} $a$ is of type $i$ iff $\mn{br}(\Amc^a)=i$ for all $a \in \mn{ind}(\Amc)$ that are leaves or of degree two.
\\[2mm]
We prove the claim by induction on the number $n$ of leaves $\Amc^a$. If $n=1$, then $a$ is a leaf itself, thus of type~0, and the statement follows. Now let $n > 1$ and let $a$ be an individual of degree two with $n$ leaves below it. We only argue the `if' direction, the `only if' direction can be argued similarly. So assume that $\mn{br}(\Amc^a) = i$ for some $i \geq 1$. Then by Claim~1, $a$ has two outgoing paths that both reach two nodes $b$ and $c$ that are a leaf or of degree two. Let $j = \mn{br}(\Amc^b)$ and $\ell = \mn{br}(\Amc^c)$ and w.l.o.g. assume $j \geq \ell$. By induction hypothesis, $b$ is of type $j$ and $c$ is of type $\ell$. There are two possibilities: Either $j = \ell$, which implies $i = j+1 = \ell+1$, or $j > \ell$, which implies $i=j$. In case $j = \ell$, $a$ must be an $rs$-individual. In fact, assuming $a$ was a $tu$-individual, then $A_i(a)$ would be derived using $B_{t,i-1} \sqcap B_{u,i} \sqsubseteq A_{i}$, so a full binary tree of depth $i$ below the $t$-successor of $a$ is not needed and one could remove any leaf below the $t$-successor of $a$ (contradicting minimality of $\Amc$), decreasing the depth of the largest binary tree minor by at most one. So since $a$ is an $rs$-individual, $B_{r,i-1} \sqcap B_{s,i-1} \sqsubseteq A_{i}$ applies and $a$ has type $i$. In case $j > \ell$, one can argue in a similar way that $a$ must be a $tu$-individual and $j = \ell +1$, and it follows that $a$ has type $i$.

\smallskip

Since $\Amc \models Q_k(a)$ for the root $a$ of $\Amc$, we know that $a$ is of type $k$, so Claim~2 says that $\mn{br}(\Amc) = k$.
\end{proof}

The following proofs rely on an estimate of $\ell_d^k(n)$, namely on the fact that $\ell_d^k(n)$ as a function of $n$ grows like a polynomial of degree~$k$. This is
established by the following lemma.
\begin{lemma}
\label{lem:numleaves}
$(d-1)^k(n-k)^k \leq \ell_d^k(n) \leq (k+1)(d-1)^k n^k$
for all $d,k \geq 0$ and $n \geq 2k$.
\end{lemma}
\begin{proof}
We aim to show that for all $d,k \geq 0$ and $n \geq 2k$,
\begin{equation*}
  \label{eq:1}
  \ell_d^k(n) = \sum_{i=0}^k (d-1)^i \binom{n}{i}
\tag{$*$}
\end{equation*}
From ($*$), the lower bound stated in the lemma is obtained by considering only the summand for $i=k$ and the upper bound is obtained by replacing every summand with the largest summand, which is the one for $i=k$ if $n \geq 2k$. 

\smallskip

Towards proving ($*$), we first observe that for all $n \geq 1$ and $k \geq 1$:
\begin{equation}
  \label{eq:2}
\ell_d^k(n) = \ell_d^k(n-1) + (d-1)\ell_d^{k-1}(n-1)
\tag{$**$}
\end{equation}
Let $T$ be a tree with degree $d$ and depth $n$ that does not contain the full binary tree of depth $k+1$ as a minor and that has the largest possible number of leaves. It can easily be seen that the root of $T$ has degree $d$ and that $T$ contains the full binary tree of depth $k$ as a minor. Consider the subtrees $T_1, \ldots, T_d$ whose roots are the children of the root of $T$. There must be one of them that also has the full binary tree of depth $k$ as a minor and all of them must have the full binary tree of depth $k-1$ as a minor, otherwise $T$ would not have the maximum number of leaves. Moreover, there cannot be two subtrees that both have the full binary tree of depth $k$ as a minor, since then $T$ would have a minor of depth $k+1$. Since the number of leaves of $T$ is the sum of the leaves of all $T_j$, ($**$) follows.

\smallskip

Now we prove ($*$) by induction on $n$.  First observe that $\ell_d^k(0) = \ell_d^0(n) = 1$ for all $d,k,n$, thus ($*$) holds for all cases where $k=0$ or $n=0$. Now let $k \geq 1$ and $n \geq 1$ and assume that ($*$) holds for $\ell_d^k(n)$ and for $\ell_d^{k-1}(n)$. We show that it also holds for $\ell_d^k(n+1)$:
\begin{align*}
  \ell_d^k(n+1) &= \ell_d^k(n) + (d-1) \cdot \ell_d^{k-1}(n)\\
  &= \sum_{i=0}^k (d-1)^i \binom{n}{i} + (d-1) \sum_{i=0}^{k-1}(d-1)^i\binom{n}{i}\\
  &=\sum_{i=0}^k (d-1)^i \binom{n}{i} + \sum_{i=1}^k(d-1)^i \binom{n}{i-1}\\
  &= 1 + \sum_{i=1}^k (d-1)^i \binom{n+1}{i}\\
  &=\sum_{i=0}^k(d-1)^i\binom{n+1}{i}
\end{align*}
\end{proof}

To show that linear Datalog rewritings of the defined family of OMQs require unbounded width, we first show that they require unbounded diameter and then proceed by showing that the width of rewritings cannot be significantly smaller than the required diameter.  To make the latter step work, we actually show the former on an infinite family of classes of ABoxes of restricted shape. More precisely, for all $i \geq 0$ we consider the class $\Cmf_i$ of all forest-shaped $\Sigma$-ABoxes in which the distance between any two branching individuals exceeds $i$ (where a forest is a disjoint union of trees and a branching individual is one that has at least two successors). Since the queries $Q_k$ are closed under subdivisions of ABoxes, each class $\Cmf_i$ contains ABoxes whose root is an answer to the query. 

The idea for proving that any linear Datalog rewriting of $Q_k$ requires high diameter is then as follows. We assume to the contrary that there is a linear Datalog rewriting $\Pi$ of $Q_k$ that has low diameter and consider the linear Datalog derivation of $Q_k(a)$ in some $\Amc^n_k$ with root $a$. A careful analysis shows that $\Amc_D$ contains a tree-shaped sub-ABox of depth $n$ that has as many leaves as $\Amc^n_k$ and thus by Lemma~\ref{lem:numleaves} contains a deep full binary tree as a minor. Thus $\Amc_D$ has high pathwidth which contradicts the assumption that $\Pi$ has low diameter.

\begin{lemma}
\label{lem:hierarchytwo}  
For any $i \geq 0$, $Q_{2k+3}$ is not rewritable into a linear Datalog program of diameter $k$ on the class of ABoxes~$\Cmf_i$. 
\end{lemma}

\begin{proof}
  Let $i \geq 0$. For $n \geq k \geq 1$, denote by $\Bmc^n_k$ the ABox obtained from $\Amc^n_k$ by subdividing every edge into a path of length $i+1$ of the same role. Note that $\Bmc^n_k \in \Cmf_i$. Using Lemma~\ref{lem:Ank}, it is easy to see that $\Bmc^n_k \in \Mmc_{Q_k}$ and $\Bmc^n_k$ has $\ell_2^k(n)$ leaves, so from Lemma~\ref{lem:numleaves} it follows that $\Bmc_k^n$ has at least $(n-k)^k$ leaves.
  
For the sake of contradiction, assume that there is a $k \geq 1$, such that $Q_{2k+3}$ is rewritable into a linear Datalog program $\Pi$ of diameter $k$ on the class $\Cmf_i$. Choose $n$ very large (we will make this precise later) and let $\Amc=\Bmc_{2k+3}^n$, so $\Amc \in \Mmc_{Q_{2k+3}}$, it has depth $n(i+1)$ and at least $(n-2k-3)^{2k+3}$ leaves.

Let $a_0$ be the root of $\Amc$. We have $\Amc, \Tmc \models \Pi(a_0)$ and thus there is a derivation $D$ of $\Pi(a_0)$ in \Amc. Consider the ABox $\Amc_D$. By Lemma~\ref{lem:nicestructure}, we have the following:
  \begin{enumerate}
  \item $\Amc_D \models \Pi(a_0)$;
  \item there is a homomorphism from $\Amc_D$ to $\Amc$ that is the identity on $a_0$;
  \item $\Amc_D$ has pathwidth at most $k$.
  \end{enumerate}
We manipulate $\Amc_D$ as follows:
  \begin{itemize}
  \item restrict the degree to $|\Tmc|$ by taking a subset according to Lemma~\ref{lem:smalldegree};
  \item remove all assertions that involve an individual $a$ that is not reachable from $a_0$ in $G_\Amc$ by a directed path.
  \end{itemize}
  We use \Bmc to denote the resulting ABox. It can be verified that Conditions 1 to~3 still hold when $\Amc_D$ is replaced with \Bmc. In particular, this is true for Condition~1 since $\Amc_D \models \Pi(a_0)$ iff $\Amc_D \models Q_k(a_0)$ iff $\Bmc \models Q_k(a_0)$ iff $\Bmc \models \Pi(a_0)$. The second equivalence is easy to establish by showing how a model witnessing $\Bmc \not \models Q_k(a_0)$ can be transformed into a model that witnesses $\Amc_D \not \models Q_k(a_0)$. 

  Choose a homomorphism $h$ from \Bmc to \Amc that is the identity on $a_0$. Then $h$ must be surjective since otherwise, the restriction $\Amc^-$ of \Amc to the individuals in the range of $h$ would satisfy $\Amc^-,\Tmc \models A_0(a_0)$, contradicting the minimality of~\Amc. Let $a_1,\dots,a_m$ be the leaves of \Amc, $m \geq (n-2k-3)^{2k+3}$. For each~$a_i$, choose a $b_i$ with $h(b_i)=a_i$. Clearly, all individuals in $b_1,\dots,b_m$ must be distinct. 
  
  By construction, $\Bmc$ is connected. Since there is a homomorphism from \Bmc to~\Amc, \Bmc must be a DAG (directed acyclic graph). We proceed to exhaustively remove assertions from \Bmc as follows: whenever $r(c_1,c),r(c_2,c) \in \Bmc$ with $c_1 \neq c_2$, then choose and remove one of these two assertions. Using the fact that every individual in \Bmc is reachable from $a_0$, it can be proved by induction on the number of edge removals that the obtained ABoxes
  \begin{enumerate}
  \item[(i)] remain connected and
  \item[(ii)] contain the same individuals as \Bmc, that is, edge removal never results in the removal of an individual.
  \end{enumerate}
  Point (i) and the fact that we start from a DAG-shaped ABox means that the ABox $\Bmc_t$ ultimately obtained by this manipulation is tree-shaped. By construction of $\Bmc_t$, $h$ is still a homomorphism from $\Bmc_t$ to \Amc, $\Bmc_t$ has pathwidth at most~$k$, and the individuals $b_1,\ldots,b_m$ are leaves in $\Bmc_t$ (and thus $\Bmc_t$ has at least $(n-2k-3)^{2k+3}$ leaves). From the former, it follows that the depth of $\Bmc_t$ is at most $n(i+1)$.
  
  Assume that $\Bmc_t$ does not contain the full binary tree of depth $2k+3$ as a minor.
  Then by Lemma~\ref{lem:numleaves}, the number of leaves of $\Bmc_t$ is
  at most
  $$(2k+3)(|\Tmc|-1)^{2k+2}(n(i+1))^{2k+2},
  $$
  which is a polynomial in $n$ of degree $2k+2$. So if we choose $n$ such that
  $$
  (n-2k-3)^{2k+3} > (2k+3)(|\Tmc|-1)^{2k+2}(n(i+1))^{2k+2}
  $$
  in the beginning, this leads to a contradiction.
  Hence, $\Bmc_t$ must contain as a minor the full binary tree of depth at least
  $2k+3$. But it is well-known that any such tree has pathwidth at least $k+1$, in contradiction to $\Bmc_t$
  having pathwidth at most $k$.
\end{proof}
We are now ready to establish the hierarchy.
\begin{proposition}
\label{lem:lowerbounddiameter}
$Q_{8\ell+13}$ is not rewritable into a linear Datalog program of width $\ell$.
\end{proposition}
%
%
\begin{proof}
  Assume to the contrary of what we have to show that $Q_{8\ell+13}$ is
  rewritable into a linear Datalog program $\Pi_0$ of width $\ell$. Let $k$ be
  the diameter of $\Pi_0$. Clearly, $\Pi_0$ is also a rewriting of  $Q_{8\ell+13}$ 
  on the class of ABoxes $\Cmf_k$. We show that $\Pi_0$ can be rewritten
  into a linear Datalog rewriting $\Pi'$ of $Q_{8\ell+13}$ of diameter $4 \ell+5$,  
  in contradiction to Lemma~\ref{lem:hierarchytwo}.

  We carry out a sequence of three rewriting steps. Informally, in the first
  rewriting we normalize the shape of rule bodies, in the second one we
  control the number of disconnected components in the rule body (or
  rather its restriction to the EDB relations), and in the third one we actually bound
  the diameter to $4 \ell + 5$. 

  In the first step, let $\Pi_1$
  be obtained from $\Pi_0$ by replacing every rule $S(\xbf) \leftarrow q(\ybf)$ 
  in $\Pi_0$ with the set of all rules $S(\xbf') \leftarrow q(\ybf')$ that can be obtained from the
  original rule by consistenly identifying variables in the rule body and head
  such that the restriction of $q(\ybf')$ to EDB relations (that is, concept and
  role names in $\Sigma$) takes the form of a forest in which every
  tree branches at most once. This step preserves equivalence on~$\Cmf_k$ since every homomorphism from
  the body of a rule in $\Pi$ 
into an ABox from~$\Cmf_k$ (and also to the extension of
 such an ABox with IDB relations) induces a variable identification that identifies a
 corresponding rule produced in the rewriting. 

 In the next step, we rewrite $\Pi_1$ into a linear Datalog program $\Pi_2$,
 as follows. Let $S(\xbf) \leftarrow q(\ybf)$ be a rule in $\Pi_1$ and call a variable in $q(\ybf)$
 \emph{special} if it occurs in \xbf or in the IDB atom in $q(\ybf)$, if existent. We obtain a new rule body $q'(\ybf')$ from $q(\ybf)$ in the following way:
 \begin{enumerate}

 \item remove the IDB atom (if existent), obtaining a forest-shaped rule body;

 \item remove all trees that do not contain a special variable;

 \item re-add the IDB atom (if existent).

 \end{enumerate}
 In $\Pi_2$, we replace $S(\xbf) \leftarrow q(\ybf)$ with $S(\xbf) \leftarrow q'(\ybf')$.

We argue that, on the class of ABoxes $\Cmf_k$, $\Pi_2$ is equivalent to $\Pi_1$.
Thus, let $\Amc$ be
an ABox from $\Cmf_k$ and $a \in \mn{ind}(\Amc)$ such that $\Amc \models \Pi_2(a)$.
We have to show that $\Amc \models \Pi_1(a)$. Let $q_1(\xbf_1),\dots,q_m(\xbf_m)$
be all trees that have been removed from a rule body during the construction of $\Pi_2$.
Let $\Amc_i$ be $q_i(\xbf_i)$ viewed as a $\Sigma$-ABox, $1 \leq i \leq m$.
Note that each $\Amc_i$ must be in $\Cmf_k$.  Let \Bmc be the disjoint union
of the ABoxes $\Amc,\Amc_1,\dots,\Amc_m$, assuming that these ABoxes
do not share any individual names, and note that $\Bmc$ is in $\Cmf_k$.
Since $\Amc \models \Pi_2(a)$, we must have $\Bmc \models \Pi_2(a)$.
By construction of $\Bmc$, this clearly implies $\Bmc \models \Pi_1(a)$.
Consequently, $\Bmc \models Q_{8\ell+13}(a)$.
Since answers to OMQs from $(\EL,\textnormal{AQ})$ depend only on the reachable part of ABoxes,
we obtain that $\Amc \models Q_{8\ell+13}(a)$, thus $\Amc \models \Pi_1(a)$ as required.

At this point, let us sum up the most important properties of the linear Datalog program~$\Pi_2$: it is a rewriting of $Q_{8\ell+13}$ on $\Cmf_k$, has width at most $\ell$ and diameter at most $k$, and 
\begin{enumerate}

\item[($*$)] the restriction of the rule body to EDB relations is a forest that consists
  of at most $2\ell$ trees.

\end{enumerate}
Note that the upper bound of $2\ell$ is a consequence of the fact that, by construction of $\Pi_2$, each of the relevant trees contains at least one special variable.

We now rewrite $\Pi_2$ into a final linear Datalog program $\Pi_3$ that is equivalent
to $\Pi_2$, has width at most $4\ell+2$, and diameter at most $4\ell+5$. Thus
$\Pi_3$ is a rewriting of $Q_{8\ell+13}$ on $\Cmf_k$ of diameter $4\ell+5$, which is a contradiction to Lemma~\ref{lem:hierarchytwo}. 

It thus remains to give the construction of $\Pi_3$.  Let $\rho = S(\xbf) \leftarrow q(\ybf)$ be a rule in $\Pi_2$ and let $\ybf' \subseteq \ybf$ be the set of variables $x$ that are special or a branching variable where the latter means that $q(\ybf)$ contains atoms of the form $r(x,y_1)$, $s(x,y_2)$ with $y_1 \neq y_2$. Due to~($*$), $\ybf'$ contains at most $4 \ell$ variables. Let $q'(\ybf')$ be the restriction of $q(\ybf)$ to the variables in $\ybf'$; we can assume that each variable $y$ from $\ybf'$ occurs in $q'(\ybf')$ since if this is not the case, we can add an atom $\top(y)$.  By construction of $\Pi_2$, it can be verified that $q(\ybf)$ is the union of $q'(\ybf')$ and path-shaped $q_1(\ybf_1),\dots,q_n(\ybf_n)$ such that for $1 \leq i \leq n$
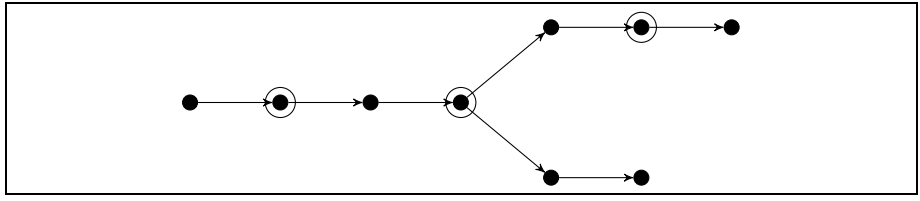
\begin{figure}
\begin{boxedminipage}[t]{\columnwidth}
\centering
\begin{tikzpicture}[->,>=stealth',level distance = 0.8cm,font=\sffamily\small]
\tikzstyle{node}=[shape=circle, draw,inner sep=2.0pt, fill=black]
\node [node] (A) at (0,0) {};
\node [node] (B) at (1.2,0) {};
\node [node] (C) at (2.4,0) {};
\node [node] (D) at (3.6,0) {};
\node [node] (E1) at (4.8,1) {};
\node [node] (F1) at (6,1) {};
\node [node] (G1) at (7.2,1) {};
\node [node] (E2) at (4.8,-1) {};
\node [node] (F2) at (6,-1) {};
\draw [->] (A) -- (B);
\draw [->] (B) -- (C);
\draw [->] (C) -- (D);
\draw [->] (D) -- (E1);
\draw [->] (E1) -- (F1);
\draw [->] (F1) -- (G1);
\draw [->] (D) -- (E2);
\draw [->] (E2) -- (F2);

\draw (B) circle (2mm);
\draw (D) circle (2mm);
\draw (F1) circle (2mm);
\end{tikzpicture}
\end{boxedminipage}
\caption{The body $q(\ybf)$ of a rule from $\Pi_2$ consists of one or several such trees, where at most one variable is branching. The branching variable and special variables are circled and they divide the body into five paths $q_i(\ybf_i)$.}
\label{fig:tree-to-path}
\vspace*{-3mm}
\end{figure}

\begin{itemize}

\item $q_i(\ybf_i)$ contains only EDB atoms,

\item each $q_i(\ybf_i)$ contains at most two variables from $\ybf'$ and
  each such variable is an end point of the path, and

\item the queries $q_1(\ybf_1),\dots,q_n(\ybf_n)$ only share variables from $\ybf'$.

\end{itemize}
The structure of $q(\ybf)$ is illustrated in Figure~\ref{fig:tree-to-path}.
We thus find  linear Datalog programs $\Gamma_1,\dots,\Gamma_n$ 
that are at most binary, of width at most two and diameter at most three
such that for any $\Sigma$-ABox \Amc and $\abf \subseteq \mn{ind}(\Amc)$,
$\Amc \models \Gamma_i(\abf)$ if and only if there is a homomorphism
$h_i$ from $q_i(\ybf_i)$ to \Amc such that $h_i(\ybf_i)=\abf$.
Let the goal relations of $\Gamma_1,\dots,\Gamma_n$ be $G_1,\dots,G_n$ and assume that $G_i$
occurs in $\Gamma_i$ only once, in a rule head $G_i(\xbf_i)$.
We assume w.l.o.g.\ that the programs $\Gamma_1,\dots,\Gamma_n$
do not share variables or IDB relations, and neither do they share variables
or IDB relations with $\Pi_2$.  In $\Pi_3$,
we replace $\rho = S(\xbf) \leftarrow q(\ybf)$ with the following rules:
\begin{itemize}
\item for any rule $P(\wbf) \leftarrow p(\zbf)$ in $\Gamma_1$ where $p(\zbf)$
  contains only EDB atoms, the rule $X^P_\rho(\ybf',\wbf) \leftarrow q'(\ybf') \wedge p(\zbf)$;

\item for any rule $P(\wbf) \leftarrow p(\zbf)$ in $\Gamma_i$, $1 \leq i \leq n$,
  where $p(\zbf)$ contains the IDB atom $R(\ubf)$, the rule
  $X_\rho^P(\ybf',\wbf) \leftarrow X^R_\rho(\ybf',\ubf) \wedge p(\zbf)$;

\item for any rule $P(\wbf) \leftarrow p(\zbf)$ in $\Gamma_i$, $1 < i \leq n$,
  where $p(\zbf)$ contains only EDB atoms, the rule
  $X^P_\rho(\ybf',\wbf) \leftarrow X^{G_{i-1}}_\rho(\ybf',\xbf_{i-1}) \wedge p(\zbf)$;

\item the rule $S(\xbf) \leftarrow X^{G_n}_\rho(\ybf',\xbf_n)$,
\end{itemize}
where the goal relations of $\Gamma_1,\dots,\Gamma_n$ become standard
(non-goal) IDB relations. It can be verified that $\Pi_3$ is as required.
\end{proof}

\section{Decidability and Complexity}
\label{sec:decidability}

We study the meta problems that emerge from the results in the previous sections such as deciding whether a given OMQ is in $\NL$, $\PTime$-hard, or rewritable into linear Datalog.  We show that all these problems are \ExpTime-complete. Apart
from applying and adapting known lower and upper bounds, the central ingredient
is giving a single exponential time decision procedure for deciding whether an OMQ from $(\EL,\textnormal{conCQ})$ has the ability to simulate \PSA. 
We start with lower bounds, which hold already for $(\EL, \textnormal{AQ})$.
\begin{theorem}
\label{thm:exptime-hard}
Given an OMQ $Q \in (\EL,\textnormal{AQ})$, the following problems are $\ExpTime$-hard:
\begin{enumerate}[(1)]
\item Is $Q$ FO-rewritable?
\item Is $Q$ rewritable into linear Datalog?
\item Is $\textsc{eval}(Q) \in \AC^0$?
\item Is $\textsc{eval}(Q) \in \NL$? (unless $\NL=\PTime$)
\item Is $\textsc{eval}(Q)$ $\NL$-hard?
\item Is $\textsc{eval}(Q)$ $\PTime$-hard?
\end{enumerate}
\end{theorem}
\begin{proof}
$\ExpTime$-hardness of (1) is proved in (the appendix of) \cite{BieLuWo-IJCAI13}. By our Theorem~\ref{thm:AC0NLdisc}, (1) and (3) are equivalent, so (3) is also $\ExpTime$-hard.

For (2), (4), (5) and (6), we analyse the mentioned hardness proof from \cite{BieLuWo-IJCAI13} a little closer. The proof is
  by a reduction from the word problem of a polynomially space bounded
  alternating Turing machine (ATM) $M$ that solves an \ExpTime-complete
  problem. The reduction exhibits a polynomial time algorithm that constructs, given an input $w$ to~$M$, an
  OMQ $Q=(\Tmc,\Sigma,
  B(x)) \in (\mathcal{EL}, \textnormal{AQ})$ such that $Q$ is not FO-rewritable if and only if
  $M$ accepts~$w$.
  A careful inspection of the construction of $Q$ and of the ``$\Leftarrow$'' part
  of the proof reveals that
  \begin{itemize}

  \item[($*$)] if $M$ accepts~$w$, then $Q$ is unboundedly branching,
    thus (by Lemma~\ref{lem:pathwidthbranching} and Theorem~\ref{thm:NLPTimedisc}) not linear Datalog rewritable, \PTime-hard, and not in
    \NL (unless $\NL = \PTime$) and (by Theorem~\ref{thm:AC0NLdisc} and since no \PTime-hard problem can be in AC$^0$) 
\NL-hard.

  \end{itemize}
  If $M$ does not accept $w$, then FO-rewritability of $Q$ implies
  that $Q$ is 
  \begin{itemize}

  \item in AC$^0$ and thus in \NL and neither \NL-hard nor \PTime-hard;

  \item linear Datalog rewritable (since every FO-rewritable OMQ from $(\EL,\textnormal{AQ})$
  is rewritable into a UCQ \cite{BieLuWo-IJCAI13}).

  \end{itemize}
  The stated hardness results for (2), (4), (5) and (6) follow.
%
%
\end{proof}
The following theorem summarizes the corresponding upper bounds.
\begin{theorem}
\label{thm:upperbounds}
Given $Q \in (\EL, \textnormal{CQ})$, the following properties can be decided in $\ExpTime$:
\begin{enumerate}[(1)]
\item Is $Q$ FO-rewritable?
\item Is $Q$ rewritable into linear Datalog?
\item Is $\textsc{eval}(Q) \in \AC^0$?
\item Is $\textsc{eval}(Q) \in \NL$? (unless $\NL=\PTime$)
\item Is $\textsc{eval}(Q)$ $\NL$-hard?
\item Is $\textsc{eval}(Q)$ $\PTime$-hard? (unless $\NL=\PTime$)
\end{enumerate}
\end{theorem}
In \cite{ijcai16}, it was shown that (1) is in $\ExpTime$. By Theorem~\ref{thm:AC0NLdisc}, the same algorithm decides (3) and (5). By Theorem~\ref{thm:NLPTimedisc}, the remaining (2), (4) and (6) come down to a single decision problem. We concentrate on deciding (6). We first argue that it suffices to decide (6) for OMQs from $(\EL,\text{conCQ})$, that is, to restrict our attention to connected CQs.  

Let $Q=(\Tmc, \Sigma, q) \in (\EL, \textnormal{CQ})$. To decide whether $\textsc{eval}(Q)$ is $\PTime$-hard (unless $\NL=\PTime$), we can first check whether $Q$ is empty. This can be done in $\ExpTime$ \cite{ijcai16}) and an empty OMQ is clearly not \PTime-hard. Otherwise, we make $Q$ non-redundant (see Section~\ref{sec:disconnected}) by exhaustively removing Boolean MCCs that cause non-redundancy. This can also be done in exponential time since containment of OMQs from $(\EL, \textnormal{CQ})$ is in \ExpTime \cite{ijcai16}. The resulting OMQ $Q'=(\Tmc, \Sigma, q')$ is equivalent to $Q$ and as seen in the proof of Theorem~\ref{thm:NLPTimedisc}, $\textsc{eval}(Q')$ is \PTime-hard if and only if there is an MCC $q'_i$ of $q'$ such that $(\Tmc, \Sigma, q'_i) \in (\EL, \textnormal{conCQ})$ is $\PTime$-hard.

It thus remains to show how (6) can be decided in \ExpTime for OMQs $Q \in (\EL,\text{conCQ})$. For such $Q$, it follows from Lemmas~\ref{lem:pathwidthbranching}, \ref{lem:bintree-psa}, \ref{lem:psa-ptimehard}, and \ref{lem:LDLog} and Theorem~\ref{thm:NLPTime} that (6) is equivalent to deciding whether $Q$ has the ability to simulate \PSA.  In the remainder of this section, we reduce the question whether a given OMQ $Q \in (\EL, \textnormal{conCQ})$ has the ability to simulate \PSA to the \mbox{(non-)emptiness} problem of two-way alternating parity tree automata (TWAPA), which is 
\ExpTime-complete. In fact, we construct a TWAPA that accepts precisely those (encodings of) pseudo tree-shaped ABoxes that witness the ability to simulate \PSA and then check non-emptiness. 

\smallskip
\noindent
\textbf{Two-way alternating parity tree automata (TWAPA).}  A \emph{tree} is a non-empty (and potentially infinite) set $T \subseteq
\mathbb{N}^*$ closed under prefixes. We say that $T$ is \emph{$m$-ary}
if $T \subseteq \{1,\dots,m\}^*$. For an alphabet $\Gamma$, a
\emph{$\Gamma$-labeled tree} is a pair $(T,L)$ with $T$ a tree and $L:T
\rightarrow \Gamma$ a node labeling function.

For any set $X$, let $\Bmc^+(X)$ denote the set of all positive Boolean formulas over
$X$, i.e., formulas built using conjunction and disjunction over the
elements of $X$ used as propositional variables, and where the special
formulas $\mn{true}$ and $\mn{false}$ are allowed as well.
An \emph{infinite path} $P$ of a tree $T$ is a
prefix-closed set $P \subseteq T$ such that for every $i \geq 0$,
there is a unique $x \in P$ with $|x|=i$.

\begin{definition}[TWAPA]
  A \emph{two-way alternating parity automaton 
    (TWAPA) on finite $m$-ary trees} is a tuple
  $\Amf=(S,\Gamma,\delta,s_0,c)$ where $S$ is a finite set of
  \emph{states}, $\Gamma$ is a finite alphabet, $\delta: S \times
  \Gamma \rightarrow \Bmc^+(\mn{tran}(\Amf))$ is the \emph{transition
    function} with $\mn{tran}(\Amf) = \{ \langle i \rangle s, \ [i] s
  \mid -1 \leq
  i \leq m \text{ and } s \in S \}$ the set of
  \emph{transitions} of \Amf, $s_0 \in S$ is the \emph{initial state},
  and $c:S \rightarrow \mathbb{N}$ is the \emph{parity condition} that 
  assigns to each state a \emph{priority}.
\end{definition}
Intuitively, a transition $\langle i \rangle s$ with $i>0$ means that
a copy of the automaton in state $s$ is sent to the $i$-th successor
of the current node, which is then required to exist. Similarly,
$\langle 0 \rangle s$ means that the automaton stays at the current
node and switches to state $s$, and $\langle -1 \rangle s$ indicates
moving to the predecessor of the current node, which is then required
to exist. Transitions $[i] s$ mean that a copy of the automaton in
state $s$ is sent on the relevant successor if that successor exists
(which is not required).
\begin{definition}[Run, Acceptance]
  Let $\Amf = (S,\Gamma,\delta,s_0,c)$ be a TWAPA and $(T,L)$ a finite $\Gamma$-labeled tree. A \emph{configuration} is a pair from $T \times S$. A \emph{run} of $\Amf$ on $(T,L)$ from the configuration $\gamma$ is a $T \times S$-labeled tree $(T_r,r)$ such that the following conditions are satisfied:
  \begin{enumerate}

  \item $r(\varepsilon) = \gamma$;
    
  \item if $y \in T_r$, $r(y)=(x,s)$, and $\delta(s,L(x))=\vp$, then
    there is a (possibly empty) set $S \subseteq \mn{tran}(\Amf)$ such
    that $S$ (viewed as a propositional valuation) satisfies $\vp$ as
    well as the following conditions:
    \begin{enumerate}

    \item if $\langle i \rangle s' \in S$, then $x \cdot i \in T$ and 
      there is a node $y \cdot j \in T_r$ such that $r(y \cdot j)=(x 
      \cdot i,s')$;

    \item if $[i]s' \in S$ and $x \cdot i \in T$, then
      there is a
      node $y \cdot j \in T_r$ such that $r(y \cdot j)=(x \cdot
      i,s')$.

    \end{enumerate}

  \end{enumerate}
  We say that $(T_r,r)$ is \emph{accepting} if on all infinite paths
  of $T_r$, the maximum priority that
  appears infinitely often on this path is even.  A finite $\Gamma$-labeled tree
  $(T,L)$ is \emph{accepted} by \Amf if there is an accepting run of
  \Amf on $(T,L)$ from the configuration $(\varepsilon, s_0)$. We use $L(\Amf)$ to denote the set of all finite
  $\Gamma$-labeled tree accepted by \Amf.
  \end{definition}
It is known (and easy to see) that TWAPAs are closed under
complementation and intersection, and that these constructions
involve only a polynomial blowup. In particular, complementation
boils down to dualizing the transitions and increasing all priorities by one.
%
It is also known that 
their emptiness problem can be solved in time single exponential in the number of states and highest occurring priority, and polynomial in all other components of the automaton \cite{Vardi}. 
In what follows, we shall generally only explicitly
analyze the number of states of a TWAPA, but only implicitly take care
that all other components are of the allowed size for the complexity
result that we aim to obtain.

\smallskip
\noindent
\textbf{Encoding pseudo tree-shaped ABoxes.} To check the ability to simulate \PSA using TWAPAs, we build one TWAPA $\Amf_{t_0,t_1}$ for every pair $(t_0,t_1)$ of $\Tmc$-types. An input tree for the TWAPA encodes a tuple $(\Amc, \abf, b, c, d)$ of a pseudo tree-shaped ABox $\Amc$ of core size at most $|q|$, a tuple $\abf$ from the core and three distinguished individuals $b$, $c$ and $d$. The TWAPA $\Amf_{t_0,t_1}$ should accept a tree that encodes $(\Amc, \abf, b, c, d)$ if and only if $t_0,t_1,\Amc, \abf, b, c$ and $d$ witness the ability to simulate \PSA according to Definition~\ref{def:psa}. The $\ExpTime$ decision procedure is obtained by checking whether at least one of the (exponentially many) $\Amf_{t_0, t_1}$ accepts a non-empty language. 

We encode tuples $(\Amc, \abf, b, c, d)$ as finite $(|\Tmc| \cdot |q|)$-ary
$\Gamma_\varepsilon \cup \Gamma_N$-labeled trees, where
$\Gamma_\varepsilon$ is the alphabet used for labeling the root node
and $\Gamma_N$ is for non-root nodes. These alphabets are different
because the root of a tree encodes the entire core of a pseudo
tree-shaped ABox whereas each non-root node represents a single
non-core individual.

Let $\Csf_{\mn{core}} \subseteq \NI$ be a fixed set of size $|q|$. Define $\Gamma_\varepsilon$ to be the set of all tuples $(\Bmc, \abf)$, where $\Bmc$ is a $\Sigma$-ABox over $\Csf_\mn{core}$ and $\abf$ a tuple of length $\mn{ar}(q)$ from $\mn{ind}(\Bmc)$. Let $\mn{ROL}$ be the set of roles that appear in $\Tmc$ or $\Sigma$ and let $\mn{CN}$ by the set of all concept names that appear in $\Tmc$ or $\Sigma$. Let $S = \mn{ROL} \cup \mn{CN} \cup \Csf_\mn{core} \cup \{b, c, d\}$. The alphabet $\Gamma_N$ is defined to be the set of all subsets of $S$ that contain exactly one element from $\mn{ROL}$, at most one element from $\Csf_\mn{core}$ and at most one element of $\{b, c, d\}$. We call a $(\Gamma_\varepsilon \cup \Gamma_N)$-labeled tree $(T,L)$ \emph{proper} if
\begin{itemize}
\item $L(\varepsilon) \in \Gamma_\varepsilon$ and $L(x) \in \Gamma_N$ for all $x \neq \varepsilon$,
\item $L(x)$ contains an element of $\Csf_\mn{core}$ if and only if $x$ is a child of $\varepsilon$,
\item there is exactly one node $x_b \in T$ with $b \in L(x_b)$, exactly one node $x_c \in T$ with $c \in L(x_c)$ and exactly one node $x_d \in T$ with $d \in L(x_d)$,
\item the nodes $x_c$ and $x_d$ are incomparable descendants of $x_b$,
\item the nodes $\varepsilon$, $x_b$, $x_c$ and $x_d$ have pairwise distance more than $|q|$ from each other.
\end{itemize}
A proper tree $(T,L)$ encodes a tuple $(\Amc, \abf, b,c,d)$ in the following way. If $L(\varepsilon) = (\Bmc, \abf)$, then
\[
\begin{array}{rcl}
   \Amc &=& \Bmc \cup \{ A(x) \mid A \in L(x), x \neq \varepsilon \}\\[0.5mm]
   && \cup\; \{r(a,x) \mid \{a,r\} \subseteq L(x) \text{ with } a \in \Csf_\mn{core}\} \\[0.5mm]
   && \cup\; \{ r(x,y) \mid r \in L(y), y \text{ is a child of } x, x \neq \varepsilon\}
\end{array}
\]
with $x_b$ replaced with $b$, $x_c$ with $c$, and $x_d$ with $d$. It is easy to see that there is a TWAPA $\Amf_\mn{proper}$ that accepts a $(\Gamma_\varepsilon \cup \Gamma_N)$-labeled tree if and only if it is proper.

From now on, let $t_0$ and $t_1$ be fixed. We construct the TWAPA $\Amf_{t_0,t_1}$ as the intersection of $\Amf_\mn{proper}$ and TWAPAs $\Amf_1, \ldots, \Amf_6$ where each $\Amc_i$ accepts a proper input tree $(T,L)$ if and only if the tuple
$(\Amc, \abf, b,c,d)$ encoded by $(T,L)$ satisfies Condition~($i$) from Definition~\ref{def:psa}. We make sure that all $\Amf_i$ can be constructed in exponential time and have only polynomially many states in the size of $Q$.

\smallskip
\noindent
\textbf{Derivation of concept names.} Before describing any of the $\Amf_i$ in detail, we describe one capability of TWAPAs that most of the $\Amf_i$ will make use of, namely to check whether a certain concept name is derived at a certain individual. We thus construct a TWAPA $\Amf_\mn{derive}$ with states $S_\mn{derive} =$
$$\{d_A \mid A \in \mn{CN}\} \cup \{d_A^a \mid A \in \mn{CN} \wedge a \in \Csf_\mn{core}\} \cup \{d_r \mid r \in \mn{ROL}\} \cup \{d_a \mid a \in \Csf_\mn{core}\}$$ such that
\begin{itemize}
\item if $\Amf_\mn{derive}$ is started on a proper input tree encoding $(\Amc, \abf, b, c, d)$ from a configuration $(a, d_A)$, then it accepts if and only if $\Tmc, \Amc \models A(a)$;
\item if $\Amf_\mn{derive}$ is started on a proper input tree encoding $(\Amc, \abf, b, c, d)$ from a configuration $(d_A^a,\varepsilon)$, then it accepts if and only if $\Tmc, \Amc \models A(a)$.
\end{itemize}

By Lemma~\ref{lem:AQderivation}, $\Tmc, \Amc \models A(a)$ if and only if there is a derivation tree for $A(a)$. We give the straightforward construction of $\Amf_\mn{derive}$, that checks for the existence of a derivation tree of $A(a)$. 
For brevity, let $\ell=|\Tmc| + |q|$. Let $\sigma \in \Gamma_N$  not contain an element of $\Csf_\mn{core}$ and let $r$ be the unique role name in $\sigma$. If $A \in \sigma$ or $\top \sqsubseteq A \in \Tmc$, we set $\delta(d_A,\sigma) = \mn{true}$. Otherwise, set
\begin{align*}
\delta(d_A,\sigma) = &\Big( \bigvee_{\Tmc \models A_1 \sqcap \ldots \sqcap A_n \sqsubseteq A} \; \bigwedge_{i=1}^n \langle 0 \rangle d_{A_i} \Big) \vee \Big( \bigvee_{\exists s.B \sqsubseteq A \in \Tmc} \; \bigvee_{i=1}^\ell \langle i \rangle (d_B \wedge d_s) \Big)\\
&\vee \Big( \bigvee_{\exists r^-.B \sqsubseteq A} \langle -1 \rangle d_B \Big)
\end{align*}
Now let $\sigma \in \Gamma_N$ contain $a \in \Csf_\mn{core}$ and let again
$r$ be the unique role name in $\sigma$. If $A \in \sigma$ or $\top \sqsubseteq A \in \Tmc$, we set $\delta(d_A,\sigma) = \mn{true}$. Otherwise, set
\begin{align*}
\delta(d_A,\sigma) = &\Big( \bigvee_{\Tmc \models A_1 \sqcap \ldots \sqcap A_n \sqsubseteq A} \; \bigwedge_{i=1}^n \langle 0 \rangle d_{A_i} \Big) \vee \Big( \bigvee_{\exists s.B \sqsubseteq A \in \Tmc} \; \bigvee_{i=1}^\ell \langle i \rangle (d_B \wedge d_s) \Big)\\
&\vee \Big( \bigvee_{\exists r^-.B \sqsubseteq A} \langle -1 \rangle d_B^a \Big)
\end{align*}
Next, let $\sigma=(\Bmc,\abf) \in \Gamma_\varepsilon$. If $A(a) \in \Bmc$ or $\top \sqsubseteq A \in \Tmc$, we set $\delta(d_A^a,\sigma) = \mn{true}$. Otherwise, set
\begin{align*}
\delta(d_A^a,\sigma) = &\Big( \bigvee_{\Tmc \models A_1 \sqcap \ldots \sqcap A_n \sqsubseteq A} \; \bigwedge_{i=1}^n \langle 0 \rangle d_{A_i}^a \Big) \vee \Big( \bigvee_{\exists s.B \sqsubseteq A \in \Tmc} \; \bigvee_{i=1}^\ell \langle i \rangle (d_B \wedge d_s \wedge d_a) \Big)\\
 &\vee \Big( \bigvee_{\exists s.B \sqsubseteq A, s(a,b) \in \Bmc} \langle 0 \rangle d_B^b \Big) \vee \Big( \bigvee_{\exists s^-.B \sqsubseteq A, s(b,a) \in \Bmc} \langle 0 \rangle d_B^b \Big)
\end{align*}

Finally, let $\sigma \in \Gamma_N$ and $a \in \Csf_\mn{core}$. Set $\delta(d_a,\sigma) = \mn{true}$ if $a \in \sigma$ and $\delta(d_a,\sigma) = \mn{false}$ if $a \notin \sigma$.

\smallskip
\noindent
\textbf{Construction of $\Amf_1$.} This TWAPA checks whether $\Amc \models Q(\abf)$. 
Since by Condition~5 of the ability to simulate \PSA, every homomorphism from $q$ to $\Umc_{\Amc, \Tmc}$ is core close, we only need to check whether $\Amc \models Q(\abf)$ via a core close homomorphism. The existence of such a homomorphism, in turn, depends only on the core of \Amc and on which concept names are derived at core individuals. If, in fact, $h$ is a core close homomorphism from $q$ to $\Umc_{\Amc, \Tmc}$, then $h$ hits at least one core individual or an element in an anonymous subtree below a core individual. Of course, $h$ might also hit individuals and anonymous elements outside the core. However, outside the core $\Umc_{\Amc, \Tmc}$ is tree-shaped. Take any $z \in \mn{var}(q)$ such that $h(z)$ lies outside of the core. Then there is a unique $a \in \Csf_\mn{core}$ such that $h(z)$ lies in a tree below $a$. Since $q$ is connected, there is an atom $r(x,y) \in q$ such that $h(x) = a$, $h(y)$ lies in the tree below $a$, and $z \in \mn{reach}(x,y)$. But then $\Cmc_q$ contains $q|_{\mn{reach}(x,y)}$ viewed as an \EL-concept $C$, $\Tmc \models C \sqsubseteq A_C$, and $A_C(a) \in \Umc_{\Amc, \Tmc}$. Thus, the concept names $A_C$ derived at core individuals completely represent the restriction of $h$ to the variables in $q$ that are mapped to outside the core.

Let $\Amc$ be a pseudo tree-shaped $\Sigma$-ABox with core $\Bmc$, $\mn{ind}(\Bmc) \subseteq \Csf_\mn{core}$, and $\abf$ a tuple from $\Csf_\mn{core}$. If $\Amc \models Q(\abf)$, then the set $\{A(a) \mid a \in \Csf_\mn{core} \text{ and } a \in A^{\Umc_{\Amc, \Tmc}}\}$ is a \emph{completion for $(\Bmc, \abf)$}. Let $\mn{Comp}{(\Bmc, \abf)}$ be the set of all completions for $(\Bmc, \abf)$. Using the arguments above, one can show the following.
\begin{lemma}
\label{lem:matchcompletion}
 Let $\Amc$ be a pseudo tree-shaped $\Sigma$-ABox with core $\Bmc$, $\mn{ind}(\Bmc) \subseteq \Csf_\mn{core}$, and $\abf$ a tuple from $\Csf_\mn{core}$. Then, 
$\Amc \models Q(\abf)$ if and only if there is an  $M \in \mn{Comp}{(\Bmc, \abf)}$ such that $\Amc \models A(a)$ for all $A(a) \in M$.
\end{lemma}
Each set $\mn{Comp}{(\Bmc, \abf)}$ has at most exponentially many completions
and can be computed in exponential time. Moreover, there are only exponentially
many choices for $(\Bmc,\abf)$. The strategy of $\Amf_1=(S_\mn{derive} \cup \{s_0\},\Gamma_\varepsilon \cup \Gamma_N,\delta,s_0,c)$ is as follows: Guess a match completion set $M$ for $Q$ regarding $(\Bmc, \abf)$, where $(\Bmc, \abf) \in \Gamma_\varepsilon$ is the label of the root of the input tree and then check whether all $A(a) \in M$ can be derived in $\Amc$. The parity condition $c$ assigns $1$ to every state, so precisely the finite runs are accepting. The transition function $\delta$ for states in $S_\mn{derive}$ is defined as before, and additionally we set

\begin{align*}
\delta(s_0,(\Bmc,\abf)) = \bigvee_{M \in \mn{Comp}_Q^{(\Bmc,\abf)}} \bigwedge_{A(a) \in M} d_A^a.
\end{align*}

\smallskip
\noindent
\textbf{Construction of $\Amf_2$.} This TWAPA checks that $t_1 = \mn{tp}_{\Amc, \Tmc} (b) = \mn{tp}_{\Amc, \Tmc} (c) = \mn{tp}_{\Amc, \Tmc} (d)$. Using $\Amf_\mn{derive}$
and its dualization, this is straightforward: send a copy to the nodes in the input
tree that are marked with $b$, $c$, and $d$, and then use $\Amf_\mn{derive}$ to
make sure that all $A \in t_1$ are derived there and the dual of $\Amf_\mn{derive}$
to make sure that no $A \notin t_1$ is derived there.

\smallskip
\noindent
\textbf{Construction of $\Amf_3$.} This TWAPA checks that $\Amc_b \cup t_0(b), \Tmc \not\models q(\abf)$. It is constructed in the same way as $\Amf_1$, but using $\Amc_b \cup t_0(b)$ instead of $\Amc$ for defining completions and
modifying $\Amf_\mn{derive}$ to that it assumes all concept names from $t_0$
to be true at the node of the input tree marked with $b$ and disregards the
subtree below. Also, we complement the constructed TWAPA at the end.

\smallskip
\noindent
\textbf{Construction of $\Amf_4$.} Similar to $\Amc_3$.


\smallskip
\noindent
\textbf{Construction of $\Amf_5$.} This TWAPA checks that every homomorphism from $q$ to $\Umc_{\Amc, \Tmc}$ is core close and this condition is only required if $q$ is Boolean. The condition is always true when $q$ is not treeifiable, since every homomorphism from a query that is not treeifiable into $\Umc_{\Amc, \Tmc}$ hits the core. Thus, if $Q$ is not Boolean or not treeifiable, we define $\Amf_5$ to be the TWAPA that accepts every input. If $Q$ is Boolean and treeifiable, we define a TWAPA $\Amf_5'$ that checks the negation of Condition~5 and then define $\Amf_5$ to be the complement of $\Amf_5'$. Let $C_q$ be the $\EL$-concept that corresponds to $q^\mn{tree}$. The TWAPA $\Amf_5'$ has to check whether $C_q$ is derived at any non-core individual $a$ or at an anonymous individual below a non-core individual. To check whether $C_q$ gets derived at an anonymous individual, define $M_Q$ be the set of all $\Tmc$-types $t$ with $t(a), \Tmc \models \exists x \, C_q(x)$. The set $M_Q$ can be computed in exponential time. Now $\Amf_5'$ guesses a non-core individual $a$ and $t \in M_Q$ and checks that $\mn{tp}_{\Amc, \Tmc}(a) =t$ using $\Amf_\mn{derive}$ and its dualization.

\smallskip
\noindent
\textbf{Construction of $\Amf_6$.} Condition~6 is only required when $q$ is Boolean. 
If $q$ is Boolean, this TWAPA checks that $b$, $c$ and $d$ all have the same ancestor path up to length $|q|$. The idea is to guess an ancestor path $r_1 r_2 \ldots r_{|q|}$ up front and then verify that $b$, $c$ and $d$ all have this ancestor path. To achieve this using only polynomially many states, the guessed path is not stored in a single state. Instead, we use $|q|$ copies of the automaton, the $i$-th copy guessing
states of the form $s_{i,r}$ which stands for $r_i = r$. This copy then further
spawns into three copies that visit the nodes labeled $b$, $c$, and $d$, travels 
upwards from there $n-i$ steps, and checks that the node label there contains $r$.
%
%

\section{Conclusion}

We have established a complexity trichotomy between AC$^0$, \NL, and \PTime for OMQs from $(\EL,\text{CQ})$. We have also proved that linear Datalog rewritability coincides with OMQ evaluation in \NL and that deciding all these (and related) 
properties is \ExpTime complete with the lower bounds applying already to
$(\EL,\text{AQ})$.

There are several natural directions in which our results can be generalized. One is to transitions from CQs to unions of CQs (UCQs), that is, to consider the OMQ language $(\EL, \textnormal{UCQ})$. We conjecture that this generalization is not difficult and
can be achieved by replacing CQs with UCQs in all of our proofs; where we work 
with connected CQs, one would then work with UCQs in which every CQ is connected.
In fact, we only refrained from doing so since it makes all proofs more technical and distracts from the main ideas.

An important direction for future work is to extend our analysis to \ELI, that is,
to add inverse roles. Even the case of  $(\ELI,\textnormal{AQ})$ appears to be
challenging. In fact, it can be seen that a complexity classification of $(\ELI,\textnormal{AQ})$ is equivalent to a complexity classification of all
CSPs that have tree obstructions. In the following, we elaborate on this
extension.

With inverse roles, there are OMQs $(\Tmc,\Sigma,A(x)$) that are \LogSpace-complete such as when setting $\Tmc = \{ \exists r . A \sqsubseteq A,\ \exists r^- . A \sqsubseteq A\}$ and $\Sigma = \{ r,A\}$. Using a variation of the techniques from Section~\ref{sec:AC0NL} and the technique of \emph{transfer sequences} from \cite{ijcai16}, it should not be too hard to establish a dichotomy between AC$^0$ and \LogSpace in $(\ELI,\text{CQ})$. We conjecture that $\AC^0$, \LogSpace, \NL, and \PTime are the only complexities that occur. We also conjecture that \LogSpace-completeness coincides with rewritability into symmetric Datalog \cite{EgLaTe07}. 

Lifting our dichotomy between \NL and \PTime to $(\ELI,\text{AQ})$ is non-trivial. In fact, we give below an example which shows that unbounded branching no longer coincides with bounded pathwidth and thus our proof strategy, which uses unbounded branching in central places, has to be revised. Moreover, it seems
difficult to approach the dichotomy between \LogSpace and \NL without first
solving the \NL versus \PTime case. In this context, it is interesting to point
out that for CSPs, the following conditional result is known \cite{Kazda15}: if
rewritability into linear Datalog coincides with \NL, then rewritability into
symmetric Datalog coincides with \LogSpace.

\begin{example}
\label{exa:ELI}
Consider the OMQ $Q=(\Tmc, \Sigma, A(x)) \in (\ELI, \textnormal{AQ})$ with $\Sigma =  \{r, s, B, E, L\}$ and
\begin{align*}
\Tmc = \{&B \sqsubseteq M_1, \exists s.M_1 \sqsubseteq M_1, \exists r M_1 \sqsubseteq M_1', \exists s^- M_1' \sqsubseteq M_2,\\
&\exists r^- M_2 \sqsubseteq M_2, M_2 \sqcap L \sqsubseteq M_1, M_2 \sqcap E \sqsubseteq A, \exists s.A \sqsubseteq A\}\,.
\end{align*}
$Q$ is unboundedly branching, as witnessed by the ABox in Figure~\ref{fig:ELI} and generalizations thereof to arbitrary depth. A derivation of the query starts at $B$, the beginning marker, then it uses markers $M_1$ and $M_2$ to visit all the leafs from left to right in sequence, until it reaches $E$, the end marker, to derive the queried concept name $A$.

At the same time, $Q$ is rewritable into linear Datalog and thus in $\NL$, showing that unbounded branching and $\PTime$-hardness no longer coincide.
\end{example}

\begin{figure}
\label{fig:ELI}
\begin{boxedminipage}[t]{\columnwidth}
\centering
\begin{tikzpicture}[->,>=stealth',
level 1/.style={sibling distance=18em},
level 2/.style={sibling distance=9em},
level 3/.style={sibling distance=5em},
level 4/.style={sibling distance=2.5em},
level 5/.style={sibling distance=0.7em},
level distance = 0.8cm,
font=\sffamily\small]
\tikzstyle{individual}=[shape=circle, draw,inner sep=2.0pt, fill=black]

\node [individual] {}
    child{ node [individual] {}
        child{ node [individual] {} 
            child{ node [individual] [label=below:$B$] {} edge from parent node[above left] {$r$}}
            child{ node [individual] [label=below:$L$] {} edge from parent node[above right] {$s$}}
            edge from parent node[above left] {$r$}}
        child{ node [individual] {}
            child{ node [individual] [label=below:$L$] {} edge from parent node[above left] {$r$}}
            child{ node [individual] [label=below:$L$] {} edge from parent node[above right] {$s$}}
            edge from parent node[above right] {$s$}}
        edge from parent node[above left] {$r$}}
    child{ node [individual] {} 
        child{ node [individual] {} 
            child{ node [individual] [label=below:$L$] {} edge from parent node[above left] {$r$}}
            child{ node [individual] [label=below:$L$] {} edge from parent node[above right] {$s$}}
            edge from parent node[above left] {$r$}}
        child{ node [individual] {}
            child{ node [individual] [label=below:$L$] {} edge from parent node[above left] {$r$}}
            child{ node [individual] [label=below:$E$] {} edge from parent node[above right] {$s$}}
            edge from parent node[above right] {$s$}}
        edge from parent node[above right] {$s$}}
; 
\end{tikzpicture}
\end{boxedminipage}
\caption{An ABox $\Amc$ with $\Amc \models Q(a)$, where $a$ is the root of $\Amc$ and $Q$ is the OMQ from Example~\ref{exa:ELI}.}
\end{figure}
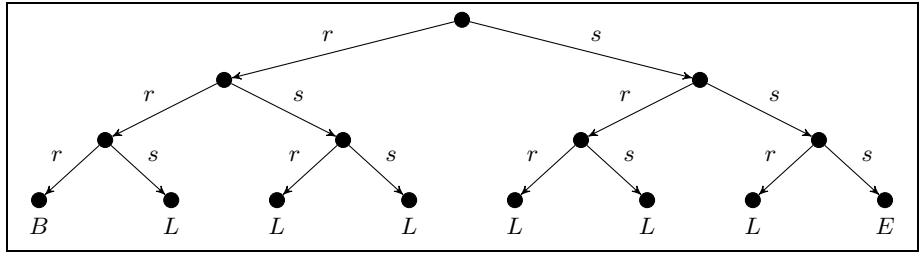

\noindent {\bf Acknowledgements}. This research was supported by ERC consolidator grant 647289 CODA.

\end{document}